\documentclass{article}
\usepackage{hyperref}
\usepackage{etoc}
\usepackage{authblk}
\usepackage{tikz}
\usetikzlibrary{shapes.geometric, arrows.meta, positioning, fit, backgrounds, calc, patterns}
\usepackage[utf8]{inputenc} 
\usepackage{graphicx}
\usepackage{algorithm}
\usepackage[noend]{algpseudocode}
\usepackage{amsmath}        
\usepackage{amsthm}
\usepackage{amsfonts}       
\usepackage{amssymb}        
\usepackage{geometry}       
\usepackage{comment}
\usepackage{xcolor}
\usepackage{booktabs}
\usepackage[toc,page,header]{appendix}
\usepackage{minitoc}

\usepackage[normalem]{ulem}

\newtheorem{theorem}{Theorem}[section]
\newtheorem{corollary}[theorem]{Corollary}
\newtheorem{lemma}[theorem]{Lemma}
\newtheorem{proposition}[theorem]{Proposition}

\newtheorem{remark}{Remark}

\newtheorem{assumption}{Assumption}
\newcommand{\toD}{\overset{\mathrm{d}}{\longrightarrow}}
\newcommand{\toP}{\overset{\mathrm{p}}{\longrightarrow}}

\newcommand{\method}{OPAL}
\newcommand{\bbeta}{\boldsymbol{\beta}}

\newcommand{\bmu}{\boldsymbol{\mu}}
\newcommand{\indep}{\perp \!\!\! \perp}

\newcommand{\bbP}{\mathbb{P}}
\newcommand{\cf}{\mathrm{cf}}

\newcommand{\B}{\mathcal{B}}

\newcommand{\E}{\mathbb{E}}

\newcommand{\G}{\mathbb{G}}

\newcommand{\R}{\mathbb{R}}

\newcommand{\N}{\mathcal{N}}

\newcommand{\Var}{\mathrm{Var}}
\newcommand{\Cov}{\mathrm{Cov}}
\renewcommand{\P}{\mathbb{P}}
\newcommand{\neff}{n_{\mathrm{effective}}}
\newcommand{\Bern}{\mathrm{Bern}}
\renewcommand{\N}{\mathcal{N}}

\newcommand{\Unif}{\mathrm{Unif}}

\newcommand{\budget}{n_{\mathrm{budget}}}

\renewcommand{\iff}{\Leftrightarrow}

\title{Optimized Labeling Resource Allocation\\
for Prediction-Assisted Inference via OPAL}
\author[1]{Virginia L. Ma}
\author[1,2]{Emmanuel J. Cand\`{e}s}

\affil[1]{Department of Statistics, Stanford University}
\affil[2]{Department of Mathematics, Stanford University}
\date{} 

\usepackage{biblatex} 
\begin{document}
\doparttoc
\faketableofcontents

\maketitle 

\begin{abstract}
Active Statistical Inference is a new framework to make precise claims about population parameters with provable statistical guarantees. It uses 
a predictive "black-box" machine learning (ML) model to strategically decide which data points to label, roughly prioritizing samples for which the ML model is unsure about their label values. A major issue is that the framework can be brittle when uncertainty estimates are noisy.
This paper introduces OPAL (Optimized Policy for Allocation of Labels), which learns a labeling strategy within a tractable class of smooth policies to yield estimators with the lowest variance. In effect, OPAL is an end-to-end pipeline that turns a black-box model's uncertainty scores into a data-adaptive labeling strategy and then performs inference on the collected samples. We evaluate OPAL on real datasets spanning medical imaging data, computational social science, and proteomics. {As a concrete example, we consider predicting breast cancer subtype from histopathology images and using OPAL to form valid confidence intervals for odds ratios for different demographic groups.} We show that OPAL achieves nominal coverage in finite samples and has the accuracy one expects from methods which have far more labeled samples. 
\end{abstract}

\section{Introduction}

\subsection{Motivation} \label{sec:motivation}
 While modern data pipelines readily accumulate vast quantities of complex features, obtaining verified labels remains a costly bottleneck. In medical imaging, for instance, acquiring chest X-rays is routine and inexpensive, yet identifying specific pathologies requires time-intensive review by expert radiologists \cite{CheXpert-paper, CheXpert, annotation-med-image}. {Similarly, in drug discovery, candidate properties predicted in silico such as binding affinity often require experimental validation. Because assays are costly and results arrive sequentially, it is natural to allocate validation effort adaptively as early measurements refine predictive models and reveal which regions of protein space are most informative \cite{reker2015active, reker2019practical, negoescu2011knowledge, warmuth2001active, warmuth2003active}.} 
 Off-the-shelf machine-learning models offer a tempting substitute via cheap provisional labels, but their errors can be unevenly distributed, e.g.~across patient subgroups and clinical contexts. This matters not only for prediction, but also for downstream scientific investigations: when we use model outputs as stand-ins for labels, systematic misclassification can translate into biased estimates of clinically meaningful quantities. This raises a resource-allocation problem: given a pool of unlabeled features and a finite budget for expert review, how should we select which samples to label so that inference about the target parameter remains valid and precise? 
 
Consider a setting where this issue is especially acute: estimating disease odds ratios across disparately represented subpopulations, a common way to quantify disparities and heterogeneity. In breast cancer, subtype identification guides treatment and prognosis, and subgroup-specific prevalence estimates help characterize disparities. Determining a sample subtype typically requires expert review of MRIs or histological slides---time-intensive work by radiologists and pathologists \cite{CNN3}. AI-based classifiers are compelling alternatives, but for fine-grained tasks like subtype diagnosis, error rates can vary across subgroups. When the estimand of interest is an odds ratio comparing two demographic strata, even modest differential misclassification can substantially distort the estimate, and these effects may be difficult to detect from aggregate accuracy metrics alone. A parallel example arises in chest X-ray imaging: in estimating the odds ratio of cardiomegaly (enlarged heart) among patients above age 40 versus those younger than 40, report-derived or model-predicted labels are inexpensive, but age-dependent misclassification (even with high overall accuracy) can bias the odds ratio unless expensive radiologist reads are allocated strategically.

A natural response is to label adaptively: rely on predictions where they are reliable, and prioritize expert effort where errors are most consequential for the estimand. Recent lines of work formalize this idea. In the semi-supervised regime where labels are obtained uniformly at random, prediction-powered inference (PPI) \cite{ppi} and follow-up work---PPI++ \cite{ppi++}, cPPI \cite{cppi}, and rePPI \cite{reppi}---guarantee valid inference by combining a black-box predictor with a modest set of ground-truth labels; related perspectives appear in \cite{witten}. Active statistical inference goes further by optimizing which points to label to reduce estimator variance within a broad M-estimation framework \cite{active, robust_active}. For mean estimation in large language model evaluation, a closely related approach prescribes a deterministic ``label-if-too-uncertain'' rule derived to optimize asymptotic variance \cite{active_LLMeval}. Parallel literatures on bandits and policy learning similarly optimize actions under budget constraints, though they typically target welfare or the average treatment effect rather than general plug-in estimands \cite{bandit, PL, TE}. For a more detailed literature review, see Appendix~\ref{app:lit}.

Existing active-labeling principles are most directly suited to target single mean or estimating-equation parameter estimands. The odds ratio example above illustrates why this is neither sufficient nor variance-minimizing. The log-odds ratio is a nonlinear function of two subgroup means, so the relevant objective might be the variance of a nonlinear function of sample averages, not the variance of averages in isolation. Treating the two mean-estimation problems separately can therefore miss the allocation that is optimal for the odds ratio itself. The same issue arises more generally for smooth functions of several component estimands. Beyond this, rank-based quantities such as Kendall's or Spearman's correlation are not directly covered by procedures designed around mean or standard M-estimation targets. OPAL addresses these cases by deriving the estimand-specific variance contribution and optimizing the labeling policy directly for the target estimand.

\subsection{Problem setup} \label{sec:setup}

We formalize the labeling setting next so we can state the optimization objective and inference guarantees.
We observe unlabeled features $X_1, \dots, X_n$ drawn i.i.d.~from some distribution $P_X$. For each observation we also record a binary group indicator $Z_1, \dots, Z_n \in \{0,1\}^n$, drawn i.i.d.~from a distribution $P_Z$ parametrized by $\P(Z_i = 1) = p$. For notational convenience, we treat $Z_i$ as one coordinate of $X_i$, so that the covariate vector includes the group membership indicator. Conditional on the features, the (unobserved) labels $Y_1, \dots, Y_n$ are generated according to $P_{Y \mid X}$.

In addition to the unlabeled data, we assume access to a fixed predictive model $f$ that maps features $X$ to predicted labels $\hat{Y}$, yielding predictions $f(X_1), \dots, f(X_n)$. We do not impose any assumptions on the form of $f$ or on the quality of its predictions; the model is used as a tool to guide our data-collection strategy and to allow for debiasing when imputed labels $\hat{Y}_i = f(X_i)$ are imperfect. For now, we assume that $f$ is a pre-trained black-box model, which was trained independently of our data sample. We provide a discussion on how to obtain $f$ when a pre-trained one is not available in Appendix~\ref{app:train_model}.

Obtaining labels is costly, so we assume a budget that allows us to label only a fraction $\rho$ of the data points, corresponding to an expected labeling budget of $\budget = n \rho$. Rather than sampling points uniformly at random, we aim to use the predictive model $f$ to identify a subset of informative points to label, thereby making more efficient use of this budget.

To capture heterogeneity across subpopulations, we single out the covariate $Z$ from the full vector $X$, emphasizing its role in defining the group structure of the data. We consider target parameters of the form $\theta = g(\theta_1, \theta_0)$, where each $\theta_j$ is a function of the subgroup distribution $P_{Y|X,Z = k} \times P_{X|Z = k}$.

When all observations belong to a single group (so that $Z_1 = \cdots = Z_n = j$), there is only one subgroup and $\theta \equiv \theta_j$ for all $j$, reducing to the one-group active statistical inference setting studied in \cite{active}. A canonical example involving subgroups is the average treatment effect $\tau = \E[Y \mid Z = 1] - \E[Y \mid Z = 0]$, obtained by taking $Z = 1$ as treatment and $Z = 0$ as control. In Section \ref{sec:general} we further generalize this framework to a parameter vector $\boldsymbol{\psi} \in \R^k$ and corresponding estimand $\theta = g(\boldsymbol{\psi})$.

Returning to the running example of diagnosing cardiomegaly from chest X-rays, $X_i$ consists of X-ray images and demographic features, and $Z_i$ is a binary indicator of whether a patient is at least 40 years old. The response $Y_i$ indicates the presence of cardiomegaly. Our inferential target is the odds ratio
$$\theta = \frac{\mu_1/(1 - \mu_1)}{\mu_0/(1 - \mu_0)},$$
where $\mu_1 = \P(Y = 1 \mid Z = 1)$ and $\mu_0 = \P(Y = 1 \mid Z = 0)$. Here, both $\mu_1$ and $\mu_0$ are functionals of the joint distribution of  $(Y, X, Z)$. For $X_i$ consisting of imaging data, a suitable classifier $f$ for predicting cardiomegaly may be a convolutional neural network such as ResNet fine-tuned on chest X-ray images \cite{CNN1, CNN2, CNN3}.

The primary goals of this work are (1) to construct confidence intervals for the target $\theta$; and (2) to design an optimized, data-driven labeling policy $\pi: \mathcal{X} \to [0,1]$ that decides which $\budget$ points to label, given the unlabeled features $X_i$ and the predictions of the black-box model $f$. 

The labeling policy $\pi$ is a \textit{soft} policy: it specifies labeling probabilities rather than hard assignments. For each data point $i$, we draw an indicator $\xi_i \sim \Bern(\pi(X_i))$. If $\xi_i = 1$, we observe $Y_i$; if $\xi_i = 0$, we do not query the label and instead use the prediction $\hat{Y}_i = f(X_i)$. Given a labeling budget of $k$ points, the policy $\pi$ guarantees labeling of $k$ points \textit{on average}, but not necessarily exactly $k$ points.

Like \cite{active}, we are interested in two labeling regimes: (1) a \textit{batch} setting, in which we have a pre-trained model $f$ and must decide all labels to query at once; and (2) a \textit{sequential} setting, in which we choose labels to query in one-at-a-time or in batches (possibly updating $f$ between batches), and update the policy $\pi$ as more data are collected.

\subsection{Our contributions} \label{sec:contributions}

We introduce Optimized Policy for Allocating Labels (\method), a framework for variance-optimal label collection and inference. First, we define a labeling policy $\pi_\theta(X)$, parametrized by $\theta$, which maps features $X$ to a labeling probability. The policy is constructed in a data-dependent manner using all observed features $X_1,\dots,X_n$. For ease of exposition, we consider a class of mean-type estimators, with target 
$$\mu_k = \mathbb{E}[s_k(X_i, Y_i)]$$ for some function $s_k(\cdot)$. For example, in subgroup mean estimation with groups $G_1, \dots, G_K$, at the population level, we take $s_k(X, Y) = Y1\{X \in G_k\} / \P(X \in G_k)$.  

We use an augmented inverse propensity weighted-style estimator, as in the active statistical inference literature on mean estimation \cite{active}, and make its dependence on the labeling policy explicit: 
\begin{equation}\label{eq:meanest}
    \hat{\mu}_k^\pi = \frac{1}{n}\sum_{i=1}^n \bigg[f_k(X_i) + \frac{\xi_i}{\pi(X_i)}(s_k(X_i, Y_i) - f_k(X_i))\bigg],  
\end{equation}
where $f_k(X)$ is a black box predictor for $s_k(X, Y)$; e.g. we may take $f_k(X) = s_k(X, \hat{Y})$ with $\hat{Y} = f(X)$. {Recall $\xi$ is the indicator of whether label $Y$ is acquired, so for unlabeled data points, the contribution of the second term  vanishes.} For subgroup means, with $n_k$ denoting the size of group $G_k$, the corresponding estimator is 
\begin{equation}\label{eq:subgpmean}
    \hat{\mu}_k^\pi = \frac{1}{n_k}\sum_{i \in G_k} \bigg[f(X_i) + \frac{\xi_i}{\pi(X_i)}(Y_i - f(X_i))\bigg].
\end{equation}

For a differentiable functional $\theta = \psi(\bmu)$, where $\bmu = (\mu_1, \dots, \mu_K)$, and a policy $\pi(X_i)$ that can vary across components (e.g. $\pi^{(k)}(X_i)$ when estimating $\mu_k$), the plug-in estimator of $\theta$ is 
\begin{equation}\label{eq:estimator}
\hat{\theta}^{\pi} = \psi(\hat{\mu}_1^\pi, \dots, \hat{\mu}_K^{\pi}).
\end{equation}

For the (log) odds ratio example in the chest X-ray setting, we take $\bmu = (\mu_1, \mu_0)$ and $\psi(\bmu) = g(\mu_1) - g(\mu_0)$ where $g(\mu) = \log [\mu/(1 - \mu)]$. Under mild regularity conditions (as specified in Appendix~\ref{app:regularity} Assumption~\ref{ass:basic_sampling}), the estimators $\hat{\mu}_{k}^\pi$ are asymptotically normal; consequently, $\hat{\theta}^\pi$ is also asymptotically normal with variance $V(\pi) = \text{AsyVar}(\hat{\theta}^\pi)$. This variance determines the width of confidence intervals and quantifies the effective precision attainable when only $\budget$ labels are collected while still leveraging predictions from $f$.

The considerations above reduce the optimized policy design problem to a variance-minimizing problem: how should we choose $\pi$, subject to the labeling budget and valid probability constraints, to minimize $V(\pi)$? Our main contribution is to answer this question by deriving an explicit expression for $V(\pi)$ in terms of the joint distribution of $(X, Z)$, the predictive model $f$, and the policy $\pi$.

One natural strategy is to define a scalar \textit{uncertainty score} $u(X)$ from the predictive model $f$ and let the labeling probability be proportional to this uncertainty. In the binary classification case, uncertainty is measured according to $u(x) = 2\min\{p(x), 1 - p(x)\}$, where $p(x) = \hat{\bbP}(Y = 1 \mid X)$ is the predicted class probability \cite{active}. Here, the predictive model is $f(\cdot) = p(\cdot)$, and the points with predicted class probability near $1/2$ are considered most uncertain. The proportional-to-uncertainty rule sets $\pi(X_i) \propto u(X_i)$, with a normalizing constant chosen so that $$\E[n_{\mathrm{labeled}}] = \E\bigg[\sum_{i=1}^n \xi_i\bigg] = \E[\pi(X)] \cdot n \leq \budget.$$
Because $\E[u(X)]$ is unknown, it must be estimated from the unlabeled data. This approach efficiently spreads labels across regions where predictions are uncertain. The odds ratio estimation problem is framed as the estimation of two separate means $\mu_1$ and $\mu_0$, which are then combined via the Delta Method in \cite{active}. However, while the proportional-to-uncertainty rule minimizes the asymptotic variance of each individual mean estimator, it does not minimize the asymptotic variance of the functional $\psi(\bmu)$. It also fixes the relationship between uncertainty and labeling probability to be linear, which can be suboptimal if very small or very large uncertainties should be treated separately.

To address these limitations, we parametrize the labeling policy as a smooth function of the uncertainty. Concretely, instead of forcing labeling probability to be linear in the uncertainty, we learn a flexible curve: $\pi_{\bbeta}(X_i) = g_{\bbeta}(u(X_i)),$
where $u(X_i)$ is a chosen uncertainty score and $g_{\bbeta}$ is a low-dimensional, spline-based function parameterized by $\bbeta$. The uncertainty therefore enters our method explicitly through $u(X_i)$, while the shape of the mapping from the uncertainty to the labeling probability is learned from the data.

We show that the resulting empirical variance $V(\pi_{\bbeta})$ is a convex function of $\bbeta$, and the optimal policy can be obtained by solving the tractable convex program
\begin{align*}
    \text{minimize } \;  & \text{Var}(\hat{\theta}^{\bbeta}), \\
    \text{subject to } \; & \hat{\E}[n_{\text{labeled}}] = \sum_{i=1}^n\pi_{\bbeta}(X_i) = \sum_{i=1}^n g_{\bbeta}(u(X_i)) \leq n_\text{budget}, \\
    &  0 \leq \pi(X_i) = g_{\bbeta}(u(X_i)) \leq 1, \; i = 1, \dots, n.
\end{align*}

We find that parametrizing policy probabilities as smooth functions of uncertainties can improve upon proportional-to-uncertainty labeling introduced in active statistical inference \cite{active}, even in relatively simple M-estimation settings such as mean estimation. The resulting policies are more robust to misspecification of the uncertainty scale and can adapt to the particular estimand. In Section \ref{sec:general} we extend the \method~framework to a broader class of estimators by deriving general expressions for $V(\pi_{\bbeta})$. In Section \ref{sec:discussion}, we discuss extensions to (i) simultaneous uncertainty quantification for multiple estimands; (ii) overlapping subgroups; and (iii) the use of data outside the (sub)groups of primary interest. Before detailing the method, we preview its utility in a real-world application.

\subsection{Chest X-ray For Cardiomegaly Classification} \label{sec:CHX_intro}

We return to the problem of estimating the odds ratio of cardiomegaly among patients above and below age 40. This is a scenario where group imbalance is expected---existing models for predicting cardiomegaly likely have been trained on data skewed heavily towards older patients due to the aging nature of chest-related health conditions. As a result, we expect predictions to potentially be less accurate for patients below age 40, but also fewer patients in this category, raising concerns of sensitivity for a quantity like the odds ratio. 

Using chest X-ray imaging data along with demographic features such as age and sex from CheXpert \cite{CheXpert}, a large publicly available dataset, we perform uncertainty quantification for the odds ratio of cardiomegaly among patients above and below 40, using a pre-trained model \cite{CheXpert-paper} to generate predicted labels. {To evaluate our method's performance on this task, we use the metric of effective sample size (ESS, denoted $\neff$): the number of samples needed for a uniform-at-random sampling method to achieve the accuracy of the labeling method under study---see Section \ref{sec:eval} for a formal definition.} The ratio $\neff/\budget$ measures the gain in ESS from additionally using ML predictions over using only the $\budget$ human-collected labels. {We compare the ESS using \method~against three alternative labeling policies: proportional-to-uncertainty \cite{active} (denoted active), uniform (reducing to PPI \cite{ppi++}), and classical (using only labels and no predictions).} Figure~\ref{fig:CHX_ESS} shows these ESS values across a range of labeling budgets given on the x-axis, averaged over 500 iterations per fixed labeling budget. We see that the green curves, corresponding to OPAL, dominate all others. That is, with optimally chosen labeling probabilities, we can achieve ESS on the order of nearly three times the number of actual labels collected by incorporating predicted labels adaptively. {Details including discussion of other quality metrics are in Section \ref{sec:CHX} and Appendix \ref{app:CHX_details}.} The source code for all implementation and experiments is available at \href{https://github.com/ginniema/OPAL-active-labeling-inference}{https://github.com/ginniema/OPAL-active-labeling-inference}.

\begin{figure}[H]
    \centering
    \includegraphics[width = 0.8\linewidth]{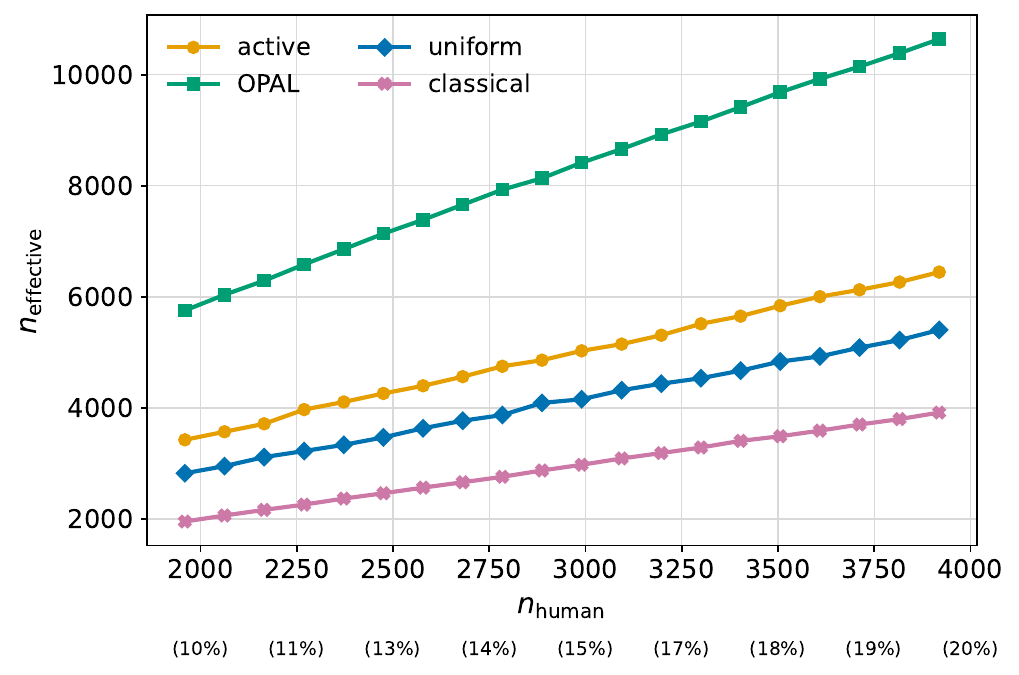}
    \caption{(a) Effective sample size and (b) coverage of each method listed in the legend. The x-axis $n_{\text{human}}$ denotes the labeling budget (how many human labels are collected) while the y-axis denotes effective sample size and coverage, respectively. The budget ranges from 10\% to 20\% of the total unlabeled observations, where labels are queried without replacement.}
    \label{fig:CHX_ESS}
\end{figure}

\section{{OPAL} for Active Statistical Inference}\label{sec:optim_label}

\method~can be broken into components consisting of two primary optimization modules: (1) an optimal labeling policy via a convex optimization program; and (2) a closed-form power tuning parameter $\lambda$. These modules are integrated into an end-to-end inference pipeline inspired by the basic structure from Zrnic and Cand\`{e}s \cite{active}, featuring both batched and sequential ``modes,'' as shown in Figure~\ref{fig:method}. 

In this section, we will introduce the optimal labeling policy module. At a high level, this module takes as input an initial, off-the-shelf predictive model $f$ and returns a labeling policy $\pi$ over the available unlabeled features. The policy is chosen by solving a convex program that targets statistical efficiency of the downstream estimator. Then, in Section~\ref{sec:inference}, this  labeling policy is incorporated in a complete inferential framework.

\begin{figure}[h!]
    \centering
    \includegraphics[width=\linewidth]{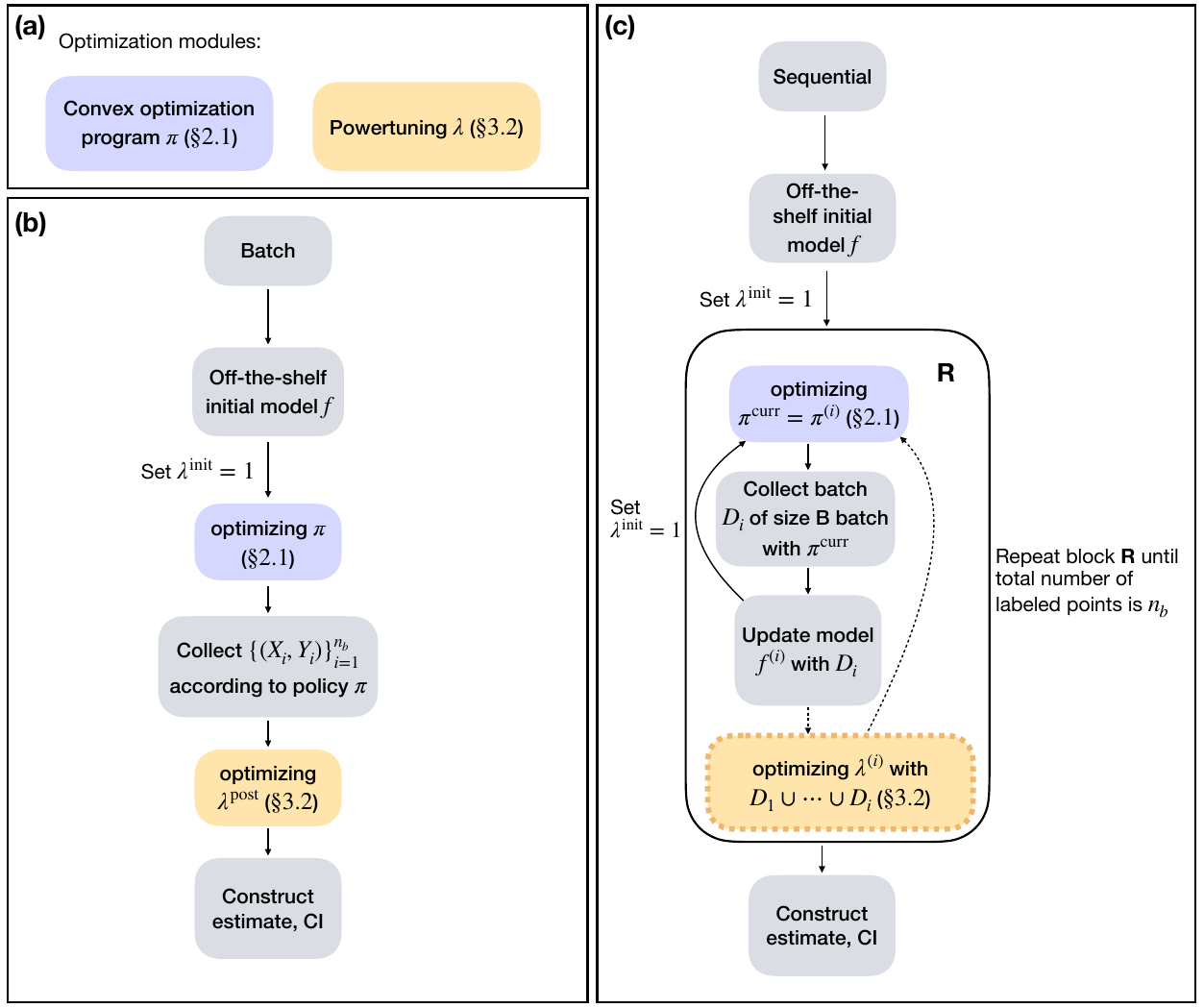}
    \caption{Overview of OPAL incorporating (a) optimization modules for finding labeling policy $\pi$ and best tuning parameter $\lambda$ (Section~\ref{sec:tuning}). The workflow using these two modules is given for (b) the batch setting where all covariate data is available; and (c) the sequential (online) setting where covariate data is revealed one-individual-at-a-time. Dashed lines around a module and dashed lines indicate optional modules. {A more complete version of this diagram, including scenarios when an initial $f$ is not available, is provided in Appendix~\ref{app:full_diagram}.}}
    \label{fig:method}
\end{figure}

\subsection{Framing min variance estimation as a convex optimization objective} \label{sec:optim}

We begin by revisiting the function-of-means setting introduced in Section \ref{sec:contributions}. Let $\theta = g(\bmu)$ be the parameter of interest, where $g$ is a smooth function and $\bmu = (\mu_1, \dots, \mu_K)$ is a vector of means such that $\mu_k = \E[s_k(X, Y)]$, where $s_k(\cdot)$ is some function of $X, Y$. Under standard regularity conditions for M-estimators, $\widehat{\bmu}^\pi = (\hat{\mu}_1, \dots, \hat{\mu}_K)$, where each $\hat{\mu}_k$ is as defined in Eq.~\eqref{eq:meanest} satisfies a central limit theorem (stated formally in Section \ref{sec:batch} and proved in Appendix \ref{app:CLT}):
$$\sqrt{n}(\hat{\bmu}^\pi - \bmu) \toD \N(0, \Sigma(\pi)),$$
where $\Sigma(\pi)$ denotes the asymptotic covariance matrix induced by the labeling policy $\pi$. To derive the asymptotic distribution for our target parameter $\theta$, we apply the Delta Method. For a smooth map $g$, the plug-in estimator $\hat{\theta}^\pi$ (defined in Eq.~\eqref{eq:estimator}) converges as
$$\sqrt{n}(\hat{\theta}^\pi - \theta) \toD \N(0, (\nabla g)^\top\Sigma(\pi)(\nabla g)).$$

Consequently, our objective is to minimize the asymptotic variance of this estimator. As shown in Appendix \ref{app:asyvar}, this variance takes a tractable form for optimization:
\begin{equation} \label{eq:mean_asyvar}
    V(\pi) = (\nabla g)^\top\Sigma(\pi)(\nabla g) =  \E\bigg[\frac{c(X)}{\pi(X)}\bigg] + C,
\end{equation}
where $C$ is a constant independent of the policy $\pi$. The numerator $c(X)$ captures the contribution of a data point to the variance. For the specific case of the log odds ratio $\theta = g(\mu_1, \mu_0)$, where $\mu_k = \E[Y \mid X \in G_k]$ and $g(x,y) = \log \frac{x}{1 - x} - \log \frac{y}{1 - y}$, this term is given by 
\begin{equation}\label{eq:OddsRatioVar}
    c(X) = \mathbb{E}[(Y - f(X))^2 \mid X]\bigg(\frac{1\{X \in G_1\}}{p_1^2\mu_1^2(1 - \mu_1)^2} + \frac{1\{X \in G_0\}}{p_0^2\mu_0^2(1 - \mu_0)^2}\bigg),
\end{equation}
as shown in Appendix~\ref{app:asyvar_OR}.
The expectation $\E[(Y - f(X))^2 \mid X]$ is the mean squared error of the predictor $f$  at $X$. If $f$ is consistent for the true conditional mean $\E[Y \mid X]$, this term simplifies to the conditional variance of the residual (since the outcomes are binary $Y \in \{0,1\}$): $\E[(Y - f(X))^2 \mid X] \approx f(X) (1 - f(X))$. However, directly minimizing $V(\pi)$ is impossible because it depends on unknown population quantities ($\mu_1, \mu_0$ and the conditional MSE). To make this optimization tractable, we construct a sample-based proxy $\hat{V}(\pi)$ by replacing population terms with empirical estimates derived from the burn-in set or training split (depending on whether we are in a batched or sequential setting):
$$\hat{V}(\pi) = \frac{1}{n}\sum_{i=1}^n \frac{\hat{c}(X_i)}{\pi(X_i)}.$$
Finding the optimal labeling policy $\hat{\pi}^*$ thus becomes a well-defined optimization problem: $$\hat{\pi}^* = \underset{\pi \in \Pi}{\arg\min} \; \hat{V}(\pi),$$ where the ``hat'' notation emphasizes that we are minimizing an \textit{estimate} of the asymptotic variance. 

{Before describing variance-minimizing policies, it is helpful to note that if we restrict the class of labeling policies to constant ones $\pi(x) \equiv \pi$ and enforcing the budget constraint $\E[\pi(X)] = \budget/n$, this yields the uniform sampling rule, i.e., each data point is labeled with the same probability. Under this choice, our estimator reduces to the prediction-powered estimator in \cite{ppi, ppi++}.}

To separate the general optimization machinery from specific estimands, we adopt the notion of an uncertainty score from \cite{active}. We denote $\hat{c}(X)$ as the uncertainty score $u(X)$ (up to a multiplicative factor). For the log odds ratio example, with the population-level $c(X)$ derived in Eq.~\eqref{eq:OddsRatioVar}, the uncertainty is driven by the predicted variance $u(X) = f(X)(1 - f(X))$, weighted by the group-specific terms: $$\hat{c}(X) = u(X) \bigg[\frac{1\{X \in G_1\}}{(\hat{\mu}_1^2(1 - \hat{\mu}_1)^2)} + \frac{1\{X \in G_0\}}{(\hat{\mu}_0^2(1 - \hat{\mu}_0)^2)\}}\bigg].$$ In this framework, the dependence of the variance estimate on the features is fully characterized by this one-dimensional summary $u(X)$.

Theoretically, one could attempt to solve this optimization problem directly by treating the labeling probabilities $\{\pi(X_i)\}_{i=1}^n$ as $n$ independent variables. The analytical solution to this problem corresponds to the classical Neyman Allocation, where labeling probabilities are proportional to the square root of the local uncertainty contribution \cite{NeymanAlloc, active}. This does not strictly enforce the probability constraints $\pi(X) \in [0,1]$, which requires numerical optimization. By employing a change of variables---defining $d_i = \prod_{j \neq i} \pi(X_j)$---this formulation can be cast as a Geometric Program, which is theoretically tractable and solvable using standard convex optimization tools \cite{boyd2007gp, agrawal2019dgp}.

However, if the uncertainty estimates are misspecified or noisy, treating each $\pi(X)$ individually makes the optimization overly sensitive to potential errors. Instead, we propose the parametrization of each $\pi(X)$ as a \textit{smooth} function of the uncertainties, which can guard against individual low-quality uncertainty estimates and provide overall regularization; that is, we set $\pi_{\bbeta}(X_j) \equiv \pi_{\bbeta}(u(X_j))$, with parameters $\bbeta$. For each mean $\mu_k$, we define a labeling policy $\pi^k_{\bbeta}$; these are concatenated for the overall policy $\pi$, parametrizing each as a spline function of the uncertainties $u(X_i)$.

To implement this smoothing, we define the spline using two externally provided tuning parameters---the degree $d$ and the number of internal knots $K$---along with a set of free coefficients $\{\beta_j^{(k)}\}_{j=1}^B$ to be estimated. Using the shorthand notation $U_i := u(X_i)$, we parametrize the labeling probability $\pi^{(k)}_{\bbeta^{(k)}}(X_i)$ via the log-linear link:
\begin{equation} \label{eq:splineprobs}
    \log \frac{1}{\pi^{(k)}(X)} = 
\boldsymbol{h}(U)^\top\boldsymbol{\beta}^{(k)}\implies \pi^{(k)}(X) = \exp\{-\sum_{b=1}^B \beta_b^{(k)}h_b(U)\},
\end{equation}
where $\{h_b\}$ are the spline basis functions, see Appendix \ref{app:spline} for further details on spline basis construction. We use functions of the scalar uncertainty $U$ rather than $X$ directly to utilize a one-dimensional summary, avoiding the complexity of tensor product splines required for high-dimensional $X \in \R^d$. We propose a default setting of 5 knots and 3 degrees of freedom. While these are technically hyperparameters, we find in simulation and experiments that the procedure is robust to these choices. This robustness is crucial, as the labeling budget constraint prevents us from performing standard hyperparameter tuning via multiple runs.

This parametrization naturally enforces valid probability constraint. Since $\exp\{-u\} \in (0,1]$ for all $u \geq 0$, ensuring that $\pi^{(k)}(X) \in (0, 1]$ is equivalent to constraining the spline exponents to be nonnegative: $\boldsymbol{h}(U)^\top\boldsymbol{\beta}^{(k)} \geq 0$. Crucially this formulation results in an exponential-cone program. Because the objective is log-convex in the parameters $\{\beta_j^{k}\}_{j=1, \dots, B}^{k = 1, \dots, K}$, we are guaranteed to efficiently find global solutions using standard convex optimization software such as \texttt{cvxpy}\cite{cvxpy1, cvxpy2} and \texttt{MOSEK}\cite{mosek}.

Formally, the explicit constrained optimization problem for a specified labeling budget $\budget$ is defined as follows:
\begin{align} 
    \text{minimize}  & \qquad \sum_{i=1}^n c(U_i) \exp\{\sum_{b=1}^B \beta_b h_b(U_i)\}, \label{eq:optim_full} \\
    \text{subject to}  & \qquad \sum_{b=1}^B \beta_b h_b(U_i) \geq 0, \quad \forall i \in [n], \nonumber \\
    & \qquad \sum_{i=1}^n \exp\{-\sum_{b=1}^B \beta_bh_b(U_i)\} \leq n_{\text{ budget}}. \nonumber 
\end{align}
Solving this problem yields the \textbf{optimized policy for acquiring labels} $\pi^\text{opt}$, which minimizes the estimate of the asymptotic variance $\hat{V}(\pi^\text{opt})$ of the corresponding estimator. In settings with multiple subgroups, this policy automatically balances the budget between groups in a principled manner, weighting them according to both the difficulty of mean estimation within each group and their respective contributions to the overall estimator variance. {If it is desired, monotonicity can be enforced with additional linear inequality constraints (see Appendix~\ref{app:monotonicity} for more details). We do not impose monotonicity in the main formulation, since optimal allocation may not necessarily be strictly monotone when the uncertainties and thus $\hat{c}(\cdot)$ are noisy or misspecified and when subgroup weights interact with the policy. }

\begin{figure}[h]
    \centering
    \includegraphics[width=0.8\linewidth]{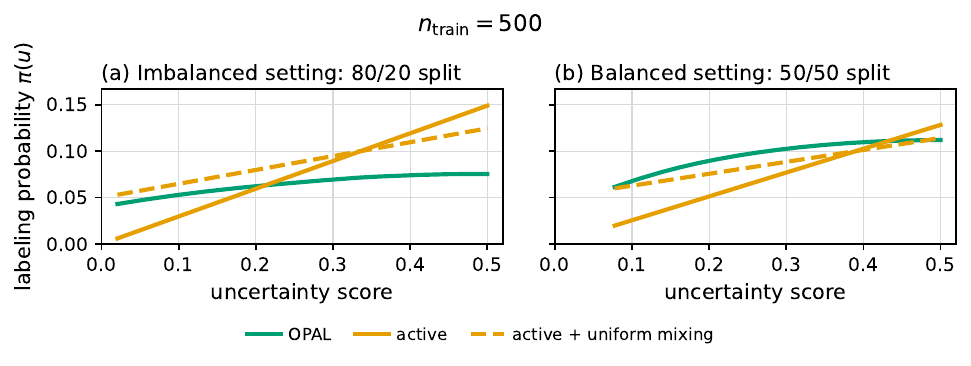}
    \caption{Labeling policy generated via active vs. OPAL based on data generated from (a) balanced setting and (b) imbalanced setting, as described in Section~\ref{sec:sim_odds}. A logistic regression model was trained on $n_{\mathrm{train}} = 500$ data points to predict label $Y \in \{0,1\}$. The x-axis represents uncertainty scores $u(X) = \min(f(X), 1 - f(X))$ and the y-axis is the labeling probability $\pi(X) = \pi(u(X))$.}
    \label{fig:probs}
\end{figure}

Figure~\ref{fig:probs} contrasts the resulting probabilities from \method~against proportional-to-uncertainty labeling, with and without mixing with the uniform sampling rule \cite{active}. While proportional sampling typically requires an added step---mixing with the budget-specified uniform sampling rule to provide regularization and guard against extreme probabilities---our method achieves this naturally. The introduction of the spline basis acts as a data-driven regularizer: it performs a similar task to uniform mixing but adapts flexibly to the specific problem geometry. This avoids the need to pre-specify a fixed mixing constant---a hyperparameter that cannot be tuned in real-world settings where labels are unknown. Further, the spline approach allows for non-linear mixing, offering greater expressivity than simple linear interpolation. In the next section, we show the same convex-program structure extends beyond the functions-of-means setting.

\subsection{A Broader Class of Estimands} \label{sec:general}

While the mean functionals discussed previously cover many applications including the log odds ratio, complex tasks may require estimating parameters defined implicitly through estimating equations. We now expand our framework to this broader class of estimators. Consider a target estimand $\psi = \Phi(P)$, a functional of the data distribution $P$, and a corresponding estimator $\hat{\psi}$ which admits an efficient influence function (EIF) of the general form:
\begin{equation} \label{eq:EIF}
    \phi(X,Y, \xi) =  \frac{\xi}{\pi(X)}\zeta(X,Y) + h(X) - \psi,
\end{equation}
where $\xi_i \in \{0,1\}$ indicates whether the label $Y_i$ is observed; in particular, $\xi_i \sim \Bern(\pi(X_i))$. We denote the propensity score---our labeling policy---as $\pi(x) = \P(\xi_i = 1 \mid X_i = x)$. The remaining terms satisfy standard influence function properties: $\zeta(X,Y)$ is a mean-zero residual term such that $\E[\zeta(X,Y) \mid X] = 0$, and $h(X)$ is an unbiased proxy for the parameter, satisfying $\E[h(X)] = \psi$. We still assume the labels are missing at random; that is, $\xi \indep Y \mid X$. Under this structure, the asymptotic variance of the estimator $\hat{\psi}$ is given by $$V[\phi] = \E[\frac{c(X)}{\pi_\beta(X)}] + V[h(X)], \; \text{where } c(X) = \E[\zeta(X,Y)^2 \mid X].$$
The first term matches the form of the asymptotic variance derived for functions of means in Eq.~\eqref{eq:mean_asyvar}. Since the second term $V[h(X)]$ is constant with respect to the labeling policy $\pi$, the convex optimization framework developed in Section \ref{sec:optim} is directly applicable. This allows us to find optimal labeling policies for a vast array of statistical targets beyond simple means.

With the proposed EIF, we now discuss strategies for constructing an estimator $\hat{\psi}$ that is asymptotically linear\footnote{Here, we define an estimator $\hat{\psi}$ as asymptotically linear if it admits a first-order expansion $$\sqrt{n}(\hat{\psi} - \psi) = \frac{1}{\sqrt{n}}\sum_i \phi(O_i) + o_p(1),$$ for $\phi$ a mean-zero function with finite variance, which is known as the influence function. \label{footnote:asymptotically_linear}} and admits $\phi$ as its EIF.  We let $\hat{\zeta}$ and $\hat{h}$ denote estimators of the nuisance functions defined above, which can be obtained using pilot data or a burn-in sample in the sequential setting. For an asymptotically linear estimator $\hat{\psi}$ for a parameter $\psi$, we can view the EIF $\phi(X,Y)$ as a measure of how much a data point $(X,Y)$ affects $\hat{\psi}$. 

\paragraph{One-step estimation}
The most direct approach is the one-step estimator \cite{bickel, vdv}. This procedure corrects the bias of an initial plug-in estimate by adding the empirical mean of the estimated influence function:
$$\hat{\psi} = \frac{1}{n}\sum_{i=1}^n \bigg[\hat{h}(X_i) + \frac{\xi_i}{\pi(X_i)} \hat{\zeta}(X_i, Y_i)\bigg].$$
This formulation generalizes the mean estimator proposed earlier in Eq.~\eqref{eq:meanest}. Specifically, in the case of mean estimation where $\zeta(X,Y) = Y - f(X)$ and $h(X) = f(X)$, this recovers the exact form of our primary estimator. Structurally, it maintains the ``plug-in + residual'' architecture in existing prediction-powered and active inference literature \cite{ppi++, active}. The corresponding one-step estimator has approximate asymptotic variance $\Var(\phi(X, \xi, \xi Y))/n$ as shown in \cite{vdv}  (estimated in practice by the empirical variance of the estimated influence function divided by $n$), validating the use of our variance-minimization objective.

\paragraph{Targeted Maximum Likelihood Estimation (TMLE)} {While the one-step estimator is asymptotically efficient, it is not a substitution estimator: it forms a bias-corrected estimate by adding an influence-function term to an initial plug-in estimate. Consequently, it does not automatically enforce parameter constraints in finite samples (e.g. probability may fall outside $[0,1]$). TMLE instead targets the parameter by updating the nuisance fit within a model that preserves the natural parameter space and thus respects these constraints \cite{tmle, tmle2}.}

Starting from initial nuisance estimates (e.g. an initial $\hat{h}$ and any additional components needed to evaluate the EIF), TMLE defines a parametric fluctuation submodel $\{h_{\varepsilon} : \varepsilon \in \R\}$ that passes through $\hat{h}$ at $\varepsilon = 0$ and is chosen so that moving along $\varepsilon$ changes the estimand in the EIF direction. We define the clever covariate $H_i$ to be the known weighting term that aligns the fluctuation with the EIF: in the missing-label mean case it is $H_i = \xi_i / \pi(X_i)$. 

The fluctuation parameter $\hat{\varepsilon}$ is then estimated so that the empirical EIF has mean approximately zero; that is, the targeted estimate solves the EIF estimating equation. Finally, the output TMLE is the plug-in estimator $\hat{\psi}_{\mathrm{TMLE}} := \Phi(\hat P^*)$, where $\hat P^*$ is the targeted estimate of the data distribution. This is done by replacing initial nuisance estimate $\hat h$ with the targeted version $\hat h^*$ along with other updating other nuisance components needed to evaluate the EIF. As a concrete example, if the target is the mean $\psi = \Phi(P) = \E_{P}[Y] = \E_P[h(X)]$ where $h(X) = \E_P[Y \mid X]$, then the nuisance parameter update produces $\hat h^*$. The corresponding TMLE is the plug-in estimator $\hat\psi = \frac{1}{n}\sum_{i=1}^n \hat h^*(X_i)$.

For binary outcomes, the targeting step can be implemented with a logistic fluctuation model using the clever covariate $H_i$ and inverse-propensity weights $\xi_i/\pi(X_i)$ with offset $\mathrm{logit}(\hat h(X_i))$, so only a low-dimensional parameter is estimated. Under standard regularity conditions, the TMLE is asymptotically linear with influence function $\phi$, and thus has asymptotic variance $\Var(\phi(X, \xi, \xi Y))/n$ \cite{vdv}. This provides an alternative strategy for constructing estimators with the same variance decomposition as the one-step estimator, while preserving the parameter-space constraints in finite samples. Details on the TMLE algorithm are provided in Appendix~\ref{app:TMLE}.

\subsection{Examples}\label{sec:examples}
This framework accommodates a diverse set of estimands, notably including non-M estimators. We provide distinct examples below to illustrate the flexibility of the influence function form. The corresponding residual term $\zeta(X, Y)$ and average conditional squared residual $c(X) = \E[\zeta(X,Y)^2 \mid X]$ are given in Table~\ref{tab:EIF_form}. Detailed derivations are provided in Appendix~\ref{app:EIF}.

\paragraph{Example 1 (Linear functional).} We consider the functional $\psi = \E[h(X)Y]$. This arises in a number of applications including subgroup mean estimation, where $h(X) = 1\{X \in G\}$. In this setting, the conditional variance $c(X)$ represents the heterogeneity of the outcome within the subgroup. This can also be generalized to the high-dimensional case where $Y, h(X) \in \R^d$ for $d > 1$ and $\psi = \E[h(X)^\top Y]$.

\paragraph{Example 2 (Regression coefficient).}  {Let $\psi = \gamma_j$ be the $j$th coordinate of the population least squares coefficient $\gamma$, defined by $\E[X(Y - X^\top\gamma)] = 0$; that is, $X^\top\gamma$ is the best linear predictor of $Y$ from $X$ in squared error loss. We define the population covariance matrix $\Sigma = \E[XX^\top]$ and weight $w_j(X) = (\Sigma^{-1}X)_j$. The residual $Y - X^\top\gamma$ is orthogonal to the span of $X$ in the population, but the conditional mean $\E[Y - X^\top\gamma]$ is not necessarily zero unless the linear model is correctly specified. If the linear model is correctly specified, then $m(X) := \E[Y \mid X] = X^\top\gamma$, so the residual term and conditional second moment in Table~\ref{tab:EIF_form} are $c(X) = w_j(X)^2\E\!\left[(Y-X^\top\gamma)^2\mid X\right]$ and $\zeta(X,Y) = w_j(X)(Y - X^\top \gamma)$.}

\paragraph{Example 3 (U-statistic).} Let $\psi = \E[k(Z, Z')]$ where $k(\cdot, \cdot)$ is a symmetric kernel function. We derive the influence function using the Hajek projection principle \cite{hajeknotes, vdv}. This entails defining the conditional expectation $\tilde{\psi}(x) = \E[k(z, Z') \mid Z = z]$, a first-order projection of $\psi$. Then, the corresponding residual term $\zeta(X,Y)$ is the centered (mean zero) version of this projection.

\paragraph{Example 4 (Kendall's tau).} For the final example, we consider a specific U-statistic: Kendall's rank correlation coefficient (Kendall's $\tau$), which is a non-parametric measure of correlation with $Z = (X,Y)$ and $k(Z, Z') = k((X,Y), (X', Y')) = \text{sign}(Y-Y')\text{sign}(X-X')$. When $Y \in \{0,1\}$ (binary outcomes), the projection $\tilde{\psi}(X,Y)$ simplifies significantly such that $c(X) = \Var(\tilde{\psi}(X,Y) \mid X)$ can be expressed in terms of the class-conditional distributions, $p(X) = \P(Y = 1 \mid X)$, $F_0(t) = \P(X \leq t \mid Y = 0)$ and $F_1(t) = \P(X \leq t \mid Y = 1)$, and the class proportions $p_1 = \P(Y = 1), p_0 = \P(Y = 0)$. 

\begin{table}[t]
\caption{Residual term $\zeta(X,Y)$ and conditional second moment $c(X)=\E[\zeta(X,Y)^2\mid X]$ for examples in Section~\ref{sec:examples}.}\vspace{0.5em}
\label{tab:EIF_form}
\centering
\small
\setlength{\tabcolsep}{6pt}
\renewcommand{\arraystretch}{1.15}

\begin{tabular}{lcc}
\toprule
\textbf{Example} & \textbf{$c(X)$} & \textbf{$\zeta(X,Y)$} \\
\midrule
Linear functional 
& $h(X)^2\,\Var(Y\mid X)$
& $h(X)\bigl(Y-\E[Y\mid X]\bigr)$ \\[2pt]

Regression coefficient 
& $w_j(X)^2\,\Var(Y \mid X)$ & $w_j(X)\left(Y - \E[Y \mid X]\right)$ \\[2pt]

U-statistic 
& $\Var\!\bigl(\tilde{\psi}(X,Y)\mid X\bigr)$
& $\tilde{\psi}(X,Y)-\E\!\bigl[\tilde{\psi}(X,Y)\mid X\bigr]$ \\[2pt]

Kendall's $\tau$ 
& $p(X)\{1-p(X)\}\Big(1-2\big[p_0F_0(X)+p_1F_1(X)\big]\Big)^2$
& $\tilde{\psi}(X,Y)-\E\!\bigl[\tilde{\psi}(X,Y)\mid X\bigr]$ \\
\bottomrule
\end{tabular}
\end{table}
\vspace{.1in}

Examples 1 and 2 correspond to familiar M-estimators, while Examples 3 and 4 involve U-statistics. Our framework, however, is not limited to these two classes. For instance, the average treatment effect (ATE) in a causal inference setting has full-data efficient influence function
$$\phi_{\mathrm{ATE,full}}(X,A,Y) = \frac{A}{e(X)}\{Y - m_1(X)\} - \frac{1-A}{1-e(X)}\{Y - m_0(X)\} + m_1(X) - m_0(X) - \psi,$$
where $Y(a)$ denotes the potential outcome under treatment $a\in\{0,1\}$, $Y=Y(A)$ is the observed outcome, $e(X)=\P(A=1\mid X)$ is the propensity score, and $m_a(X)=\E[Y\mid A=a,X]$. 

In our missing-label model, we allow the querying rule to depend on both features and treatment, with $\xi\in\{0,1\}$ satisfying $\P(\xi=1\mid X,A)=\pi(X,A)$ and $\xi \perp Y \mid (X,A)$. The corresponding observed-data EIF can be written as
$$\phi_{\mathrm{ATE}}(X,A,\xi,\xi Y) = \frac{\xi}{\pi(X,A)}\left[ \frac{A}{e(X)}\{Y - m_1(X)\} - \frac{1-A}{1-e(X)}\{Y - m_0(X)\} \right] + m_1(X) - m_0(X) - \psi,$$ where $\bar\phi_{\mathrm{ATE}}(X,A)=\E[\phi_{\mathrm{ATE,full}}(X,A,Y)\mid X,A]
= m_1(X)-m_0(X)-\psi$. The details for this example are again found in Appendix~\ref{app:EIF}. These examples illustrate that \method~results apply to a broad range
of asymptotically linear estimators beyond standard M-estimators.

\section{End-to-end inference with optimal labeling policy} \label{sec:inference}
Section~\ref{sec:optim_label} focused on designing a variance-minimizing labeling policy. In this section, we integrate that policy into a complete inference pipeline—specifying the estimator, confidence interval construction, and the algorithmic workflow in both batch and sequential modes, as depicted in Figure~\ref{fig:method}. We first establish asymptotic normality under a learned policy in the batch setting, then introduce power tuning and cross-fitting for improved efficiency and robustness, and finally adapt the procedure to the online setting where both the predictor and policy are updated over time.

\subsection{Batch inference with \method} \label{sec:batch}

To streamline the presentation of statistical validity, we state the main CLT and its corollary, which guide construction of confidence intervals; the intermediate propositions and all proofs are deferred to Appendix \ref{app:CLT}. 

We focus on the context of odds ratio estimation, where the data consists of two disjoint groups ($Z_i \in \{0,1\}$). All unlabeled data points $X_1, \dots, X_n$ are observed at the start. The optimized labeling policy for estimating each group mean $\mu_k = \E[Y \mid Z = k]$ is denoted by $\pi^{(k)}_{\bbeta_k}(\cdot)$, with the corresponding estimate $\hat{\mu}_k^\pi$ constructed according to Eq.~\eqref{eq:meanest}. In practice, we solve the convex optimization problem introduced in Eq.~\eqref{eq:optim_full} using sample-level quantities to obtain the estimated optimal spline coefficients $\hat{\bbeta}_k$. The resulting global policy $\pi(X_i)$ selects samples according to the group membership: 
\begin{equation} \label{eq:policy_OR}
    \pi(X_i) = 1\{i \in G_1\}\pi^{(1)}_{\hat{\bbeta}_1}(X_i) + 1\{i \in G_0\}\pi^{(0)}_{\hat{\bbeta}_0}(X_i),
\end{equation}
and we accordingly collect labels $\xi_i \sim \Bern(\pi(X_i))$ to form the observed data $O_i = (X_i, \xi_i, \xi_iY_i)$. The resulting log odds ratio estimate follows Eq.~\eqref{eq:estimator}. 
\begin{theorem}[CLT for batch inference on mean estimation] \label{thm:batch_mean} 
Under Assumptions~\ref{ass:basic_sampling}--\ref{ass:emp_process} in Appendix~\ref{app:regularity}, the estimator $\hat{\bmu}^{\pi}$ constructed using the learned policy $\pi_{\hat{\bbeta}}$ satisfies $$\sqrt{n}(\hat{\bmu}^\pi - \bmu)= \sqrt{n}(\hat{\bmu}^{\hat{\boldsymbol{\beta}}} - \bmu) \toD \N(0, \Sigma_*),$$ where the exact form of $\Sigma_*$ can be found in Appendix \ref{app:CLT}.
\end{theorem}

A central technical challenge is that the learned policy $\pi_{\hat \bbeta}$ depends on the entire covariate sample $\{X_i\}_{i=1}^n$, so the labeled outcomes are no longer sampled with a fixed, independent propensity. In Zrnic and Cand\`{e}s \cite{active}, the policy parameter $\eta$ is assumed to lie on a discrete grid $\mathcal{H}$. This discreteness allows the labeling indicator $\xi_i$ to be coupled with a Bernoulli variable depending on a fixed oracle $\eta^*$, effectively sidestepping the complex dependence structure.

In contrast, our framework requires optimizing $\bbeta$ over a continuous, multi-dimensional (fixed) space involving all data points $X_1, \dots, X_n$. Consequently, the learned policy $\pi_{\hat{\bbeta}}$ depends on the entire dataset, and the coupling argument is insufficient. We therefore combine an oracle-policy CLT, consistency of $\hat{\bbeta}$, and a stochastic decomposition/equicontinuity argument 
(Appendix~\ref{app:equicontinuity}), adapting techniques from empirical process theory \cite{epnotes, vdv2}. In particular, we decompose the error for a single mean $\hat{\mu}^\pi_k$ into an oracle term, a parameter-estimation term, and a remainder:
\begin{equation} \label{eq:expansion}
    \sqrt{n}(\hat{\mu}^\pi_k - \mu_k) = \underbrace{\sqrt{n}(\hat{\mu}^{\pi^*}_k - \mu_k)}_{\text{Oracle Term}} + \underbrace{\sqrt{n}(\psi(\hat{\bbeta}_k) - \psi(\bbeta^*_k))}_{\text{Parameter estimation cost}} + \underbrace{B_n}_{\text{Remainder}}.
\end{equation}
Asymptotic normality is then shown by the following: the oracle term is asymptotically linear\footref{footnote:asymptotically_linear} by the oracle CLT, we can show the parameter estimation term is exactly zero in this case, and the remainder $B_n$ is shown to be asymptotically negligible via an asymptotic equicontinuity argument. 

This proof strategy offers an advantage in generality. Unlike grid-based approaches, it relies on the consistency of the continuous policy parameter together with stochastic equicontintuity, rather than on selection over a finite grid. This makes the argument directly applicable to multi-dimensional continuous hyperparameters, such as the spline coefficients.

This decomposition also yields the joint normality of the odds ratio components, with diagonal covariance due to the disjoint groups. Applying the Delta Method to the log-odds transformation gives the practical confidence intervals used in our algorithm.

\begin{corollary}[CLT for batch inference  on odds ratio]\label{corr:batch_odds}
Given the joint normality of the subgroup means, $$\sqrt{n}\bigg(\begin{bmatrix} \hat{\mu}_1^\pi \\ \hat{\mu}_0^\pi \end{bmatrix} - \begin{bmatrix} \mu_1 \\ \mu_0 \end{bmatrix}\bigg) \toD \N\bigg(\mathbf{0}, \begin{bmatrix} V_1 & 0 \\ 0 & V_0 \end{bmatrix}\bigg),$$ the log odds ratio estimator $\hat{\theta}^\pi$ satisfies
$$\sqrt{n}(\hat{\theta}^\pi - \theta) \toD \N(0,\, V^*), \; \text{with } V^* := \frac{V_1}{\mu_1^2(1 - \mu_1)^2} + \frac{V_0}{\mu_0^2(1 - \mu_0)^2}.$$
Taking a consistent estimate $\hat{V} \toP V^*$, a valid $1 - \alpha$ confidence interval is given by
$$\hat{\theta}^\pi \pm z_{1 - \alpha/2}\sqrt{\frac{\hat{V}}{n}},$$
where $z_{1-\alpha/2}$ is the $1 - \alpha/2$ quantile of a standard normal random variable.
\end{corollary}

In practice, a consistent estimate of $\hat{V}$ can be found by taking consistent estimates of $V_1$, $V_0$, $\mu_1$, and $\mu_0$ and using the standard plug-in estimator. In particular, $V_1$ and $V_0$ are each asymptotic variances for mean estimates, so they can be consistently estimated with the sample variance of the individual influence function contributions: for $k = 0$, $\hat{V}_0 = \frac{1}{n_0-1}\sum_{j \in G_0}(T_j - \bar{T}_0)^2$ where $T_j = f(X_j) + \frac{\xi_j}{\pi(X_j)}(Y - f(X_j))$ and thus $\bar{T}_0 \equiv \hat{\mu}_0^\pi$, and similarly for $k = 1$.

Algorithm \ref{alg:batch} summarizes the complete procedure for implementation. 

\begin{algorithm}[H]
\caption{Algorithm for batch active inference with optimal labeling}\label{alg:batch}
\begin{algorithmic}[1]
\State \textbf{Input:} Classifier $f$, budget $\budget$, unlabeled features $(X_1, \dots, X_n)$, error level $\alpha$.
\State Compute uncertainty scores $u(x_i) = f(x_i)(1 - f(x_i))$.
\State Solve convex program (Sec. \ref{sec:optim}) to find optimal policies $\pi^{(1)}$ and $\pi^{(0)}$.
\State \textbf{Labeling:} Sample $\xi_i \sim \Bern(\pi(X_i))$ and collect labels $\{Y_i: \xi_i = 1\}$.
\State \textbf{Estimation:} Compute means $\hat{\mu}_1^\pi, \hat{\mu}_0^\pi$ via Eq.~\eqref{eq:meanest}.
\State Return confidence interval per Corollary \ref{corr:batch_odds}, estimating $V_1, V_0$ via sample variance as described above.
\end{algorithmic}
\end{algorithm}

While the preceding analysis focused on the odds ratio, the theoretical machinery extends naturally to the broader class of estimands defined in Section~\ref{sec:general}. Because the labeling policy $\pi$ is optimized to minimize the variance component $\E[c(X) / \pi(X)]$, the same consistency and normality guarantees apply to any estimator $\hat{\psi}$ constructed via one-step estimation or TMLE, provided the influence function structure holds. A formal statement and proof is provided in Appendix \ref{app:generalCLT}. Practically, this allows us to construct valid confidence intervals for complex parameters using the same workflow introduced in Algorithm \ref{alg:batch}. Given a consistent estimator $\hat{\sigma}$---such as the empirical variance of the influence function contributions---for the limiting standard deviation $\sigma_*$, a valid symmetric $1 - \alpha$ interval is $[\hat{\psi} \pm z_{1 - \alpha/2} \frac{\hat{\sigma}}{\sqrt{n}}].$ 

\subsection{Enhancing efficiency via power tuning}\label{sec:tuning}
Optimizing the labeling policy targets the policy-dependent variance term $\E[c(X)/\pi(X)]$, but we can further reduce the variance by \textit{power tuning}~\cite{ppi++}, which introduces a one-dimensional control variate based on the predictor $f$. For a fixed policy $\pi$, power tuning preserves unbiasedness for all choices of weight $\lambda$ and chooses $\lambda^*$ which minimizes variance. To make things concrete, consider mean estimation under labeling policy $\pi$ and define the power-tuned estimator \begin{equation}\label{eq:mean_power}
    \hat{\mu}^{\pi,\lambda} = \frac{1}{n}\sum_{i=1}^n \left[\lambda f(X_i) + \frac{\xi_i}{\pi(X_i)}(Y_i - \lambda f(X_i))\right]. 
\end{equation}
This estimator remains unbiased for $\mu = \E[Y]$ for every $\lambda$. Conceptually, the weight $\lambda$ controls the reliance on the predictor $f$: if $\lambda = 0$, then $\hat{\mu}^{\pi, \lambda}$ reduces to the classical mean estimator using only labeled data while if $\lambda = 1$, then $f$ is well-aligned with $Y$ and $\hat{\mu}^{\pi, \lambda}$ is exactly the estimator introduced in Section~\ref{sec:contributions}. For a fixed policy $\pi$, it can be shown that the optimal $\lambda$ which minimizes the asymptotic variance is the usual ordinary least squares/control-variate slope given by a ratio of the covariance between the predictions and labels to the variance of the predictions: $$\lambda^* = \frac{\Cov(\tilde{f}(X), \tilde{Y}))}{\Var(\tilde{f}(X))}, \; \text{where } \tilde{f}(X) := \frac{1 - \pi(X)}{\pi(X)}f(X) \; \text{and} \; \tilde{Y} = \frac{\xi}{\pi(X)}Y,$$ as in \cite{ppi++}. In practice, we estimate $\lambda^*$ using the plug-in sample covariance and variance. See Appendix~\ref{app:add-on} for details and derivations in the log odds ratio and general EIF analogue. As shown in Figure~\ref{fig:method}(b), power tuning is implemented twice: $\lambda$ is first estimated on pilot/training data prior to label acquisition, and a second time after collecting labels for the final estimator.

\subsection{Cross-fitting to prevent overfitting}\label{sec:crossfitting}

Ideally, the predictive model $f$ and the labeling policy $\pi$ should be independent of the data used for final estimation. If $f$ is trained (or fine-tuned) on the same data used to compute $\hat{\mu}$, the estimator may suffer from overfitting bias. To mitigate this, we employ cross-fitting, adapting the procedure from \cite{cppi} to the \method~framework.

The general principle is to split the data into disjoint folds (e.g. $K = 2$). We use one fold to finetune (and/or train) the model and the other fold to construct an estimate. Then, the roles of the two folds are reversed to construct a second estimator. The final estimator is the weighted average of two fold-specific estimates. Algorithm~\ref{alg:crossfitting} in Appendix~\ref{app:algorithms} details this procedure. Since we begin with an entirely unlabeled dataset, we first run the optimization program (Section~\ref{sec:optim}) on the full dataset to determine the global labeling policy and explicitly collect labels before the splitting occurs.

\subsection{Sequential inference with \method} \label{sec:sequential}

{The batch setting assumes all unlabeled data are available upfront and that a single global labeling policy can be computed before collecting any labels. Many applications, however, benefit from \textit{sequential} inference, in the sense that labeling decisions are allowed to adapt to information revealed by previously queried outcomes. This includes truly online settings (e.g., clinical trials or streaming monitoring) where data arrive over time, but it also includes offline settings in which all unlabeled data is available at $t = 0$. Even then, it can be advantageous not to fix the predictor $f$ (and hence the uncertainty scores $u$ and the optimized policy $\pi$) a priori. }

{Specifically, as labeled observations accumulate, we can update the predictor to obtain $f_t$ with improved calibration and reduced residual variance, and recompute uncertainty scores $u_t$ that more accurately reflect the contribution $c(X)$ to the downstream variance objective. Updating the policy $\pi_t$ accordingly can allocate the remaining budget more efficiently than a policy computed from the coarse initial model.}

The full end-to-end process for sequentially obtained data is shown in Figure~\ref{fig:method}(c). We will follow the general procedure introduced by Cand\`{e}s and Zrnic~\cite{active}, with the key change of using the labeling policy introduced in Section~\ref{sec:optim_label} and iteratively updating the power tuning parameter $\lambda$ rather than only as a final step. This procedure allows for the finetuning of $f_t$ as more data are labeled, adapting to the observed responses $Y_i$. 

\begin{theorem}[Mean CLT for sequential inference]\label{thm:seq}
    Let $\Delta_t = f_t(X_t) + \frac{\xi_t}{\pi_t(X_t)}(Y_t - f_t(X_t))$ such that $\hat{\mu}^\pi = \frac{1}{t} \sum_{i=1}^n \Delta_t$. Suppose $\Delta_t$, which are martingale increments, satisfy the Lindeberg condition. Then, $$\sqrt{n}(\hat{\mu}^\pi - \mu) \toD \N(0, V^2).$$
\end{theorem}
\noindent The proof follows the same martingale CLT argument from \cite{active}. Once this CLT is established, Corollary~\ref{corr:batch_odds} can be invoked to construct confidence intervals. A corresponding extension beyond mean functionals is given in Appendix \ref{app:seq_general}.

As in the batch setting, we assume a pre-specified labeling budget $\budget$ and a target error level $\alpha \in (0,1)$. The sequential algorithm is conceptually similar to the batch procedure, but differs in several important aspects. First, data points are revealed one at a time (or $k$ at a time), so the labeling policy cannot be fixed in advance for all observed features. Second, one may begin without a predictive model $f$ (or initial model to be updated). In that case, we recommend labeling a small initial set of points (referred to as a burn-in sample) to train a starting model $f_0$.

Following the set-up of~\cite{active}, we also specify a fine-tuning batch size $B$. For example, when $B = 100$, the model is fine-tuned after every 100 newly labeled observations. To spread the labeling budget across all $n$ unlabeled points, we introduce an ``imaginary'' cumulative budget, $\budget^{(t)} = t\budget / n$, which represents the target expected number of labels used by step $t$, if they are spread evenly. We then define the remaining budget at step $t$ as $n_{\Delta}^{(t)} = \budget^{(t)} - n_{\mathrm{lab}}^{(t-1)}$, where $n_{\mathrm{lab}}^{(t)}$ denotes the number of labels collected up to and including step $t$.

Then, for a proposed labeling probability $\pi_t(X_t)$ at step $t$, we enforce the budget and probability constraints by clipping the adjusted probability to $[0,1]$ after applying the remaining-budget cap: $$\tilde{\pi}_t(X_t) = \Big[\min\{\pi_t(X_t), \, n_{\Delta}^{(t)}\}\Big]_{[0,1]}.$$ Because the model $f_t$ is adaptively fine-tuned as labels are collected, the labeling policy is updated accordingly. Algorithm~\ref{alg:seq} in Appendix~\ref{app:algorithms}, adapted from Algorithm~2 of~\cite{active}, gives a step-by-step outline of the resulting procedure for unlabeled features $X_1, \dots, X_n$ considered sequentially. See Figure~\ref{fig:method} for the high-level workflow. 

\section{Real data experiments} \label{sec:experiments}

At a fixed labeling budget $\budget$, we evaluate in four real data experiments when adaptive labeling and prediction-assisted inference yield tighter confidence intervals and improved finite-sample performance relative to the following methods.

\begin{itemize}
    \item \textit{Classical:} estimator constructed using only $\budget$ fully observed samples $(X_i, Y_i)$, with labels collected uniformly at random (i.e., $\pi(x)=\budget/n$).
    \item \textit{Uniform:} prediction-assisted estimator and uncertainty quantification under uniform-at-random labeling, $\pi(x)=\budget/n$.
    \item \textit{Uncertainty-proportional:} prediction-assisted inference with adaptive labeling proportional to an uncertainty score, $\pi(x)\propto u(x)$ (with $u(x)$ chosen based on the estimand and outcome type). This is mixed with the uniform sampling rule so $\pi(x) \leftarrow \tau \pi(x) + (1 - \tau)\budget/n$. We default to sampling proportional to group size and using $\tau = 0.5$ as the mixing proportion, as implemented in Zrnic and Cand\`{e}s \cite{active}, unless otherwise stated.
    \item \textit{Spline-parametrized (\method):} labeling probabilities of the form $\pi(x)=\exp\{\hat\bbeta^\top \mathbf{h}(x)\}$ as in Eq.~\eqref{eq:splineprobs}, where $\hat\bbeta$ solves Eq.~\eqref{eq:optim_full} and $\mathbf{h}(x)$ is a degree-$B$ spline basis expansion of $u(x)$.
\end{itemize}
All approaches other than \textit{Classical} leverage auxiliary unlabeled data to augment the $\budget$ labeled samples. 

\subsection{Performance evaluation} \label{sec:eval}
We measure and compare the performance of the methods on two metrics: coverage and effective sample size (the gain in efficiency by using additional data beyond fully labeled observations). We will use a nominal level $\alpha = 0.1$ unless otherwise noted. To formally define the effective sample size (ESS) $n_{\mathrm{eff}}$, we first consider the classical baseline defined above, where labels are uniformly sampled and no predictions are used. Its effective sample size is $\budget$. For the other labeling policies considered above, the effective sample size is $n_{\mathrm{eff}}$ if the estimator achieves the same variance as the baseline estimator with budget $n_{\mathrm{eff}}$. As an example, if an estimator with corresponding $\budget = 500$ achieves the same variance as the baseline estimator constructed using twice the budgeted labels, then the estimator has $n_{\mathrm{eff}} = 1000$. Equivalently, ESS is a ratio of asymptotic variances, which is precisely the asymptotic relative efficiency; as such, a larger $n_{\mathrm{eff}}$ indicates a more efficient estimator. Details for ESS calculations are provided in Appendix~\ref{app:ESS}.

Since there is randomness in the method, as which data points are labeled is drawn from Bernoulli random variables, we will average over $n_{\text{trials}}$ when presenting results. For each trial, we decide which $\budget$ data points to label, computing the corresponding estimator, asymptotic variance, and coverage based on true odds ratio $\theta^*$. Coverage for a setting is therefore reported as the average coverage over $n_{\text{trials}}$, which should achieve at least the nominal $90\%$ level. In simulation, the true parameter is known so coverage is computed with respect to the true $\theta^*$ which is known by design. In real data experiments, however, the underlying population $\theta^*$ is not known. Taking a finite-sample perspective, we treat the sample as the whole population and compute the underlying $\theta^*$ if all labels were known. Then, the labels are masked for the labeling and uncertainty quantification procedure. 

\subsection{Finite-population calibration for coverage evaluation}\label{sec:overcoverage}

In real-data experiments, the scientific estimand is the corresponding superpopulation parameter $\theta(P)$. However, since this is an unknown quantity, empirical coverage must be evaluated against an observable proxy: the parameter $\theta_N = \theta(P_N)$ computed on the fixed evaluation population. This creates a mismatch between the inferential target used in the asymptotic theory and the target used to measure coverage
in finite-sample experiments. Conditional on the evaluation population, $\theta_N$ is fixed, and the only remaining randomness comes from the label-query indicators. As a result, intervals calibrated to superpopulation variation may be conservative when assessed against $\theta_N$, especially when the queried fraction is non-negligible. This is analogous to the finite-population correction in survey sampling: as the sample fraction increases, variance conditional on the realized finite population can be smaller than the corresponding superpopulation variance.

We therefore distinguish between two intervals in the real-data experiments. The usual Wald interval is the appropriate interval for inference about the superpopulation target $\theta(P)$. The finite-population-calibrated interval below is used only to assess coverage against the empirical benchmark $\theta_N$, and to diagnose how much of the observed overcoverage is due to evaluating against a fixed finite population rather than against the unavailable superpopulation parameter.

We illustrate this effect in the simple case of mean estimation under a uniform policy. Let $F_N=\{(X_i,Y_i)\}_{i=1}^N$ denote the fixed test population, $\theta_N = \frac{1}{N}\sum_{i=1}^N Y_i$, and suppose each point is labeled independently with probability $\pi_i := \pi = \budget / N$. Consider the debiased estimator from Eq.~\eqref{eq:meanest}, 
$$\hat\theta=\frac{1}{N}\sum_{i=1}^N\left[f(X_i)+\frac{\xi_i}{\pi}\{Y_i-f(X_i)\}\right].$$ Then, defining $R_i = Y_i - f(X_i)$, conditional on the observed data $F_N = \{(X_i, Y_i)\}_{i=1}^N$ (fixed test population), $$\hat{\theta} - \theta_N = \frac{1}{N}\sum_{i=1}^N\left(\frac{\xi_i}{\pi} - 1\right)R_i,$$ so the conditional finite-population variance is $$V_{\text{act}} := \Var(\hat\theta-\theta_N \mid F_N) =\frac{1}{N^2}\sum_{i=1}^N \frac{1-\pi}{\pi}\,R_i^2.$$ In contrast, the superpopulation-based Wald interval uses a variance estimate
$V_{\mathrm{int}}=\hat\sigma^2/N$, where $\hat\sigma^2$ is the sample variance of the summands in $\hat\theta$. To quantify this discrepancy, we define the random inflation factor $\gamma:= V_{\mathrm{int}} / V_{\mathrm{act}}$.
Then, conditional on $(F_N,\gamma)$ and under the normal approximation,
$$\text{covg}:=\P\!\left(\theta_N\in[\hat\theta\pm z_{1-\alpha/2}\sqrt{V_{\mathrm{int}}}]\,\middle|\,F_N\right) \approx 2\Phi\!\left(\sqrt{\gamma}\,z_{1-\alpha/2}\right)-1,$$
so the overcoverage amount $\Delta:=\text{covg}-(1-\alpha)$ depends on the magnitude of $\gamma$ (and is naturally upper bounded by $\alpha$). This phenomenon extends to OPAL and other non-uniform policies $\pi_i=\pi(X_i)$ by replacing $\pi$ with $\pi_i$.

Since we do not observe every $R_i$, we can estimate $V_{\mathrm{act}}$ by $$\hat V_{\mathrm{HT}} := \frac{1}{N^2}\sum_{i=1}^N \xi_i\,\frac{1-\pi_i}{\pi_i^2}\,R_i^2, \; \text{such that } \; \E[\hat V_{\mathrm{HT}}\mid F_N]=V_{\mathrm{act}}.$$ When reporting coverage for the benchmark $\theta_N$, we also report the finite-population-calibrated
interval
$$\hat\theta\pm z_{1-\alpha/2}\sqrt{\hat V_{\mathrm{HT}}}.$$
For the odds-ratio experiments, we apply the same calculation groupwise and combine the resulting variance estimates
through the Delta method. Detailed derivations are provided in Appendix~\ref{app:overcoverage}. 

\subsection{Using chest X-rays for predicting cardiomegaly} \label{sec:CHX}

We begin with a more in-depth look at the cardiomegaly odds ratio estimation problem aided by predictions using chest X-ray data first introduced in Section~\ref{sec:CHX_intro}. We use data from CheXpert, a large chest radiograph dataset comprising 65,240 patients who received care at Stanford Health Care between October 2002 and July 2017, including both inpatient and outpatient centers \cite{CheXpert}. The dataset includes additional  features such as patient age, sex, and the orientation of the X-ray image (anterior-posterior or the opposite), as well as several response variables including cardiomegaly -- an enlarged heart, where the heart muscle becomes larger than usual \cite{cardiomegaly}. Cardiomegaly is a radiographic sign of underlying pathology, generally viewed as a sign of cardiac dysfunction such as heart failure. It can also arise transiently under short-term physiologic stress (e.g., pregnancy) and occurs across ages.
In this application, we predict the odds ratio of cardiomegaly among younger and older patients, where we set the cut-off for old patients to be those above 40 years old. We curate the CheXpert dataset to include all patients with cardiomegaly status available, resulting in a total of 19,596 samples which are summarized in Table~\ref{tab:cardiomegaly_age_totals}.

\begin{table}[ht]
\centering
\caption{Contingency table of age vs.\ Cardiomegaly status (with totals)}
\label{tab:cardiomegaly_age_totals}
\begin{tabular}{lrrr}
\toprule
\textbf{Age} & \textbf{No cardiomegaly} & \textbf{Cardiomegaly} & \textbf{Total} \\
\midrule
$< 40$ & 3488 & 575  & 4063 \\
$\geq 40$ & 9576 & 5957 & 15533 \\
\midrule
\textbf{Total} & 13064 & 6532 & 19596 \\
\bottomrule
\end{tabular}
\end{table}

The true odds ratio in this dataset is $\theta = 0.265$, so the odds of an older patient having cardiomegaly is nearly 4 times that of a patient below 40 years old. That is, in the CheXpert data, cardiomegaly is diagnosed much more in older patients than younger ones. As the presence of cardiomegaly is indicative of another underlying condition, it may be the case that older patients tend to be sicker and thus will see higher prevalence of an enlarged heart. We see, however, that the number of younger patients is vastly outnumbered (comprising about 20\% of the total patients), so predictive algorithms may be biased towards correct predictions in the larger, above 40 population. 

To predict cardiomegaly based on X-ray images, we use a pre-trained model based on DenseNet-121 with weights learned from the official CheXpert model, available in the \texttt{TorchXRayVision} library \cite{TorchXRayVision, TorchXRayVision2}. Details for the model, which was published with the CheXpert dataset are available in \cite{CheXpert-paper}. The library also contains weights for several other models, including those trained on other large-scale chest X-ray datasets such as MIMIC-CXR \cite{mimic_cxr}, which could be explored as alternative prediction models. We treat the classifier as fixed and focus on inference under a labeling budget, not on optimizing predictive accuracy. 

\begin{figure}[h]
    \centering
    \includegraphics[width = \linewidth]{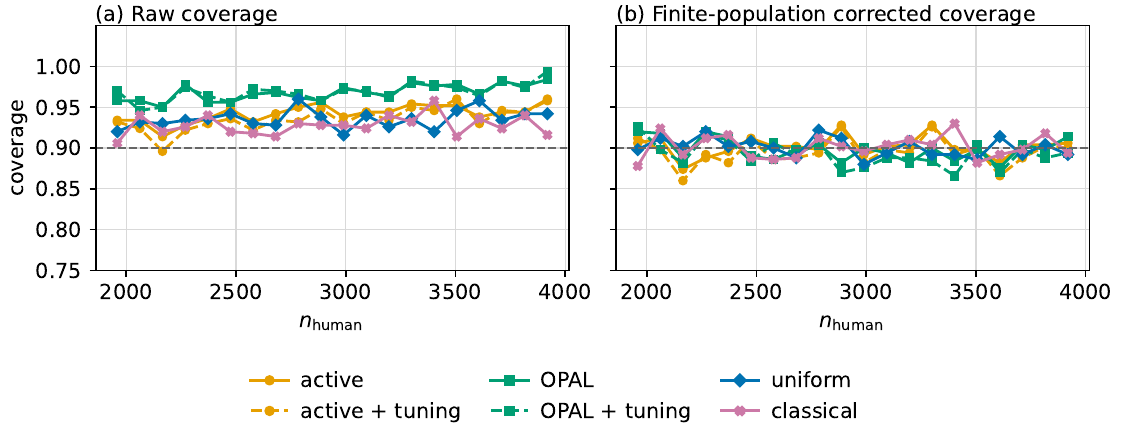}
    \caption{\textbf{Coverage for odds ratio estimation of cardiomegaly in patients below vs.~over 40 years of age}  (a) usual Monte Carlo coverage; (b) coverage after finite-population calibration, detailed in Section~\ref{sec:overcoverage}, accounting for the fact that inference is evaluated against the fixed empirical population rather than an independent superpopulation draw. The dashed horizontal line indicates the nominal 90\% target. Results are average over 500 Monte Carlo trials. The budget given on the x-axis (denoted by the number of labels acquired, $n_{\text{human}}$) ranges from 10\% to 20\% of the total unlabeled observations.}
    \label{fig:CHX_covg}
\end{figure}

Details on model fitting and performance are provided in Appendix \ref{app:CHX_details}, including the confusion matrix and other diagnostic statistics for each of the two subgroups. We compare the methods summarized in Section~\ref{sec:eval} and construct confidence intervals for the odds ratio, shown in Figure~\ref{fig:CHX_ESS}. All methods meet the set coverage level of 90\%, as seen in coverage plots provided in Figure~\ref{fig:CHX_covg}. We see that by adjusting for the finite-population setting, as discussed in Section~\ref{sec:overcoverage},
the overcoverage is corrected (Figure~\ref{fig:CHX_covg}(a) vs (b)).

In this example (Figure~\ref{fig:CHX_covg}), we see that collecting labels at random (blue) improves in effective sample size over classical (pink), showing that augmenting with predicted values can result in tighter intervals. By collecting labels in a data-adaptive manner proportional to uncertainty (orange), there is an additional boost in effective sample size. Specifying the labeling policy to explicitly minimize estimated asymptotic variance via \method~(green) results in the largest gain in effective sample size. In addition to increasing effective sample size, \method~also has more stable performance. As seen in Figure~\ref{fig:CHX_var}, the estimates, interval widths, and endpoints obtained using~\method~show much less variability between the 500 instantiations of the sampling procedure compared to the other methods. This indicates that using \method~to obtain labels and perform inference is more stable---if the procedure were done several times, the resulting intervals and estimates from each trial do not vary too much.

\begin{figure}[h]
    \centering
    \includegraphics[width = \linewidth]{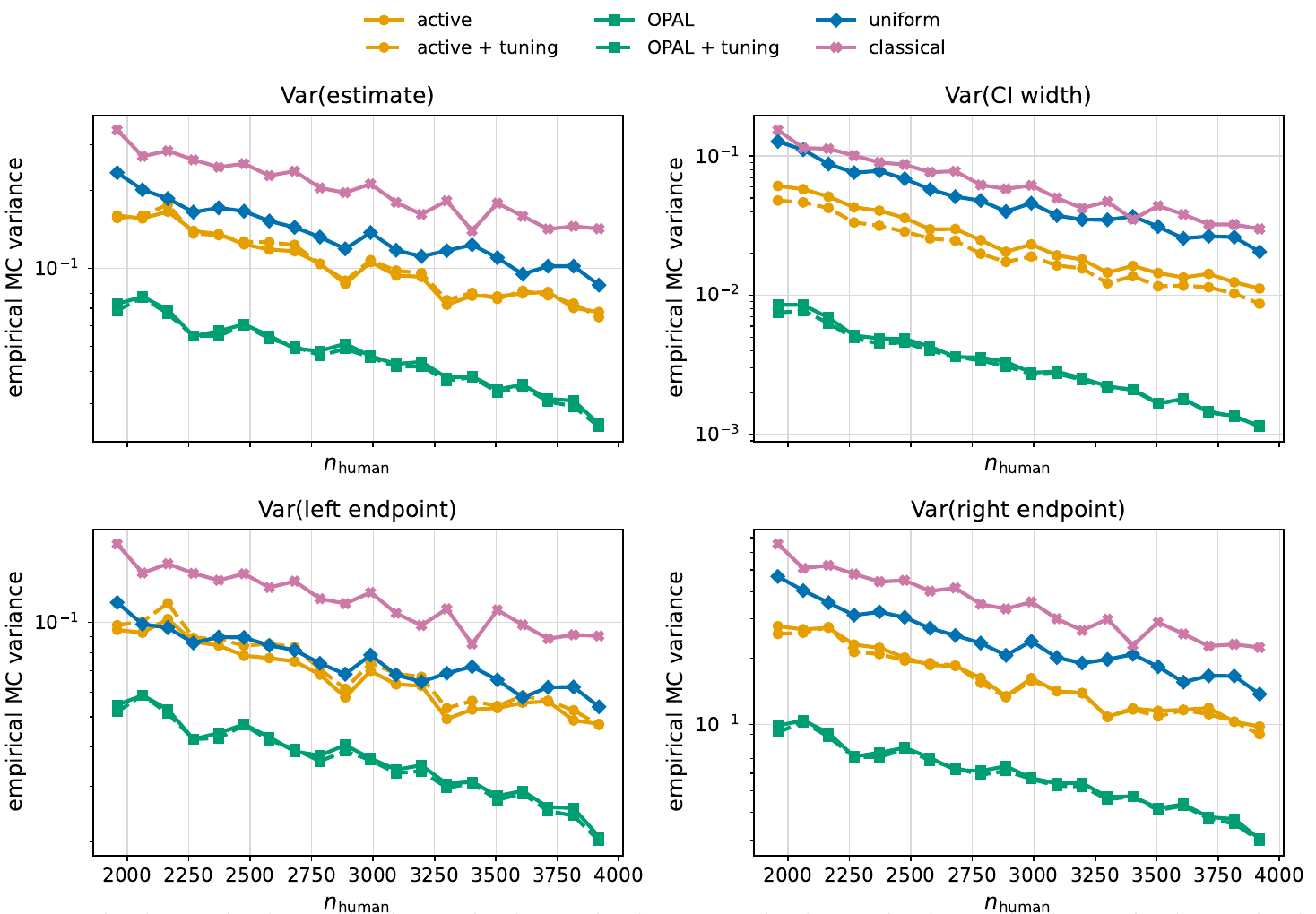}
    \caption{\textbf{Stability of odds ratio estimation of cardiomegaly in patients below vs. over 40 years of age.} Variability of estimates, interval widths, left and right endpoints over 500 Monte Carlo trials. The budget given on the x-axis (denoted by the number of labels acquired, $n_{\text{human}}$) ranges from 10\% to 20\% of the total unlabeled observations.}
    \label{fig:CHX_var}
\end{figure}

\subsection{Using imaging data for breast cancer subtype diagnosis} \label{sec:BRCA}

Here, we are interested in estimating the odds ratio of triple negative breast cancer (TNBC) among African American and White women using histopathological imaging data from The Cancer Genome Atlas (TCGA). We define our two subgroups as patients identifying as white (W) and black (B). Using OpenSlide to download and process the images, we curate a dataset of 932 samples which are summarized in Table~\ref{tab:race_tnbc_totals}. Unlike the CheXPert dataset, which is large and has a stronger pretrained model, this task has noisier predictions and sharper subgroup imbalance, making budget allocation more delicate.
\begin{table}[ht]
\centering
\caption{Contingency table of race vs.\ TNBC status (with totals)}
\label{tab:race_tnbc_totals}
\begin{tabular}{lrrr}
\toprule
\textbf{Race} & \textbf{Non-TNBC} & \textbf{TNBC} & \textbf{Total} \\
\midrule
B & 129 & 53  & 182 \\
W & 644 & 106 & 750 \\
\midrule
\textbf{Total} & 773 & 159 & 932 \\
\bottomrule
\end{tabular}
\end{table}

The true odds ratio of triple negative breast cancer is $\theta = 2.5$. That is, the odds of an African American breast cancer patient having TNBC is 2.5 times that of a Caucasian breast cancer patient. As Table~\ref{tab:race_tnbc_totals} shows, while TNBC is more common among Black patients, but Black patients are underrepresented in the dataset; this combination makes accurate subgroup inference both important and challenging.
We split the data into two parts, taking 20\% ($n_{\text{train}} = 186$) to train a classifier and reserving the remaining 80\% ($n_{\text{test}} = 746$) for inference. There are intricacies to training a classifier for histopathological images, particularly for a difficult task like labeling subtypes of breast cancer (compared to for example, diagnosing disease or no disease). Image segmentation and extracting relevant parts of these massive images is an active area of research, which is not the focus of this paper. Our classifier, based on finetuning ResNet-18 (with more details provided in Appendix \ref{app:BRCA_details}), achieves an accuracy of about 80\% on the test set; however, due to the large group imbalance and positive/negative label imbalance, minority group recall is quite poor and there are a large fraction of false negatives. We however see in Figure~\ref{fig:BRCA} that after applying each of the methods to estimate the odds ratio of interest, OPAL still outperforms  other methods, particularly with slightly larger budget sizes. In fact, we see that proportional-to-uncertainty labeling has ESS similar to that of the actual number of human labels for all budget sizes. 

\begin{figure}
    \centering
    \includegraphics[width = \linewidth]{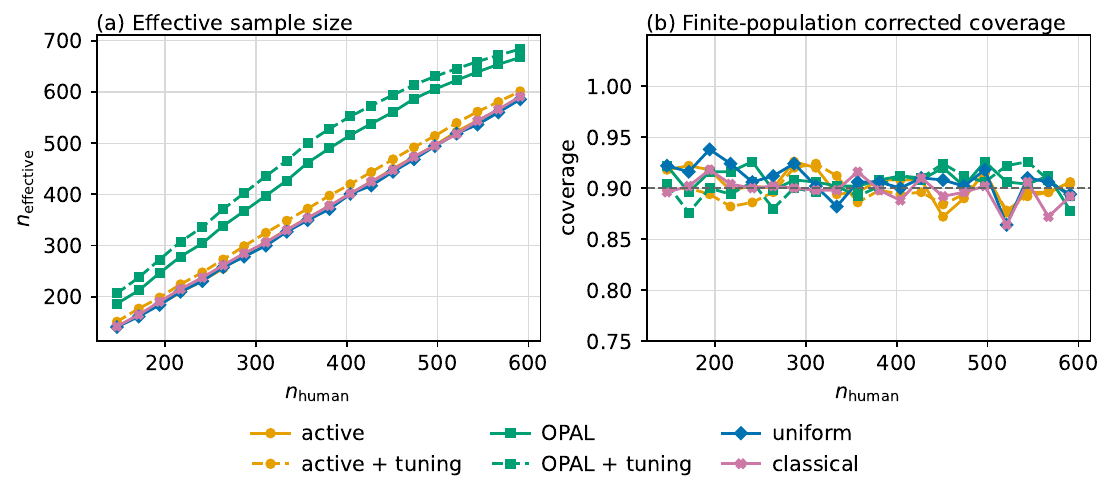}
    \caption{\textbf{Odds ratio estimation of triple negative breast cancer in Caucasian vs.~African American women} (a) effective sample size of each method where solid line denotes baseline and dashed denotes with power tuning (for active proportional-to-uncertainty labeling and active spline-parametrized optimal labeling); (b) coverage of each method, with correction to adjust for finite-population effects. We perform 500 trials per method at each budget level  (20-50\%), and average over these trials in the reported results.}
    \label{fig:BRCA}
\end{figure}

If we were to change the way in which the budget is allocated to the two different groups (default set to proportional to group size following \cite{active}), it is possible to see improvements in the ESS for the proportional-to-uncertainty policy, however, this is something we must determine a-priori, not after seeing results. This application shows that even with a classifier that exhibits non-ideal behavior, substantial gains in sample size can still be seen by incorporating predicted labels according to our method. power tuning increases the ESS for both active sampling methods by about the same amount. All methods meet the set coverage level of 90\%.

\subsection{Using affirming-device indicators as predictors for global-warming stance}

Next, we study an application in computational social science by examining stance annotations to study linguistic differences between media affirming or denying global warming \cite{stance}. The study consists of 2,300 news headlines which agree, are neutral, or disagree with the claim that global warming is a serious concern. The quantity of interest is the odds ratio of agreement given the presence of affirming devices like ``expert'' and ''renowned.'' That is, if $Z_{\text{affirm}} \in \{0,1\}$ is the indicator of an affirming device being present and $Y_{\text{agree}} \in \{0,1\}$ is the indicator of agreement (with the stance about global warming being a serious threat), then $$\theta^* = \frac{\mu_{\text{agree} \mid \text{affirm}}/(1 - \mu_{\text{agree} \mid \text{affirm}})}{\mu_{\text{agree} \mid \text{no affirm}} / (1 - \mu_{\text{agree} \mid \text{no affirm}})},$$ where $\mu_{\text{agree} \mid \text{affirm}} = \P(Y_{\text{agree}} = 1 \mid Z_{\text{affirm}} = 1)$ and $\mu_{\text{agree} \mid \text{no affirm}} = \P(Y_{\text{agree}} = 1 \mid Z_{\text{affirm}} = 0)$. We use data from Gligori\'{c} et al. \cite{llms}, which queries the LLM for verbalized confidence via a two-stage prompting procedure: the model is first asked to provide an answer and then to assign a probability $C_i$ to the correctness of the answer. We can then train a black-box predictor of LLM error via confidence score $C_i$ and the difference between predicted and true label. To perform batch inference, we reserve 200 data points to train a predictor of the LLM error and then proceed. In the sequential setting, we collect a burn-in sample by labeling the first $n_{\text{burnin}}$ data points, which are used to obtain an initial predictor of the LLM error. As more labels $Y_i$ are collected, we use $\{(C_j, (\hat{Y}_j - Y_j)^2)\}_{j < t, \xi_j = 1}$ to finetune the predictor. Budgets ranging from 20\% to 50\% of all unlabeled data are considered. We perform 500 trials per method at each budget level, and average over these trials in the reported results.

\begin{figure}
    \centering
    \includegraphics[width = \linewidth]{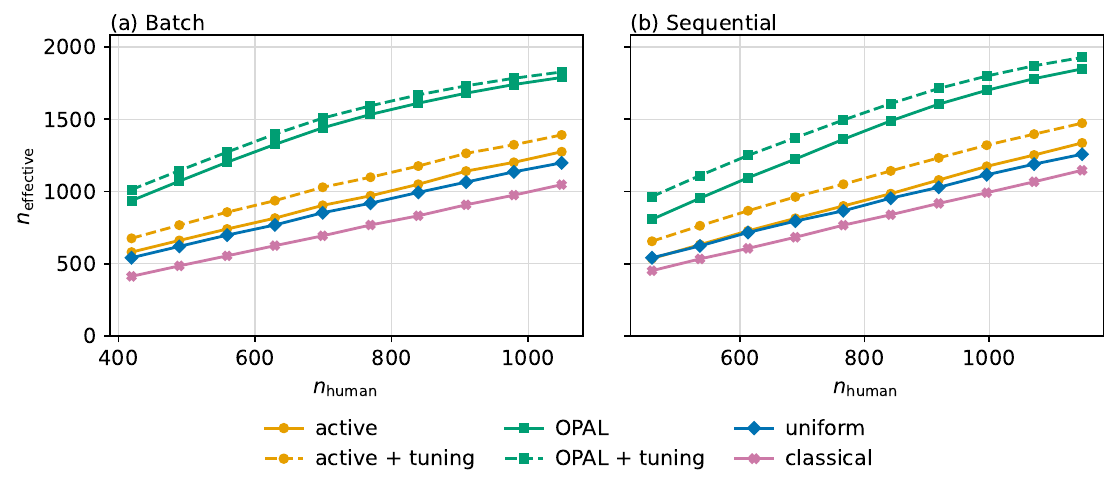}
    \caption{\textbf{Odds ratio estimation of global warming stance with affirming devices} Effective sample size of each method under (a) batch sampling and (b) sequential sampling. We perform 500 trials per method at each budget level (20-50\%), and average over these trials in the reported results.}
    \label{fig:stance-ESS-combined}
\end{figure}

In Figure~\ref{fig:stance-ESS-combined}, we compare our method to each of the other labeling policies considered, in both batch (a) and sequential (b) settings. Sequential sampling is implemented with finetuning performed after every $B = 100$ collected labels. Here, finetuning refers to both finetuning of the predictive model and also of the labeling policy to account for the newly acquired labels.  All three methods which use predictions in addition to collected labels showed improvement in ESS when using sequential over batch sampling, but the ordering remains the same. Power tuning increases ESS for both \method~and proportional-to-uncertainty labeling, in both the sequential and batch settings. Notably, using \method~in the batch setting outperforms proportional-to-uncertainty labeling even in the sequential setting. All methods achieve 90\% coverage (with finite-population correction), which is shown in Figure~\ref{fig:stance-coverage-combined}.

Further, \method~is more stable, as shown in Figure~\ref{fig:stance-sequential-dist}: the distributions of ESS resulting from the intervals constructed across 500 iterations for each method are plotted. Those for \method~are less variable compared to the others and clearly result in higher ESS values.

\subsection{Using AlphaFold-derived predictors for intrinsic disorder prediction}

In our last experiment, we consider an application to odds ratio estimation of a protein being phosphorylated, a functional property, and the protein coming from an intrinsically disordered region (IDR), a structural property. Learning the protein structure requires expensive experimental techniques to measure accurately, and  AlphaFold, a machine learning model that predicts a protein’s structure from its amino acid sequence, has been used as a cheaper alternative \cite{active, bludau, ppi, ppi++}. In the context of label budget-constrained inference, Zrnic and Cand\`{e}s \cite{active} show that by selectively picking which $\budget$ protein structures to measure experimentally, it is possible to leverage imperfect predictions from AlphaFold and produce valid confidence intervals which are tighter than ones constructed using uniform sampling, where the labeled data points are chosen at random. We also use the post-processed AlphaFold outputs made available by Angelopoulos et al. \cite{ppi} and predict the IDR structural property based on the raw AlphaFold outputs. There are 10802 total observations of ground-truth IDR indicators. While in \cite{ppi, active}, the logistic regression model trained to predict disorder is only trained on AlphaFold outputs, we instead opt to train the model with the phosphorylation indicator incorporated as well, since it is possible this indicator is a relevant covariate to incorporate in building a predictor. 

The odds ratio estimand is $\theta = (\mu_1/(1 - \mu_1))/(\mu_0(1- \mu_0))$ where $\mu_1 = \P(Y = 1 \mid Z_\text{ph} = 1)$ and $\mu_0 = \P(Y = 1 \mid Z_{\text{ph}} = 0)$ with $Y \in \{0,1\}$ being the indicator of disorder and $Z_{\text{ph}} \in \{0,1\}$ the indicator of phosphorylation. We use the estimator in Eq.~\eqref{eq:estimator} and the resulting confidence interval from Corollary~\ref{corr:batch_odds}, based on the labeling policies outlined in Section~\ref{sec:simulation}. We note that for both proportional-to-uncertainty labeling and spline-optimization-based labeling, results are nearly identical with and without power tuning. We use 500 training points from each group (phosphorylated and no phosphorylation) for training a logistic regression classifier, resulting in a test set of 9802 data points. For the purpose of evaluating coverage, we take the odds ratio among the test data as the ground truth $\theta^*$. We consider budgets ranging from 2\% to 20\% of data. 

\begin{figure}[h]
    \centering
    \includegraphics[width = \linewidth]{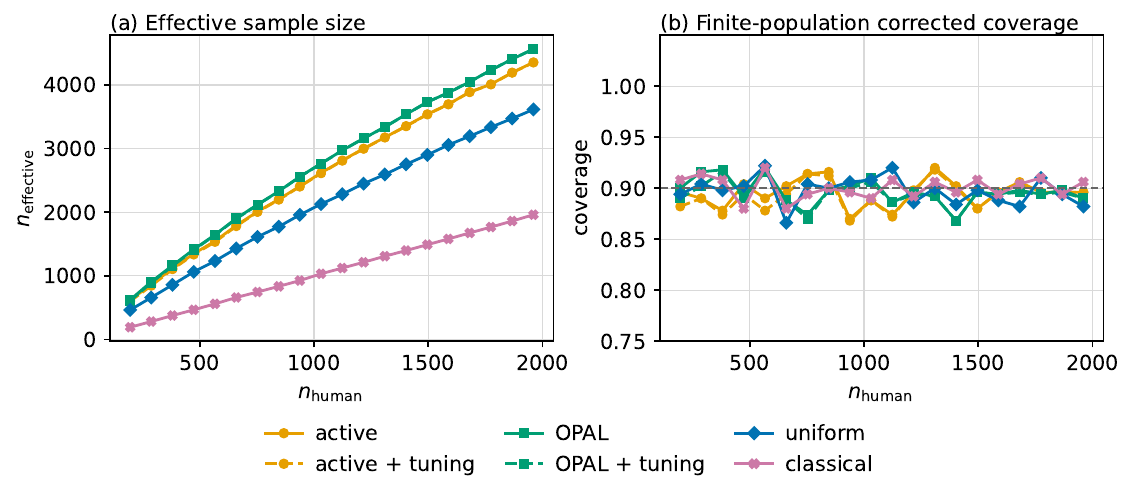}
    \caption{\textbf{Odds ratio estimation of intrinsic disorder using AlphaFold-derived predictors} (a) effective sample size of each method where solid line denotes baseline and dashed denotes with power tuning (for active proportional-to-uncertainty labeling and active spline-parametrized optimal labeling); (b) coverage of each method, with correction to adjust for finite-population effects. We perform 500 trials per method at each budget level (2-20\%), and average over these trials in the reported results.}
    \label{fig:alphafold}
\end{figure}

In Figure~\ref{fig:alphafold}, we compare our method to proportional-to-uncertainty, uniform, and classical labeling policies. The classical method does not always achieve 90\% coverage. For small budgets, the two active sampling methods sometimes undercover, though spline-parametrized less-so. Once the budget is at least 500 collected labels, then coverage is no longer a concern for any of the methods other than classical, and we see that spline-parametrized almost consistently has higher coverage than the others. We see that as in the previous examples, OPAL has the largest effective sample size, though closer to that of proportional-to-uncertainty compared to other applications. We note that in the AlphaFold example, the two groups have much more balanced representation compared to the previous two applications, which might explain part of this. Further, power tuning does not have much impact for either of the active sampling methods, a sign that the predictions are relatively reliable overall. In general, this experiment shows the power of using predicted labels in addition to collected ones, as all prediction-boosted methods have ESS at least twice that of the number of collected labels.

\section{Simulation study} \label{sec:simulation}
To isolate the effects of an unreliable predictor $f$ and the smooth parametrization of the labeling policy on estimator efficiency, we consider several simulation settings. We perform sensitivity analyses to show the robustness of \method~to the choice of number of splines and degrees of freedom. We also implement the end-to-end inference framework for estimation and uncertainty quantification of a Kendall's $\tau$ estimator to demonstrate the use on a non M-estimator quantity.

\subsection{Odds ratio estimation} \label{sec:sim_odds}

We present experiments assessing OPAL's capabilities in an odds ratio estimation setting. We consider a setting where our two groups have different underlying covariate distributions and the outcomes conditional on features have distinct distributions (both logistic relationships). In detail, we let $Z \sim \Bern(p_z)$; $X|Z = 0 \sim \frac{1}{3}\N(-2, 0.5) + \frac{1}{3}\N(0, 0.5) + \frac{1}{3}\N(2, 0.5)$ and $X|Z = 1 \sim \N(0,1)$; $\P(Y = 1 \mid X = x, Z = 1) = 1/(1 + \exp(-1.5 x))$ and $\P(Y = 1 \mid X = x, Z = 0) = \text{clip}(0.3 + 0.2/(1 + \exp(-x)), 0.05, 0.6)$. We consider an imbalanced setting where $p_z = 0.8$ and a balanced one where $p_z = 0.5$. 

To disentangle the effects of a good model (the predicted label/probability $f(X_i)$ being close to true $Y_i$) and the effect of smoothing labeling probabilities via spline parameterization, we consider two ways of specifying uncertainties: 

\begin{enumerate}
    \item Oracle $\pi$: correctly specified $\pi$ (true uncertainties are known) using the true conditional expectations.
    \item Misspecified $\pi$: uncertainties are estimated using the predicted probabilities from model $f$ (concretely, logistic regression with corresponding labeling probabilities estimated).
\end{enumerate}

For our experiments, we generate 10,000 data points total. Then, we randomly split the data into a training portion ($n_{\text{train}} = 200$) and the remaining for inference ($n_{\text{test}} = 9,800$) on the odds ratio. We train a predictive model using the training data $(X_i, Y_i)$, which we then treat as a fixed, black-box model $f$. Then, we proceed with defining uncertainty scores $u(x)$ and deriving corresponding labeling policies $\pi^{\text{policy}}$ where policy corresponds to uniform, active, and spline (as defined above, respectively). In this example, we fit $f$ using a logistic regression model. 

\begin{figure}[h]
    \centering
    \includegraphics[width=\linewidth]{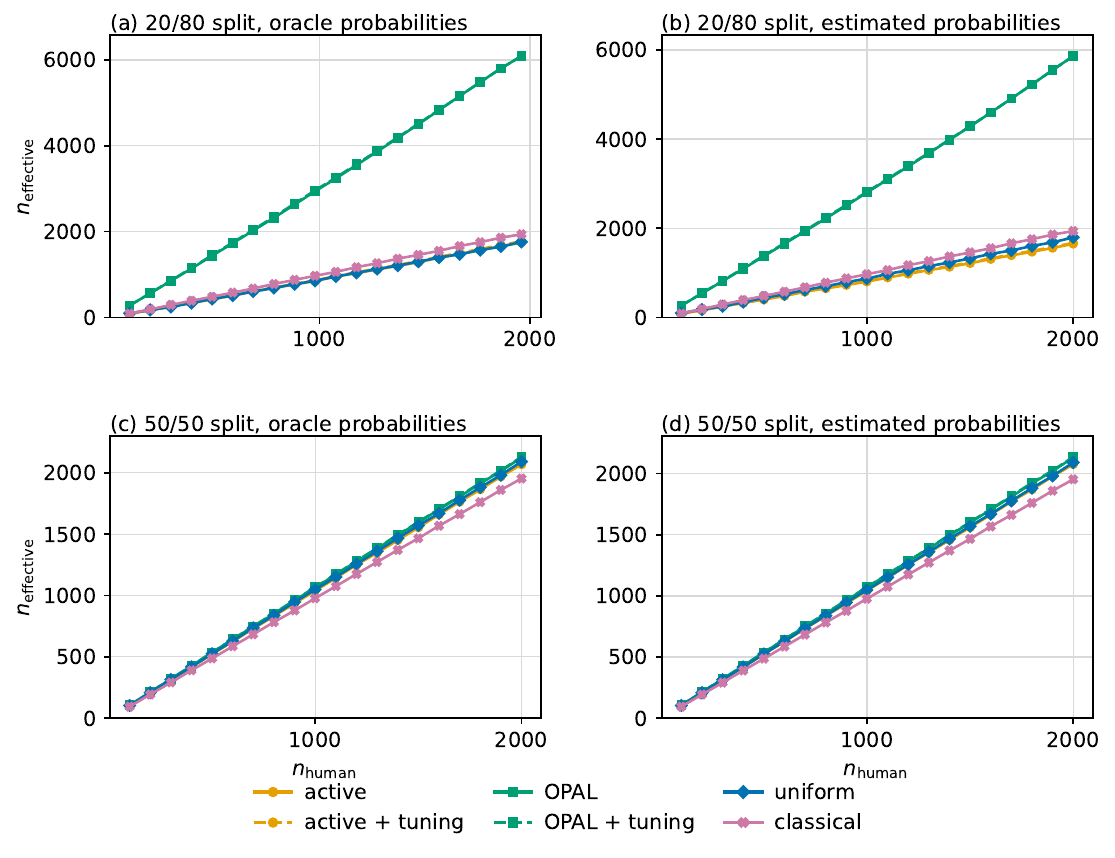}
    \caption{Effective sample size in the unbalanced group size setting with (a) oracle uncertainties, (b) estimated uncertainties; and the balanced group size setting with (c) oracle uncertainties, (d) estimated uncertainties. We perform 500 trials per method at each budget level (1-20\%), and average over these trials in the reported results.}
    \label{fig:sim_ESS}
\end{figure}

We consider labeling budgets ranging from 1\% to 20\% of the test data. The ESS plots for the imbalanced setting, where 20\% of data comes from group 0 (which is relatively ``harder'' due to noisy signal) and 80\% from group 1 (which is easier since it has a straightforward sigmoid relationship), are provided in Figure~\ref{fig:sim_ESS}(a)--(b). Results from the balanced setting are reported in Figure~\ref{fig:sim_ESS}(c)--(d). 

In both settings, OPAL matches or outperforms uniform and proportional-to-uncertainty labeling policies consistently. It also demonstrates regularization as OPAL's effective sample size is larger. In both, oracle probabilities result in larger ESS, though this is much more pronounced in the imbalanced group case: for budget 20\% (i.e $n_{\text{labeled}} = 2000$), the ESS for our method is 3500 when using oracle uncertainties to parametrize smooth labeling probabilities versus 2750 using the estimated uncertainties. 

These results also show that \method~is more robust to misspecified predictions and uncertainty score, particularly in the imbalanced setting of Figure~\ref{fig:sim_ESS}(b), where we see that it has large ESS relative to all other methods when using estimated probabilities. On the other hand, we see that active and uniform both produce wider intervals than classical, indicating that a poor prediction method has downstream implications that affect efficiency.

We further examine whether the adaptive policies benefit from explicit uniform exploration in Appendix~\ref{app:mixing}. We tune a mixture between each adaptive policy and uniform sampling. The selected mixing weights, based on minimizing a variance measure, show that active uncertainty sampling often requires substantial uniform mixing, whereas OPAL typically selects the unmixed or nearly unmixed policy (see Figure~\ref{fig:sim-lambdatune}). This provides another diagnostic that OPAL's variance-oriented smooth parameterization already regularizes the labeling probabilities in a way that improves stability for inference.

Another notable observation is OPAL's much larger ESS in the imbalanced case compared to when the two groups are approximately equal in size. This is due to the fact that ESS is a ratio metric. When $p_z$ is large (so the simpler group 1 is overrepresented), then the intervals generated with only labeled data tend to be wider and variance scales approximately by $5/n$ assuming the easier group can be learned perfectly (scaling factor determined by $100/20$ the fraction of the harder group), since the classical method labels $\budget$ points proportional to group size. As ESS is a ratio of classical variance over method variance, consequently the scale is inherently in favor of our method, since it adapts to the specific data and subgroup structure while the classical method does not and thus the variance scales as e.g.~$5/n$ vs.~$2/n$ for the two different settings. This is less dramatic in practice, since group 0 is non-trivial to learn. We see this phenomenon in Figure~\ref{fig:sim_0.95_ESS} where more extreme imbalance (95\% group 1, 5\% group 0) results in even more inflation of effective sample size. 
 
\subsection{Kendall's tau}

Now we present a simulation study on estimation of Kendall's tau, a non-parametric rank-based quantity. We generate features $X \sim \N(0, I) \in \R^d$ where $d = 12$ and binary outcomes $Y \in \{0,1\}$ according to the following data generating process:
\begin{enumerate}
    \item For a single $X$, define the latent scalar score $S^*(X) = 1.1 X_0 - 0.9 X_1 + 0.7 X_2 + 0.6 \sin(X_3) - 0.6 X_4X_5 + 0.41\{X_6 > 0\} - 0.4 X_71\{X_8 > 0\}$ (a non-linear function of $X$).
    \item Let $S_{\text{obs}}$ be the noisy, observed version of $S^*$: $S_{\text{obs}}(X) = S^*(X) + \varepsilon$ where $\varepsilon \sim \N(0, \sigma^2)$, $\sigma^2 = 0.5$.
    \item Generate binary $Y$ according to $\operatorname{logit} \P(Y = 1 \mid X) = -0.3 + S^*(X) + 0.4\sin(X_0) - 0.5 X_2X_3$.
\end{enumerate}
We generate $n = 20,000$ data points $(X_i, Y_i, S_i)$, where $S_i$ represents the noisy score. Taking a random split of 10\% of the data, we fit a logistic regression on $(X_i, Y_i)$ pairs to learn $\hat{\mu}(x) = \P(Y = 1 \mid X)$. The remaining data is reserved for inference. Recall the form of the estimator, $$\psi = \frac{\# \text{concordant pairs} - \# \text{discordant pairs}}{\binom{n}{2}} = \frac{2}{n(n-1)}\sum_{i < j}\text{sign}(S_i - S_j)\text{sign}(Y_i - Y_j),$$
where we replace $X_i$ from the definition in Section~\ref{sec:examples} Example 5 with the score $S_i$, since the features $X_i \in \R^d$ are multi-dimensional but Kendall's tau is a measure of correlation between scalars. From the example, we know the form of the EIF and relevant quantities. As such, the nuisance quantities which need to be estimated are 
$$A_i = \E_j[\text{sign}(S_i - S_j)], \; B_i = \E_j[\text{sign}(S_i - S_j)\mu_j].$$ Estimates are computed with a single stable sort of $S$:
$$\hat{A}_i = 2F_{\text{mid}}(S_i) - 1, \; \hat{B}_i = \frac{1}{n}\sum_j \text{sign}(S_i - S_j) \hat{\mu}_j,$$
where $F_{\text{mid}}(\cdot)$ is the mid-rank CDF (used to handle ties). The corresponding estimator is $$\hat{\tau} = \frac{1}{n}\sum_i \bigg[(\hat{\mu}_i A_i - B_i) + \frac{\xi_i}{\pi(X_i)}(Y_i - \hat{\mu}_i)A_i \bigg].$$ The influence values $\phi_i$ are the corresponding inner terms of the above average subtracting off $\hat{\tau}$, which are used to construct a sample variance and form a confidence interval. We use a variance proxy of $c_i = \hat{A}_i^2 \hat{\mu}_i(1 - \hat{\mu}_i)$ as a measure of uncertainty which is used for the non-uniform sampling methods. The true parameter value $\tau^*$ is computed from the test split of the data. While Kendall's $\tau$ is not directly covered by existing active inference methods \cite{active}, since it is not a mean/standard M-estimator, we implement a proportional-to-uncertainty baseline by using the specific influence-function variance contribution derived in Section~\ref{sec:examples} Example 5 as the uncertainty score.
We see in Figure~\ref{fig:kendall_ESS} that the effective sample size for our spline-based method outperforms both the uniform and direct proportional-to-uncertainty labeling policies.

\begin{figure}[h]
    \centering
    \includegraphics[width=0.5\linewidth]{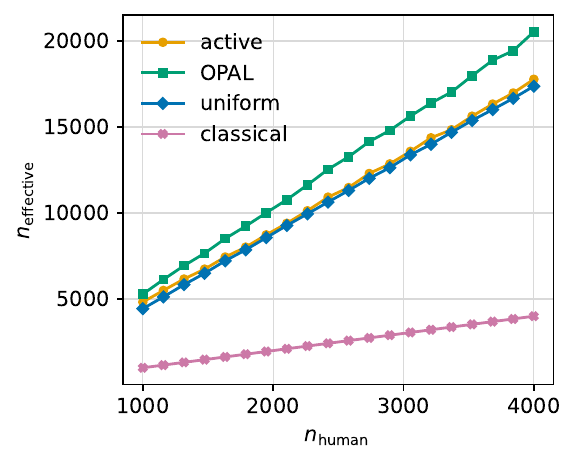}
    \caption{Effective sample size for Kendall's Tau simulation. We perform 500 trials per method at each budget level (1-20\%), and average over these trials in the reported results.}
    \label{fig:kendall_ESS}
\end{figure}

\section{Discussion} \label{sec:discussion}
We conclude this paper by discussing extensions and future directions, including joint inference for multiple target estimands, overlapping subgroups, and incorporating additional data sources.
 
\subsection{Estimating multiple quantities related to more than two subgroups}\label{sec:more_groups}

Suppose we have $k>2$ subgroups and wish to estimate multiple odds ratios jointly. For example, when $k=3$ we may target
$\theta=\mathrm{odds}_1/\mathrm{odds}_2$ and $\psi=\mathrm{odds}_1/\mathrm{odds}_3$, which depend on the shared mean
$\mu_1=\P(Y=1\mid X\in G_1)$ (and $\mu_2,\mu_3$). A naive approach would run \method\ separately for each odds ratio, but this duplicates effort by re-estimating shared quantities. A more efficient direction is to allocate the labeling budget jointly across groups while maintaining valid inference for all targets.

Using PPI-style estimators for $(\mu_1,\mu_2,\mu_3)$ with group-specific policies $(\pi^1,\pi^2,\pi^3)$, one can establish a joint CLT and hence
$$\sqrt{n}\Big( \begin{bmatrix}\hat\theta\\ \hat\psi\end{bmatrix}- \begin{bmatrix}\theta\\ \psi\end{bmatrix} \Big)\ \Rightarrow\ \N(0,\Sigma(\pi^1,\pi^2,\pi^3)).$$
This suggests optimizing policies using a scalar measure of joint uncertainty rather than treating each variance separately. Natural choices include minimizing a weighted trace (e.g., $\mathrm{tr}(W\Sigma)$) or the generalized variance $\det(\Sigma)$, which accounts for correlation between $\hat\theta$ and $\hat\psi$. Developing conditions under which such objectives remain tractable under our spline-parametrized policy class is a promising direction for future work.

\subsection{Overlapping subgroups}
We consider estimation of the log-odds ratio between two overlapping subgroups $G_1$ and $G_0$; for example,
$G_1=\{20\le \text{age}\le 35\}$ and $G_0=\{30\le \text{age}\le 45\}$.
Given a labeling policy $\pi:\mathcal{X}\to[0,1]$, define the group-specific mean estimators
$$\hat{\mu}_k=\frac{1}{n_k}\sum_{i=1}^n 1\{X_i\in G_k\}\Big[f(X_i)+\frac{\xi_i}{\pi(X_i)}\big(Y_i-f(X_i)\big)\Big], \qquad \xi_i\sim \Bern(\pi(X_i)),$$
where $n_k=\sum_{i=1}^n 1\{X_i\in G_k\}$. In the ``single-policy'' setting, we use a shared policy of the form
$$\pi(X_i)=\exp\Big\{-\sum_{b=1}^B \beta_b h_b(U_i)\Big\},$$
where $U_i:=u(X_i)$ denotes an uncertainty score and $\{h_b\}_{b=1}^B$ is a spline basis shared between the two groups.

For the log-odds ratio $\log\theta=h(\mu_1,\mu_0)=\log\frac{\mu_1}{1-\mu_1}-\log\frac{\mu_0}{1-\mu_0}$, we have the following empirical variance objective
$$\Var(\hat{\theta})=g^\top \Sigma(\pi) g =\frac{1}{n}\E\!\left[\frac{v(X)}{\pi(X)}\Big(g_1 1\{X\in G_1\}+g_0 1\{X\in G_0\}\Big)^2\right],$$
where $v(X)=\Var(Y\mid X)$ and
$$\mathbf{g}=\begin{bmatrix} g_1\\ g_0\end{bmatrix} =\begin{bmatrix} 1/\{\mu_1(1-\mu_1)\}\\[2pt] 1/\{\mu_0(1-\mu_0)\} \end{bmatrix},\qquad \Sigma(\pi)=\E\!\left[\frac{v(X)}{\pi(X)} \begin{bmatrix} 1\{X\in G_1\} & 1\{X\in G_1\cap G_0\}\\ 1\{X\in G_1\cap G_0\} & 1\{X\in G_0\} \end{bmatrix}\right].$$
In earlier sections we focused on disjoint groups, in which case $1\{X\in G_1\cap G_0\}=0$ and $\Sigma(\pi)$ is diagonal. With overlapping groups, however, $\Sigma(\pi)$ typically has nonzero off-diagonal terms, and the variance-minimizing policy is no longer the same as in the disjoint case. The corresponding empirical objective is
$$\widehat{\Var}(\hat{\theta}) =\frac{1}{n}\sum_{i=1}^n \frac{v(X_i)}{\pi(X_i)} \Big(g_1 1\{X_i\in G_1\}+g_0 1\{X_i\in G_0\}\Big)^2.$$
When $\pi$ is parameterized via spline coefficients, this objective is convex in $(\beta_1,\dots,\beta_B)$, so the optimization method in Section~\ref{sec:optim} applies without modification.

While convenient, a single shared policy may be restrictive. One possibility is to partition the covariate space into three disjoint regions,
$\tilde G_1=G_1\setminus G_0$, $\tilde G_0=G_0\setminus G_1$, and $\tilde G_2=G_1\cap G_0$,
and to fit separate policies on each region. The resulting variance objective decomposes across the partition and remains convex, but this approach effectively treats the intersection as a distinct group.

A more natural alternative is to retain two subgroup-specific policies $\pi_1$ and $\pi_0$, while ensuring that labels in the intersection are reusable across both estimators. Concretely, let
$\xi_i^1\sim\Bern(\pi_1(X_i))$ and $\xi_i^0\sim\Bern(\pi_0(X_i))$, and define $\xi_i=\min(\xi_i^1+\xi_i^0,1)$ so that an observation in $G_1\cap G_0$ is labeled if either policy selects it. Equivalently, $\xi_i\sim\Bern(\pi(X_i))$ with
$$\pi(x)=1\{x\in G_1\}\pi_1(x)+1\{x\in G_0\}\pi_0(x)-1\{x\in G_1\cap G_0\}\pi_1(x)\pi_0(x).$$
Using separate spline coefficients for $\pi_1$ and $\pi_0$ generally makes $\widehat{\Var}(\hat{\theta})$ non-convex in the combined parameter vector. However, imposing the constraint $\pi_1(x)=\pi_0(x)$ for all $x\in G_1\cap G_0$ (e.g., via equality constraints on logit-transformed probabilities) restores convexity. This extension appears to fit naturally within the optimization framework developed above, and we leave a full treatment to future work.

\subsection{Data outside of the two groups}
Finally, we consider inference for estimands defined on two target subgroups when additional labeled data are available from outside those groups. For example, let $G_1$ denote women under 25 and $G_0$ denote men under 25, with unlabeled features from $n_1$ and $n_0$ individuals in the two groups, respectively. Suppose we also observe labeled data $\{(X_i,Y_i)\}_{i=1}^K$ from a third group $G_3$, e.g., women ages 26--27. Although these observations do not belong to either target subgroup, it can be reasonable to leverage them—both for training (or fine-tuning) a predictive model $\hat f$, and potentially for improving estimation directly since group 3 may be covariate-similar to the targets. To incorporate group 3 labels into estimation, let $q_g(x)$ denote the density of $X\mid X \in G_g$ and define the density ratios
$r_g(x)=q_g(x)/q_3(x)$ for $g\in\{1,2\}$. Then one can augment the mean estimator $\hat\mu_g^\pi$ with an additional debiasing term that uses residuals $(Y_i-f(X_i))$ from the $G_3$ samples reweighted by an estimate $\hat r_g(X_i)$. Concretely, this amounts to adding a correction based on $\sum_{i:X_i \in G_3}\hat r_g(X_i)\{Y_i-f(X_i)\}$ (with appropriate normalization), which targets the covariate distribution of group $g$ while drawing signal from labeled data outside the target groups.

\section*{Acknowledgments}
The authors would like to thank John Cherian, Bradley Efron, Yunhe Gao, Joonhyuk Lee, Lihua Lei, Yash Nair, Asher Spector, Skyler Wu, James Yang, and Tijana Zrnic for helpful discussions. V.L.M. acknowledges support from the Stanford Data Science Scholarship. E.J.C. was supported by the Office of Naval Research grant N00014-24-1-2305. Some of the computing for this project was performed on the Sherlock cluster. We would like to thank Stanford University and the Stanford Research Computing Center for providing computational resources and support that contributed to these research results. 

\printbibliography

\newpage

\appendix

\addcontentsline{toc}{section}{Appendix}
\part{Appendix}
\mtcsettitle{parttoc}{}
\mtcsetrules{parttoc}{off}
\parttoc

\section{Related literature}\label{app:lit}

This appendix expands upon the discussion of related work in Section \ref{sec:motivation}. The introduction reviews the most closely related modern work on prediction-powered inference, active statistical inference, and ML-assisted inference. Here, we explore four additional perspectives that clarify the structure of our framework: Neyman allocation and variance-optimal sampling; two-phase, validation, and surrogate-assisted sampling; recent semi-supervised inference for general targets; and semiparametric efficiency theory through the efficient influence function used in Section~\ref{sec:general} \cite{vdv,bickel}. 

\subsection{Neyman Allocation}\label{app:neyman}

A classical point of reference for our variance-minimizing objective is Neyman Allocation \cite{NeymanAlloc, CochranSampling}. In its simplest form, Neyman Allocation prescribes assigning more samples to strata with larger outcome  variability, so that the total variance of a stratified estimator is minimized under a fixed budget. In the continuous-covariate analogue, where each point may have its own local variance contribution, the corresponding idealized allocation assigns sampling probabilities proportional to the square root of the local variance contribution.

This is the structure that appears in our setting. For estimators with influence-function decomposition $$\phi_\pi(X, \xi, Y) = h(X) - \psi + \frac{\xi}{\pi(X)}\zeta(X,Y),$$ the variance is of the form  $$V(\pi) = \E\bigg[\frac{c(X)}{\pi(X)}\bigg] + C,$$ where $C$ is constant with respect to the labeling policy $\pi$. If one ignores the smoothness restrictions imposed by the spline basis and instead optimizes over unrestricted probabilities $\{\pi(X_i)\}_{i=1}^n$, a standard Lagrange multiplier calculation yields $$\pi^*(x) \propto \sqrt{c(x)},$$ subject to budget and probability constraints. This is precisely the intuition behind the ``proportional-to-uncertainty'' policy proposed by Zrnic and Cand\`{e}s \cite{active}, where the uncertainty is a measure of variance contribution per data point.

\method~differs from the classical allocation setting in two ways. First, $c(X)$ is unknown and must be estimated from model predictions, uncertainty scores, annd nuisance estimates. Second, directly optimizing one probability per data point can be unstable when these estimates are noisy. We address this by restricting the policy to a smooth, tractable function class---in this case by spline parametrization---and thus regularizing the mapping from estimated uncertainty to labeling probability while preserving the same variance-minimizing principle. 

Shape constraints can also be incorporated: for example, we can consider monotone rules (i.e. formalizing the notion that the labeling probability should be larger for higher uncertainty data points) such as isotonic regression or piecewise constant functions via binning of uncertainties, both of which achieve other forms of regularization. This is discussed in further detail in Appendix~\ref{app:monotonicity}.

\subsection{Two-phase, validation, and surrogate-assisted sampling}

Another classical perspective comes from two-phase or validation sampling. In these designs, inexpensive surrogate variables are observed for an entire cohort, while expensive gold-standard measurements are collected only for a subset. The second-phase sample is then used, often through inverse-probability weighting, calibration, or augmented estimators, to obtain valid inference for the target parameter. This mirrors our setting, where features, model predictions, and uncertainty scores are available for all units, but true labels are obtained only for a selected subset of observations. Classical two-phase analyses and designs have been developed for semiparametric models, measurement-error settings, and validation-sample selection \cite{breslow, amorim}.

A particularly related line of work studies cost-efficient chart-review designs for EHR-based association studies, where an often noisy algorithm-derived phenotype is available for all patients but gold-standard outcomes are obtained only for a validation subset. Yin et al.~\cite{costeffective} use EHR-derived phenotypes to guide manual chart review and construct augmented estimators for risk-factor associations, with a downstream association model in a mind. More recently, surrogate-powered inference methods combine abundant surrogate labels with a smaller set of groundtruth labels: Chen et al.~\cite{surrogatePI} develop a toolbox that includes augmentation, regularization over multiple surrogates, and an adaptive multiwave validation procedure.

While these methods align with \method~by considering label acquisition as a design problem, \method~takes a broader
approach: it formulates the allocation objective through the estimand-specific variance contribution $c(X)$, derived
from either the Delta method or the efficient influence function. This makes the same optimization machinery
applicable across odds ratios, regression coefficients, U-statistics, and other EIF-based targets, rather than being
tailored to a particular validation-study model, measurement-error setting, or surrogate-label procedure.

\subsection{Semi-supervised and ML-assisted inference for general targets}

We also discuss related literature on semi-supervised and ML-assisted inference. Recent work has developed theory for general M-estimators in semi-supervised settings~\cite{song2024generalM}, regression coefficients under misspecification and covariate shift~\cite{tian2024semisupervised}, and doubly robust semi-supervised inference under selection-based or weak-overlap labeling~\cite{zhang2023doublerobustSS}. More recent semiparametric formulations allow missing-at-random labeling, distribution shift, and decaying overlap for broad classes of smooth targets including means, regression coefficients, quantiles, and causal effects~\cite{testa2025semiparametric}. A separate line of work develops post-prediction inference procedures that are task agnostic---they can be paired with existing analysis routines without deriving a new debiasing formula for each specific task~\cite{miao2024taskagnostic}. A line of work on predict-then-debias also extends ML-assisted inference to settings with imputed covariates and nonuniform complete-data sampling designs~\cite{kluger2025imputed}. Related to the adaptive design aspect of our setting, recent work on two-phase multiwave sampling develops valid M-estimation when expensive measurements are collected adaptively across waves~\cite{kluger2026multiwave}.

These methods are all complementary to \method: they study how to obtain valid and efficient inference under a given labeled/unlabeled sampling mechanism, while \method~uses the estimand-specific variance decomposition to determine a labeling scheme under a fixed budget. The unlabeled covariates and predictions are used not only to augment the estimator, but also to determine the labeling policy.

\subsection{Semiparametric efficiency and efficient influence functions}

Our general formulation is grounded in semiparametric efficiency theory~\cite{bickel,vdv}. For a target parameter
$\psi=\psi(P)$, the efficient influence function identifies the first-order sensitivity of the estimand to
perturbations of the data-generating distribution and determines the corresponding efficiency bound. In our setting,
this perspective is useful not only for analyzing estimators, but also for designing the labeling policy.

Suppose first that the full-data influence function admits the decomposition
$$\phi_{\mathrm{full}}(X,Y)=h_\psi(X)-\psi+\zeta(X,Y), \qquad \E[\zeta(X,Y)\mid X]=0,$$
where $h_\psi(X)$ is a covariate-only component satisfying $\E[h(X)]=\psi$. Under missing-at-random labeling with
$\xi\indep Y\mid X$ and labeling probability $\pi(X)$, the corresponding observed-data influence function is
$$\phi_\pi(X,\xi,Y) = h_\psi(X)-\psi+\frac{\xi}{\pi(X)}\zeta(X,Y).$$
This immediately yields the variance decomposition
$$\Var(\phi_\pi) = \Var(h_\psi(X))+\E\left[\frac{c(X)}{\pi(X)}\right], \quad \text{where } \; c(X):=\E[\zeta(X,Y)^2\mid X].$$
Only the second term depends on the labeling policy. Thus, the EIF identifies exactly the part of the efficiency bound
that governs optimal label allocation. The role of semiparametric efficiency theory here is therefore slightly different from its usual one: rather than using the EIF only to analyze an estimator after the sampling design has been fixed, we
use it to guide the design of the labeling policy itself.

This also clarifies the relation to M-estimators, which are the focus of some related PPI and active-inference procedures~\cite{ppi++,active}. If $\psi$ is defined by an estimating equation $\E[m(O;\psi)]=0$, with $O = (X,Y)$, then under standard smoothness and nonsingularity conditions the corresponding estimator is asymptotically linear with full-data influence function
$$\phi_{\mathrm{full}}(X,Y) = -A^{-1}m(X,Y;\psi), \quad \text{where } \; A:=\E\left[\frac{\partial}{\partial\psi^\top}m(O;\psi)\right].$$
Applying the preceding decomposition gives the observed-data influence function
$$\phi_{\mathrm{obs}}(X,\xi,\xi Y) = \E[\phi_{\mathrm{full}}(X,Y)\mid X] + \frac{\xi}{\pi(X)} \left\{\phi_{\mathrm{full}}(X,Y)-\E[\phi_{\mathrm{full}}(X,Y)\mid X] \right\},$$
which is of the form used in Section~\ref{sec:general}, after identifying the residual component
$$\zeta(X,Y) = \phi_{\mathrm{full}}(X,Y)-\E[\phi_{\mathrm{full}}(X,Y)\mid X].$$
Thus, regular M-estimators are contained as a major special case.

The same logic also covers estimands that are not naturally expressed as M-estimators. For functions of means, including
the log odds ratio, the Delta method identifies the relevant componentwise variance contributions. For regression
coefficients, the influence function follows from the estimating equation above. For U-statistics and rank-based
estimands, the first-order term is given by the H\'{a}jek projection~\cite{hajek1968asymptotic,serfling1980approximation,vdv}. Once the residual component $\zeta(X,Y)$ is identified, OPAL optimizes the corresponding policy-dependent variance
contribution $\E[c(X)/\pi(X)]$. This connection also explains how one-step estimators and TMLE can be used to construct
estimators whose asymptotic variance is governed by the same EIF-based objective~\cite{tmle,tmle2}.

\section{Additional set-up details}\label{app:setup}

\subsection{Data and model scenarios}\label{app:train_model}

While we make the initial simplifying assumption that predictor model $f$ is a black-box and is trained on data independent of our sample, this is not a necessary assumption. In the following, we detail potential scenarios covering cases where a pre-trained $f$ is not available a-priori and/or there is training data available. We define a potential training data set $\mathcal{D}_{\text{train}} = \{(X_i', Y_i')\}_{i = 1}^{n_{\text{train}}}$, which are independent and identically distributed, and are also independent of the data sample of interest.

\begin{itemize}
    \item No pre-trained $f$, no training data
    \begin{itemize}
        \item Batch: take a random train-test split and train a predictive model $f$ on the training split. Perform inference on the remaining test split.
        \item Sequential: take a burn-in sample to train an initial model $f$ and iteratively update $f$ as more labeled data pairs are collected.
    \end{itemize}
    \item No pre-trained $f$, given training/pilot study data $\mathcal{D}_{\text{train}}$
    \begin{itemize}
        \item Batch: train a predictive model $f$ with $\mathcal{D}_{\text{train}}$, optionally use $\mathcal{D}_{\text{train}}$ to get an initial tuning parameter $\lambda$
        \item Sequential: train a predictive model $f_0$ with $\mathcal{D}_{\text{train}}$, then optimize tuning parameter $\lambda$ with the same training data before beginning iterative updating of labeling policy $\pi_t$, tuning parameter $\lambda_t$, and model $f_t$ based on newly collected labels. No need to collect a burn-in sample.
    \end{itemize}
    \item Pre-trained $f$ and no training data: this is precisely the case discussed in the main text. 
    \item Pre-trained $f$ and training data $\mathcal{D}_{\text{train}}$:
    \begin{itemize}
        \item Batch: finetune $f$ with $\mathcal{D}_{\text{train}}$, optionally use $\mathcal{D}_{\text{train}}$ to get an initial tuning parameter $\lambda$
        \item Sequential: finetune $f$ with $\mathcal{D}_{\text{train}}$ and optionally get initial tuning parameter, no need to collect a burn-in sample. 
    \end{itemize}
\end{itemize}

\subsection{Full schematic} \label{app:full_diagram}
In this section, we provide an extended version of Figure~\ref{fig:method}, which includes the possibility of training data $D_{\text{train}} = \{(X_i, Y_i)\}_{i \in \text{train}}$. 

\begin{figure}[h!]
    \centering
    \includegraphics[width=\linewidth]{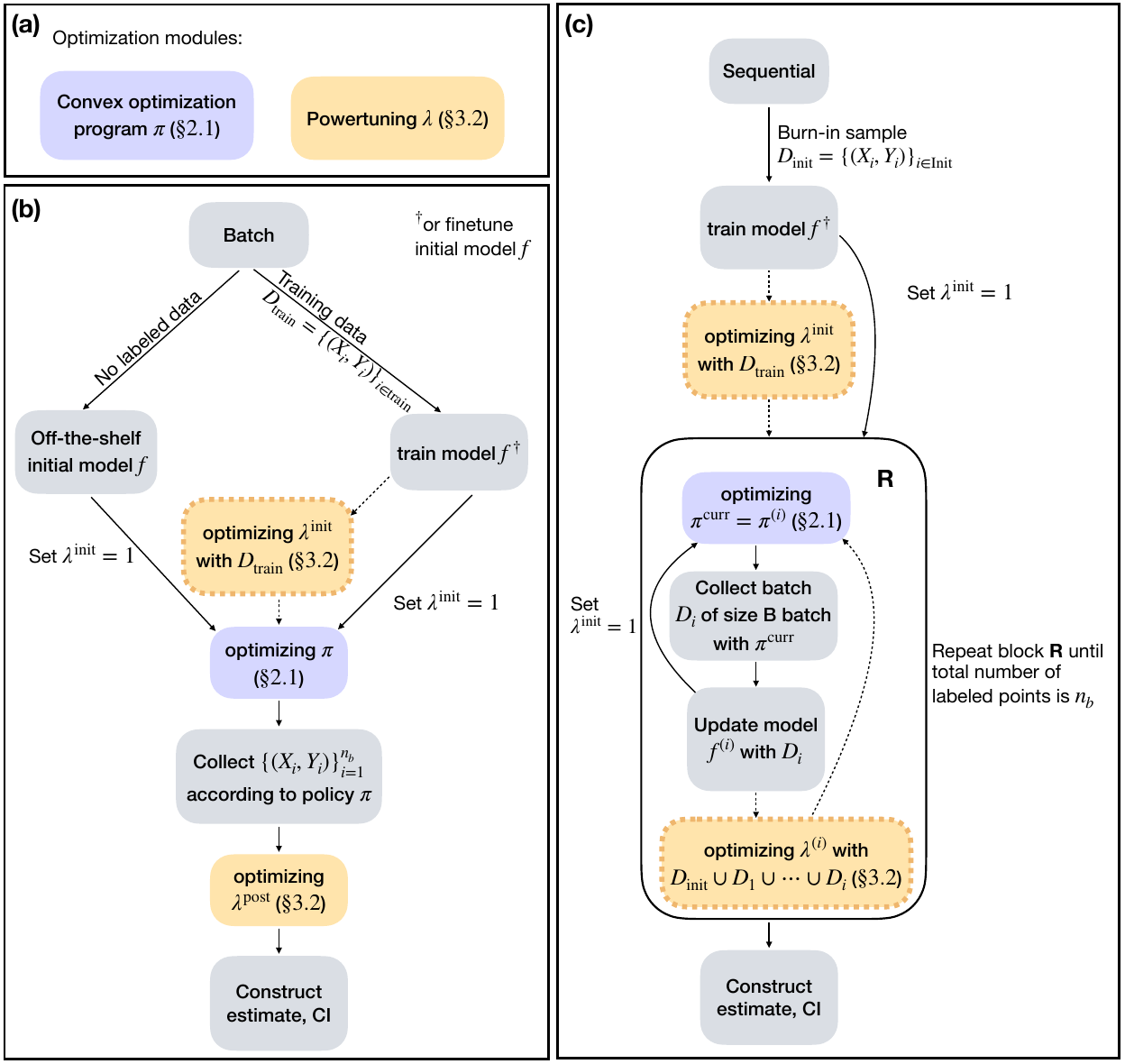}
    \caption{Full overview of OPAL incorporating (a) optimization modules for finding labeling policy $\pi$ and best tuning parameter $\lambda$ (Section~\ref{sec:tuning}). The workflow using these two modules is given for (b) the batch setting where all covariate data is available; and (c) the sequential setting where covariate data is revealed one-individual-at-a-time. Dashed lines around a module and dashed lines indicate optional modules.}
    \label{fig:full_method}
\end{figure}

\section{Variance derivations for the policy objective}

This appendix derives the variance expressions underlying the policy-learning objective presented in Section~\ref{sec:optim}. We first consider the function-of-means setting introduced in Section~\ref{sec:optim}, where the target parameter is a smooth functional of a vector of subgroup means, and show that the asymptotic variance of the resulting plug-in estimator decomposes into a policy-dependent term of the form $\E[c(X)/\pi(X)]$ plus a constant term independent of the labeling policy. We then provide a specific expression for the case of the log odds ratio.

\subsection{Function-of-means setting}\label{app:asyvar}
First, recall the setting of Section~\ref{sec:optim}. Let $\mu = (\mu_1, \dots, \mu_K)^\top$, where $$\mu_k = \E[s_k(X,Y)],$$ and let the target parameter be $\theta = g(\mu)$ where $g: \R^K \to \R$ is a differentiable map. For each component $k$, define the mean estimator with respect to labeling policy $\pi$: $$\hat{\mu}_k^\pi = \frac{1}{n}\sum_{i=1}^n \left[f_k(X_i) + \frac{\xi_i}{\pi(X_i)}(s_k(X_i, Y_i) - f_k(X_i))\right],$$ where $\xi_i \sim \Bern(\pi(X_i))$ and $\xi_i \indep Y_i \mid X_i$. Writing the residual $$r_k(X_i, Y_i) = s_k(X_i, Y_i) - f_k(X_i),$$ the $k^{th}$ summand can be expressed as $$T_{ik}^\pi = f_k(X_i) + \frac{\xi_i}{\pi(X_i)}r_k(X_i,Y_i).$$

By construction, since $\E[\xi_i \mid X_i] = \pi(X_i)$, $\E[T_{ik}^\pi \mid X_i] = \E[x_k(X_i, Y_i) \mid X_i]$, so $\hat{\mu}_k^\pi$ is unbiased for $\mu_k$ (the residual term augments the sample mean and serves as debiasing). Assuming asymptotic normality holds (see Appendix \ref{app:CLT}) with covariance matrix $\Sigma(\pi)$, denoting $\hat{\mu}^\pi = (\hat{\mu}_1^\pi, \dots, \hat{\mu}_K^\pi)^\top$, applying the Delta method to $\hat{\theta}^\pi = g(\hat{\mu}^\pi)$ yields $$\sqrt{n}(\hat{\theta}^\pi - \theta) \toD \N(0, (\nabla g(\mu))^\top\Sigma(\pi)\nabla g(\mu)).$$ Thus, the variance quantity corresponding to policy $\pi$ is the scalar quantity $$V(\pi) := (\nabla g(\mu))^\top\Sigma(\pi)\nabla g(\mu).$$ Let $a = \nabla g(\mu) \in \R^K$. The first-order fluctuation of the scalar target $\hat{\theta}^\pi$ is a linear combination of the component fluctuations of $\hat{\mu}^\pi$, $a^\top\hat{\mu}^\pi$. We define the residual term $\zeta(X,Y) = a^\top r = \sum_{k=1}^K a_k(s_k(X,Y) - f_k(X))$ and the componentwise conditional mean $$h(X) = \sum_{k=1}^K a_k\E[s_k(X,Y) \mid X].$$

The first-order influence contribution (from a linearization) is $$\phi_\pi(X,Y, \xi) = \frac{\xi}{\pi(X)}\zeta(X,Y) + h(X) - \theta,$$ with $\E[\zeta(X,Y) \mid X] = 0$. Thus, $$V(\pi) = \Var(\phi_\pi) = \E\left[\frac{c(X)}{\pi(X)}\right] + \Var(h(X)), \quad \text{where } c(X) = \E[\zeta(X,Y)^2 \mid X].$$
Since $\Var(h(X))$ is constant with respect to $\pi$, minimizing the asymptotic variance is equivalent to minimizing $\E[c(X) / \pi(X)]$, subject to the labeling budget and valid probability constraints. This population-level objective translates to the empirical objective introduced in Section \ref{sec:optim}.

As an aside, we note this is just the special case of Section~\ref{sec:general}, taking $\psi = \theta = g(\mu)$. Accordingly, the variance form matches.

\subsection{Log odds ratio}\label{app:asyvar_OR}

We now specify the above general expression to the log odds ratio variance given in the main text in Eq.~\eqref{eq:OddsRatioVar}. Let $\mu_0 = \E[Y \mid X \in G_0]$ and $\mu_1 = \E[Y \mid X \in G_1]$, and define $$\theta = g(\mu_1, \mu_0) = \log\left(\frac{\mu_1}{1 - \mu_1}\right) - \log \left(\frac{\mu_0}{1 - \mu_0}\right).$$ Then, $$\nabla g(\mu_1, \mu_0) = \left(\frac{1}{\mu_1(1 - \mu_1)}, \; -\frac{1}{\mu_0(1 - \mu_0)}\right)^\top.$$
Thus, following the format of the function-of-means setting in the previous part, we define $s_1(X,Y) = Y1\{X \in G_1\}$ and $s_0(X,Y) = Y1\{X \in G_0\}$. Taking $f_k(X) \equiv f(X)$ for $k\in \{0,1\}$, the residual term $\zeta(X,Y)$ becomes $$\zeta(X,Y) = \frac{1\{X \in G_1\}}{\mu_1(1 - \mu_1)}(Y - f(X)) - \frac{1\{X \in G_0\}}{\mu_0(1 - \mu_0)}(Y - f(X)).$$

The groups are disjoint, so the cross-term vanishes as $1\{X \in G_1\}1\{X \in G_0\} = 0$. Thus, $$c(X) = \E[\zeta(X,Y)^2 \mid X] = \E[(Y - f(X))^2 \mid X] \left[\frac{1\{X \in G_1\}}{p_1^2\mu_1^2(1 - \mu_1)^2} + \frac{1\{X \in G_0\}}{p_0^2\mu_0^2(1 - \mu_0)^2}\right],$$
matching the expression given in Eq.~\eqref{eq:OddsRatioVar}, if we take the plug-in estimates $\hat{p}_k = n_k/n$. This is the policy-relevant variance contribution for the log odds ratio. Replacing the population quantities with pilot-sample estimates yields the practical object for optimization. 

\section{Convex optimization}
This section of the appendix provides details on the spline-based parametrization of the labeling policy and the resulting convex optimization problem introduced in Section~\ref{sec:optim}, Eq.~\eqref{eq:optim_full}. We first briefly discuss the unrestricted pointwise formulation in which each labeling probability is treated as its own optimization variable, and explain why this problem is convex and can also be viewed as a geometric program. We then show that the empirical objective corresponding to the spline-based parametrization used in the paper is convex in the spline coefficients, and explain how the probability and budget constraints are enforced in practice.

\subsection{Pointwise optimization} \label{app:pointwise_opt}

We begin with the version of the empirical policy-learning problem in which the labeling probabilities are optimized directly without any smoothness restriction. For a fixed sample with features $X_1, \dots, X_n$, let $$p_i := \pi(X_i), \quad i = 1, \dots, n$$
and write $\widehat{c}_i = \widehat c(X_i)$. The empirical variance proxy can then be written as a function of the labeling probabilities $p_1, \dots, p_n$: $$\widehat V(p_1, \dots, p_n) = \sum_{i=1}^n \frac{\widehat{c}_i}{p_i},$$ so the unrestricted optimization problem is
\label{eq:optim_unrestricted}
\begin{align} 
    \min_{p_1, \dots, p_n} \quad & \sum_{i=1}^n \frac{\widehat c_i}{p_i}, \\
    \text{subject to} \quad &  \sum_{i=1}^n p_i \leq \budget, \\
    & 0 < p_i \leq 1, \quad i=1, \dots, n. 
\end{align}
This problem is convex: each term $\widehat c_i / p_i$ is convex on $(0, \infty)$ and the feasible set is defined by linear inequalities with positivity constraints.

As a brief aside, this problem can also be viewed as a geometric program. Each term $\widehat c_i/p_i$ is a monomial in the positive variable $p_i$, and thus the objective, which is a sum of these monomials, is a posynomial. Similarly, budget constraint $\sum_{i=1}^n p_i \leq \budget$ is a posynomial inequality. Algebraically, if we define $d := \prod_{j=1}^n p_j$ and $d_i := \prod_{j \neq i} p_j = d / p_i$, then $\frac{1}{p_i} = \frac{d_i}{d}$, and we can rewrite the optimization objective as $$\sum_{i=1}^n \frac{\widehat c_i}{a_i} = \frac{1}{d}\sum_{i=1}^n \widehat c_i d_i.$$ This is precisely the algebraic form of a geometric program, and can thus be solved in this way.

Returning to the initial pointwise optimization problem, if we ignore the upper-bound constraints $p_i \leq 1$, the problem introduced in Eq.~\eqref{eq:optim_unrestricted} admits a closed-form solution. The Lagrangian is $$L(a,\lambda) = \sum_{i=1}^n \frac{\widehat c_i}{p_i} + \lambda\left(\sum_{i=1}^n p_i - n_{\budget}\right),$$ and the first-order conditions give
$$-\frac{\widehat c_i}{p_i^2}+\lambda=0 \qquad\Longrightarrow\qquad p_i = \sqrt{\frac{\widehat c_i}{\lambda}}.$$
Imposing the budget constraint yields
$$\sqrt{\lambda} = \frac{\sum_{j=1}^n \sqrt{\widehat c_j}}{n_{\budget}},$$
and therefore
$$p_i^* = \frac{n_{\budget}\sqrt{\widehat c_i}}{\sum_{j=1}^n \sqrt{\widehat c_j}}.$$ Thus, the unrestricted interior solution is proportional to $\sqrt{\widehat c_i}$, which is the finite-sample analogue of Neyman allocation discussed in Appendix~\ref{app:neyman}. When the additional constraint $p_i \leq 1$ is factored in, the same square-root holds with truncation at 1 for those whose unconstrained optimum exceeds the upper bound.

The spline parametrization is therefore not necessary to create convexity, since the unrestricted pointwise problem is already convex and in fact explicitly solvable in the interior. Rather, the spline class regularizes the map from uncertainty to labeling probability by replacing $n$ free optimization variables with a low-dimensional smooth parametrization. This is especially useful when the empirical uncertainty quantities $\widehat c_i$ are noisy or misspecified, in which case the fully pointwise solution may be overly sensitive to local estimation error.

\subsection{Spline basis construction} \label{app:spline}

We use a B-spline basis\cite{bspline}, whose elements satisfy a recursive relationship. Let $\{t_i\}_{i=1}^{K + 2d + 2}$ be a knot sequence with $d + 1$ repeated boundary knots at each end. The basis functions are defined recursively by
\begin{align*}
    h_{j,0}(u) &= 1_{[t_{j}, t_{j + 1})}(u) \\
    h_{j,d}(u) &= \frac{u - t_j}{t_{j + d} - t_j}h_{j, d-1}(u) + \frac{t_{j + d + 1} - u}{t_{j + d + 1} - t_{j + 1}} h_{j + 1, d - 1}(u)
\end{align*}
For each uncertainty score $U_i := u(X_i)$, we write $$h(U_i) = \left(h_{1,d}(U_i), \dots, h_{B,d}(U_i)\right)^\top,$$
where the total number of basis functions is $B = K + d + 1$. This reflects the fact that the piecewise-polynomial segments are constrained to agree, together with their first $d-1$ derivatives, at each interior knot \cite{spline}. For each group $k$, we use the same basis construction and parametrize the log-inverse labeling probability as $$\log \frac{1}{\pi_{\beta^{(k)}}^{(k)}(X_i)} = h(U_i)^\top \beta^{(k)}, \qquad \text{equivalently} \qquad \pi_{\beta^{(k)}}^{(k)}(X_i) = \exp\{-h(U_i)^\top\beta^{(k)}\}.$$
Thus, the map from uncertainty to labeling probability is represented by a piecewise polynomial of degree $d$ with knots at $\{t_j\}$ and $C^{d-1}$ continuity at the interior knots. We construct the policy as a smooth function of the scalar uncertainty score $U_i$ rather than of the full covariate vector $X_i$, which avoids the need for higher-dimensional basis expansions while still allowing flexible nonlinear dependence of the labeling probability on uncertainty.

\subsection{Convexity of the spline-parametrized optimization problem} \label{app:convexity_spline}

We first recall the empirical objective $$\widehat V(\pi) = \frac{1}{n}\sum_{i=1}^n \frac{\widehat c(X_i)}{\pi(X_i)}.$$
Substituting in the spline-parametrization $\pi_\beta(X_i) = \exp\{-h(U_i)^\top\beta\}$, $$\widehat V(\beta) = \frac{1}{n}\sum_{i=1}^n c(X_i)\exp\{h(U_i)^\top\beta\}.$$
Each term in the sum is a nonnegative constant, $c(X_i)$, times the exponential of an affine function of $\beta$, $h(U_i)^\top\beta$, and is therefore convex in $\beta$. As a result, $\widehat V(\beta)$ is a sum of convex functions and is itself convex. Alternatively, we see its Hessian is positive semidefinite: $$\nabla^2 \widehat V(\beta) = \frac{1}{n}\sum_{i=1}^n \widehat c(X_i) \exp\{h(U_i)^\top\beta\}\, h(U_i)h(U_i)^\top \succeq 0.$$ Now, we verify convexity of the feasible set. The probability constraints $\pi_\beta(X_i) \leq 1$ are equivalent to $h(U_i)^\top\beta \geq 0$, which are linear inequalities and hence a convex set. The budget constraint $\sum_{i=1}^n \pi_\beta(X_i) = \sum_{i=1}^n \exp\{-h(U_i)^\top\beta\} \leq \budget$ is also convex: the left-hand side is a sum of exponentials of affine functions. Thus, the optimization problem consists of a convex objective minimized over a convex feasible set, and is thus a well-defined convex optimization problem.

\subsection{Constraint handling and convex formulation}

The spline parametrization makes it straightforward to enforce both the probability constraints and the budget constraint. Because
$$\pi_\beta(X_i)=\exp\{-h(U_i)^\top\beta\},$$
we have $\pi_\beta(X_i)\in(0,1]$ whenever
$$h(U_i)^\top\beta \ge 0.$$
Thus the valid-probability constraints become the linear inequalities 
$$h(U_i)^\top\beta \ge 0, \qquad i=1,\dots,n.$$
The empirical budget constraint requires
$$\sum_{i=1}^n \pi_\beta(X_i) = \sum_{i=1}^n \exp\{-h(U_i)^\top\beta\} \le n_{\budget}.$$
Therefore, the optimization problem can be written as
\begin{align}
\min_{\beta}\quad &
\sum_{i=1}^n \widehat c(X_i)\exp\{h(U_i)^\top\beta\}
\label{eq:spline_obj}
\\
\text{subject to}\quad &
h(U_i)^\top\beta \ge 0, \qquad i=1,\dots,n,
\label{eq:spline_nonneg}
\\
&
\sum_{i=1}^n \exp\{-h(U_i)^\top\beta\}\le n_{\budget}.
\label{eq:spline_budget}
\end{align}

The resulting optimization problem is convex: the objective in \eqref{eq:spline_obj} is a nonnegative weighted sum of exponentials of affine
functions of the spline coefficients, and the budget constraint in \eqref{eq:spline_budget} is a convex sublevel constraint of the same form, and the constraints in \eqref{eq:spline_nonneg} are linear. In our implementation, we specify the problem using \texttt{CVXPY}'s exponential atoms; \texttt{CVXPY} then canonicalizes the problem to a conic
form involving exponential cones and solves it with a conic solver such as \texttt{MOSEK}.

\subsection{Imposing additional constraints for monotonicity} \label{app:monotonicity}

As discussed in Section~\ref{sec:optim}, although we do not impose monotonicity constraints by default in the formulation of the convex optimization problem for finding labeling policies, they can be enforced with additional linear constraints or directly through the choice of basis.

We can add linear constraints for monotonicity by the following. Let $u_{(1)} < \cdots < u_{(m)}$ be a sorted grid, either of the sorted unique uncertainty values or a fixed grid. Then, we enforce the following constraint: $$\pi_\beta(u_{(j+1)}) \geq \pi_\beta(u_{(j)}) \iff h(u_{(j+1)})^\top\beta \leq h(u_{(j)})^\top\beta, \quad j = 1, \dots, m-1.$$ These are linear inequalities in $\beta$, so the overall optimization problem, written below, remains convex:
\begin{align*} 
    \text{minimize}  & \qquad \sum_{i=1}^n c(U_i) \exp\{\sum_{b=1}^B \beta_b h_b(U_i)\}, \\
    \text{subject to}  & \qquad \sum_{b=1}^B \beta_b h_b(U_i) \geq 0, \quad \forall i \in [n], \\
    & \qquad \sum_{b=1}^B \beta_b h_b(u_{(j+1)}) \leq \sum_{b=1}^B \beta_b h_b(u_{(j)}),\quad \forall j \in [m-1], \\
    & \qquad \sum_{i=1}^n \exp\{-\sum_{b=1}^B \beta_bh_b(U_i)\} \leq n_{\text{ budget}}. \\
\end{align*}
We note that by picking a potentially coarse grid of $u_i$ points for the monotonicity constraint, it maintains monotonicity on a moderately fine grid due to spline smoothness, but is not technically guaranteeing monotonicity between every point $i$. If desired, a grid with each $u_i$ as a grid point will ensure this, however, it is more restrictive, adds a large number of linear constraints, and may become overly sensitive to specific $u_i$ values.

If global monotonicity is desired, one can instead use a monotone spline parameterization, such as I-splines \cite{Ispline1, Ispline2}, or impose the equivalent coefficient-difference constraints on the B-spline expansion. These constraints enforce the sign of the derivative over the full range of uncertainty scores while preserving convexity. If we define $s(u) = 1/\log(\pi(u))$, then $s'(u) \leq 0$ for all $u$ ensures $\pi(u) = \exp\{-s(u)\}$ is non-decreasing for all $u$. A sufficient and convex condition for $s'(u) \leq 0$ for all $u$ is to enforce sign constraints on adjacent coefficient differences since $s(u) = h(u)^\top \bbeta$: $$\beta_b - \beta_{b-1} \leq 0, \quad j = 2, \dots, B \iff \beta_1 \geq \beta_2 \geq \cdots \geq \beta_B.$$
The corresponding optimization problem, which remains convex, is
\begin{align*} 
    \text{minimize}  & \qquad \sum_{i=1}^n c(U_i) \exp\{\sum_{b=1}^B \beta_b h_b(U_i)\}, \\
    \text{subject to}  & \qquad \sum_{b=1}^B \beta_b h_b(U_i) \geq 0, \quad \forall i \in [n], \\
    & \qquad \beta_1 \geq \beta_2 \geq \cdots \geq \beta_B, \\
    & \qquad \sum_{i=1}^n \exp\{-\sum_{b=1}^B \beta_bh_b(U_i)\} \leq n_{\text{ budget}}. \\
\end{align*}
Alternatively, we can enforce global monotonicity by replacing the B-spline basis with an I-spline \cite{Ispline1, Ispline2} basis and constraining coefficients to be nonnegative. I-splines are integrals of nonnegative M-splines, which are essentially a normalized B-spline: each basis function is scaled so that it integrates to 1 over its support. We can write $$s(u) = a - \sum_{b=1}^B \theta_b I_b(u), \quad \theta_b \geq 0,$$ where $I_b(\cdot)$ are spline basis functions. Then, since $I_b(u)$ are non-decreasing in $u$ and $\theta_b \geq 0$ for all $b$, the sum $\sum_b \theta_b I_b(\theta)$ is also non-decreasing and thus $s(u)$ is non-increasing.

\subsection{Practical implementation details} \label{app:implementation}

In practice, we first use pilot data, a training split, or a burn-in sample to obtain a predictive model $f$ and the corresponding estimated uncertainty quantities $\widehat c(X_i)$ or $u(X_i)$. Then, we construct the spline basis matrix $$H \in \R^{n \times B}, \qquad H_{ib} = h_b(U_i),$$ for the uncertainty scores of the unlabeled data. 

Conditional on the unlabeled features and the uncertainty estimates (pilot-derived), the optimization problem is deterministic. The resulting fitted coefficients $\widehat \beta$ determine the labeling probabilities $$\pi_{\beta}(X_i) = \exp\{-h(U_i)^\top\widehat\beta\},$$
which are then used to collect labels according to independent Bernoulli draws in the batch setting, or updated sequentially in the online setting. We formulate the problem in \texttt{cvxpy} and use MOSEK, with other solvers supporting exponential-cone constraints as fallbacks.

\section{Generalized asymptotically linear estimators}\label{general_details}

This appendix provides additional details for the broader asymptotically linear framework introduced in Section~\ref{sec:general}. We first summarize how TMLE may be implemented in our label querying setting, and then derive the efficient influence function representations, residual terms, and conditional second moments used in Table~\ref{tab:EIF_form} for the examples discussed in the main text (also in Section~\ref{sec:general}).

\subsection{TMLE details} \label{app:TMLE}
Here, we provide details on using Targeted Maximum Likelihood Estimation (TMLE), which updates an initial plug-in estimate along a low-dimensional fluctuation model chosen to target the parameter of interest. The update is constructed so that the empirical efficient score equation is solved (exactly or approximately), yielding a targeted plug-in estimator with favorable finite-sample behavior and the same first-order asymptotic efficiency under regularity conditions.

Let $O_i$ denote one observed data unit (e.g., $O_i=(X_i,\xi_i,\xi_iY_i)$ in our selective-label setting), and let $\phi(O;\eta)$ denote the efficient influence function (EIF) evaluated at nuisance parameters $\eta$. Starting from initial nuisance estimates (here, $\hat h$ and $\hat\zeta$), the TMLE procedure proceeds as follows.

\begin{enumerate}
    \item \textbf{Construction of the clever covariate.} First, we construct the \emph{clever covariate} $H(X,\xi)$, which is the EIF-based direction used to target the nuisance update for $h$. Equivalently, $H$ is chosen so that the fluctuation submodel score aligns with the relevant component of the efficient influence function. This direction determines the targeting update in Steps 2--4. For simple mean functionals, this takes the form $H(X,\xi)=\xi/\pi(X)$. In more complex settings, such as conditional quantiles where the EIF involves nonsmooth indicator terms, deriving $H$ analytically can be challenging. In such cases, one may use automatic differentiation and, when appropriate, smooth approximations (e.g., sigmoid relaxations of indicators) to construct a practical targeting update.
    \item \textbf{Definition of the fluctuation submodel.} We next define a low-dimensional parametric submodel $\{h_\varepsilon : \varepsilon \in \R\}$ that passes through the initial estimate $\hat{h}$ at $\varepsilon = 0$. The submodel is chosen so that its score at $\varepsilon = 0$ aligns with the relevant EIF component (equivalently, with the clever covariate direction from Step 1), ensuring that estimation of $\varepsilon$ targets the parameter of interest. For binary outcomes, a standard choice is the logistic fluctuation:
    $$\operatorname{logit} (h_\varepsilon(X)) = \operatorname{logit} (\hat{h}(X)) + \varepsilon H(X, \xi),$$
    which automatically preserves $h_\varepsilon(X) \in (0,1)$. More generally, one may consider $$h_\varepsilon(X) = g^{-1}\!\big(g(\hat h(X)) + \varepsilon d(X)\big),$$ where $g$ is a link function and $d$ is a direction chosen to align the submodel score with the EIF-based targeting direction.
    \item \textbf{Estimation of the fluctuation parameter.} We estimate the fluctuation parameter $\varepsilon$ by maximizing the likelihood over the chosen fluctuation submodel (equivalently, by solving the associated empirical score equation). Conceptually, this step determines how far to move from the initial estimate $\hat h$ along the EIF-targeted update direction to remove the first-order bias. 

    More concretely, let $H_i = H(X_i, \xi_i)$ denote the clever covariate and let $\eta_{0,i} = \mathrm{logit}(\hat h(X_i))$ denote the initial linear predictor. Under the logistic fluctuation submodel, $$\mathrm{logit}(h_\varepsilon(X_i)) = \eta_{0,i} + \varepsilon H_i,$$ the conditional mean is $$\mu_i(\varepsilon) = \sigma(\eta_{0,i} + \varepsilon H_i).$$ In the binary-outcome mean-estimation case ($H_i = \xi_i / \pi(X_i)$), the fluctuation parameter $\hat \varepsilon$ solves the empirical score equation $$0 = \frac{1}{n}\sum_{i=1}^n H_i\big(Y_i - \mu_i(\varepsilon)\big) = \frac{1}{n}\sum_{i=1}^n \frac{\xi_i}{\pi(X_i)}\bigg[Y_i - \sigma\bigg(\mathrm{logit}(\hat h(X_i) + \varepsilon\frac{\xi_i}{\pi(X_i)}\bigg)\bigg].$$
    This can be implemented efficiently using standard generalized linear model software by fitting an intercept-free logistic regression of $Y_i$ on $H_i$, with $\eta_{0.i} = \mathrm{logit}(\hat h(X_i))$ included as a fixed offset. The offset anchors the update at the initial estimate, while the fitted coefficient of $H_i$ yields the fluctuation parameter $\hat \varepsilon$.
    \item \textbf{Update and form the targeted estimator.} 
    Finally, we update the nuisance estimate by setting $$\hat{h}^*(X) := h_{\hat \varepsilon}(X),$$
    where $\hat\varepsilon$ is the fluctuation parameter estimated in Step 3. We then form the targeted plug-in estimator by replacing $\hat h$ with $\hat h^*$:
    $$\hat{\psi}^{\mathrm{TMLE}} = \Phi(\hat P^*),$$ where $\hat P^*$ denotes the targeted estimate of the relevant components of the data-generating distribution and is obtained by replacing $\hat h$ with $\hat h^*$ and updating any other nuisance components needed to evaluate the parameter. In the mean estimation case where we have missing labels so $\psi = \E[Y] = \E[h(X)]$, this reduces to $$\hat{\psi}^{\mathrm{TMLE}} = \frac{1}{n}\sum_{i=1}^n \hat h^*(X_i).$$
    By construction, this is a substitution estimator, so it preserves the parameter-space constraints induced by the model for $h$ (e.g. probability estimates remain in $[0,1]$ under a logistic fluctuation). 
    
    At the same time, if the targeting step solves the empirical EIF equation up to negligible error (satisfied under standard regularity conditions~\cite{tmle}), the resulting estimator has the same first-order asymptotic linear expansion as the corresponding one-step estimator, with asymptotic variance $\Var(\phi(X, \xi, \xi Y))/n$~\cite{vdv}. 
\end{enumerate}
As in the one-step case, inference is based on the estimated EIF in practice: if $\hat\varphi_i$ denotes the estimated influence function evaluated at the targeted nuisance estimates, then
$$\widehat{\Var}(\hat\psi^{\mathrm{TMLE}})\approx \frac{1}{n}\,\widehat{\Var}(\hat\varphi_i).$$
In more detail, $\displaystyle \hat{\varphi}_i^* = \hat h^*(X_i) - \hat \psi_{\mathrm{TMLE}} + \frac{\xi_i}{\pi(X_i)}\hat\zeta^*(X_i, Y_i),$ and we define
$$\widehat{\sigma}^2_{\mathrm{TMLE}} = \frac{1}{n}\sum_{i=1}^n \left(\hat\varphi_i - \overline{\hat\varphi^*}\right)^2, \quad \text{where} \; \overline{\hat\varphi^*} = \frac{1}{n}\sum_{i=1}^n \hat \varphi_i^*.$$
Then, we can construct the following Wald-style confidence interval: $$\hat\psi_{\mathrm{TMLE}} \pm z_{1 - \alpha/2} \sqrt{\widehat{\sigma}^2_{\mathrm{TMLE}} / n}.$$

\subsection{Deriving efficient influence function forms} \label{app:EIF}

Below, we provide detailed computations for the efficient influence function forms, residual forms, and conditional second moments reported in Table~\ref{tab:EIF_form}. These calculations are standard and included for completeness.

\paragraph{General setup and variance form.}
Let $Z=(X,Y)\sim P$ denote the full-data random vector and let the observed data be
$O=(X,\xi,\xi Y)$, where $\xi\in\{0,1\}$ indicates whether $Y$ is revealed. Assume the labeling mechanism is missing at random with known propensity $\pi(X)=\P(\xi=1\mid X)$ and $\xi \perp Y \mid X$.
If the full-data EIF admits a decomposition of the form
$$\phi_{\mathrm{full}}(Z)=\underbrace{\bar\phi(X)}_{\E[\phi_{\mathrm{full}}(Z)\mid X]} \;+\; \underbrace{\zeta(X,Y)}_{\phi_{\mathrm{full}}(Z)-\E[\phi_{\mathrm{full}}(Z)\mid X]}, \qquad \E[\zeta(X,Y)\mid X]=0,$$
then the efficient influence function in the observed-data model is
$$\varphi(O)=\bar\phi(X)\;+\;\frac{\xi}{\pi(X)}\,\zeta(X,Y),$$
which satisfies $\E[\varphi(O)]=0$ and yields $\mathrm{AsyVar}(\hat\psi)=\Var(\varphi(O))/n$.
We record $c(X):=\E[\zeta(X,Y)^2\mid X]$ since it governs the contribution of $X$ to the variance under inverse-probability weighting. Since $$\E\left[\frac{\xi}{\pi(X)}\zeta(X,Y) \mid X\right] = \frac{\E[\xi \mid X]}{\pi(X)}\E[\zeta(X,Y) \mid X] = 0,$$ the law of total variance gives 
\begin{align*}
    \Var(\phi((X,\xi, \xi Y))) &= \E\left[\Var\left(\frac{\xi}{\pi(X)}\zeta(X,Y) \mid X\right)\right] + \Var(\bar\phi(X)) \\
    &= \E\left[\frac{1}{\pi(X)} \E[\zeta(X,Y)^2 \mid X]\right] + \Var(\bar\phi(X))
\end{align*}
Thus, $$\Var(\varphi((X, \xi, \xi Y))) = \E\left[\frac{c(X)}{\pi(X)}\right] + \Var(\bar\phi(X)).$$
This is the variance form used in Section~\ref{sec:general}. In particular, only the first term depends on the labeling policy $\pi$, which is why the same policy-learning objective extends beyond the function-of-means setting.

\paragraph{Linear functional.}
Consider the linear functional $\psi = \E[h(X)Y],$
where $h:\mathcal X\to\R$ is known (e.g., $h(X)=\mathbf 1\{X\in G\} / \P(X \in G)$ for a subgroup mean). The full-data EIF is
$\phi_{\mathrm{full}}(X,Y)=h(X)Y-\psi.$
Let $\mu(X):=\E[Y\mid X]$. Then
$$\E[\phi_{\mathrm{full}}(X,Y)\mid X]=h(X)\mu(X)-\psi, \qquad \zeta(X,Y)=h(X)\{Y-\mu(X)\}.$$
Consequently,
$$c(X)=\E[\zeta(X,Y)^2\mid X]=h(X)^2\,\Var(Y\mid X).$$
Under our labeling scheme (which can be interpreted as missing at random), the observed-data EIF is
$$\varphi(O)=h(X)\mu(X)-\psi+\frac{\xi}{\pi(X)}\,h(X)\{Y-\mu(X)\}.$$

\paragraph{Regression coefficient (population least squares).}
Let $\gamma\in\R^d$ be the population least-squares coefficient defined by
$$\E\big[X(Y-X^\top\gamma)\big]=0,$$
and let $\theta=\gamma_j$ denote the $j$th coordinate. Write $\Sigma:=\E[XX^\top]$ and assume $\Sigma$ is invertible.
This is a smooth $Z$-estimation problem with score $\psi_\gamma(Z):=X(Y-X^\top\gamma)$ and Jacobian
$$A:=\frac{\partial}{\partial\gamma}\E[\psi_\gamma(Z)] = -\E[XX^\top]=-\Sigma.$$
The standard influence function for $Z$-estimators yields
$$\phi_{\mathrm{full}}^{(j)}(X,Y) = e_j^\top \Sigma^{-1} X \,(Y-X^\top\gamma) =: w_j(X)\,(Y-X^\top\gamma), \qquad w_j(X):=(\Sigma^{-1}X)_j.$$
To be compatible with the missing-label EIF form of Section~\ref{sec:general}, the residual term must be conditionally mean-zero given $X$. Let $m(X) := \E[Y \mid X]$. We decompose the full-data influence function into its conditional mean and conditionally mean-zero residual:
$$\bar\phi_j(X) := \E[\phi_{\mathrm{full}}^{(j)}(X,Y)\mid X]= w_j(X)[m(X)-X^\top\gamma],$$
and $$\zeta_j(X,Y) = \phi_{\mathrm{full}}^{(j)}(X,Y)-\bar\phi_j(X) = w_j(X)[Y-m(X)].$$
Thus, $\E[\zeta_j(X,Y) \mid X] = 0$. In the notation of Section~\ref{sec:general}, one may take $$h_j(X) = \gamma_j + w_j(X)[m(X) - X^\top\gamma], \qquad \zeta_j(X,Y) = w_j(X)[Y - m(X)].$$
The corresponding observed-data influence function is
$$\varphi^{(j)}(O) = h_j(X) - \gamma_j + \frac{\xi}{\pi(X)}\zeta_j(X,Y) = w_j(X)[m(X) - X^\top\gamma] + \frac{\xi}{\pi(X)}w_j(X)[Y - m(X)].$$
Thus, the policy-relevant conditional second moment is $$c_j(X) = \E[\zeta_j(X,Y)^2 \mid X] = w_j(X)^2\Var(Y \mid X).$$  If the linear model is correctly specified so that $m(X) = X^\top\gamma$, then $\bar\phi_j(X) = 0$, and the residual simplifies to $$\zeta_j(X,Y) = w_j(X)(Y - X^\top\gamma), \qquad c_j(X) = w_j(X)^2\Var(Y \mid X).$$

\paragraph{U-statistics.}
First, we distinguish between the exact EIF residual and the equivalent policy score ($c(X)$ is only needed up to a positive target-specific multiplicative constant since it doesn't change the optimizer). Let $\psi=\E[k(Z,Z')]$ be a (symmetric) U-statistic functional with $Z=(X,Y)$ and i.i.d.\ copy $Z'$.
Define the (scaled) Hájek projection
$$\tilde\psi(Z):=2\,\E\big[k(Z,Z')\mid Z\big] - \psi,$$
so that the full-data EIF is
$$\phi_{\mathrm{full}}(Z)=\tilde\psi(Z)-\psi = 2\{\E\big[k(Z,Z')\mid Z\big]-\psi\}.$$
(With this definition, we can absorb the conventional factor of $2$.) Writing $\tilde\psi(X,Y)$ for $\tilde\psi(Z)$,
$$\E[\phi_{\mathrm{full}}(Z)\mid X]=\E[\tilde\psi(X,Y)\mid X]-\psi, \qquad \zeta(X,Y)=\tilde\psi(X,Y)-\E[\tilde\psi(X,Y)\mid X].$$
Therefore,
$$c(X)=\E[\zeta(X,Y)^2\mid X]=\Var(\tilde\psi(X,Y)\mid X) = 4\Var(\E\big[k(Z,Z')\mid Z\big]).$$
For learning the labeling policy, we can use the score
$c(X):=\Var(q(X,Y)\mid X)$ equivalently, multiplying
the objective by a positive constant does not change the optimizer.

\paragraph{Kendall's $\tau$ (binary outcome case).}
As a special case of the U-statistic setup, let $Y\in\{0,1\}$ and consider Kendall's rank correlation
$$\tau = \E\big[\mathrm{sign}(Y-Y')\text{sign}(X-X')\big],$$
corresponding to the symmetric kernel $k((X,Y),(X',Y'))=\text{sign}(Y-Y')\text{sign}(X-X')$.
Let $p_1=\P(Y=1)$, $p_0=\P(Y=0)$, and define the class-conditional CDFs
$$F_y(t):=\P(X\le t\mid Y=y),\quad y\in\{0,1\}, \qquad p(X):=\P(Y=1\mid X).$$
A direct calculation of the Hájek projection gives (up to the same scaling convention as above) the two class-conditional projections
$$\tilde\psi(X,1)=p_0\{2F_0(X)-1\}, \qquad \tilde\psi(X,0)=p_1\{1-2F_1(X)\}.$$
Hence the centered residual term is of the same form as the general U-statistic row, $$\zeta(X,Y)=\tilde\psi(X,Y)-\E[\tilde\psi(X,Y)\mid X],$$
and because $Y\mid X$ is Bernoulli with success probability $p(X)$, its conditional variance is
$$c(X)=\Var(\tilde\psi(X,Y)\mid X) = p(X)\{1-p(X)\}\,\left(\tilde\psi(X,1)-\tilde\psi(X,0)\right)^2,$$
where we've excluded the multiplicative factor of 4 (so technically it's a proportional-to sign).
Substituting the expressions above yields
$$\tilde\psi(X,1)-\tilde\psi(X,0) = p_0(2F_0(X)-1) - p_1(1-2F_1(X)) = -\Big(1-2[p_0F_0(X)+p_1F_1(X)]\Big),$$
and therefore $$c(X)=p(X)\{1-p(X)\}\left(1 - 2[p_0F_0(X)+p_1F_1(X)]\right)^2.$$

\paragraph{Average treatment effect (ATE).}
Recall the average treatment effect $\psi = \E[Y(1)-Y(0)]$, where we adopt standard potential outcome framework notation: $A\in\{0,1\}$ denotes treatment, $Y(a)$ is the potential outcome under treatment level $a$, and $Y=Y(A)$ is the observed outcome. Let $e(X):=\P(A=1\mid X)$ and $m_a(X):=\E[Y\mid A=a,X]$ for $a\in\{0,1\}$ denote the propensity score and mean function, respectively. Under standard identification conditions, the full-data efficient influence function is
$$\phi_{\mathrm{ATE,full}}(X,A,Y) = \frac{A}{e(X)}\{Y-m_1(X)\} - \frac{1-A}{1-e(X)}\{Y-m_0(X)\} + m_1(X)-m_0(X)-\psi.$$

Letting $W=(X,A)$ denote the always observed variables, we take conditional expectation given $(X,A)$:
$$\bar\phi_{\mathrm{ATE}}(X,A) := \E[\phi_{\mathrm{ATE,full}}(X,A,Y)\mid X,A] = m_1(X)-m_0(X)-\psi,$$
since $\E[Y-m_1(X)\mid X,A=1]=0$ and $\E[Y-m_0(X)\mid X,A=0]=0.$ Hence the residual term is
\begin{align*}
    \zeta(X,A,Y) &= \phi_{\mathrm{ATE,full}}(X,A,Y)-\bar\phi_{\mathrm{ATE}}(X,A) \\
    &= \frac{A}{e(X)}\{Y-m_1(X)\} - \frac{1-A}{1-e(X)}\{Y-m_0(X)\}.
\end{align*}

In our missing-label model, we allow the label querying rule to depend on both features and treatment, with
$$\pi(X,A):=\P(\xi=1\mid X,A), \qquad \xi \perp Y \mid (X,A).$$
The corresponding observed-data EIF is therefore
$$\varphi_{\mathrm{ATE}}(X,A,\xi,\xi Y) = \frac{\xi}{\pi(X,A)} \left[\frac{A}{e(X)}\{Y-m_1(X)\} - \frac{1-A}{1-e(X)}\{Y-m_0(X)\}\right] + m_1(X) - m_0(X) - \psi.$$

The policy-relevant conditional second moment is $c(X,A)=\E[\zeta(X,A,Y)^2\mid X,A]$. Since only one residual term is ``active'' for a given treatment value (and $A \in \{0,1\}$),
$$c(X,A) = \frac{A}{e(X)^2}\Var(Y\mid X,A=1) + \frac{1-A}{(1-e(X))^2}\Var(Y\mid X,A=0).$$
Equivalently, writing $\sigma_a^2(X):=\Var(Y\mid X,A=a)$,
$$c(X,A) = \frac{A\,\sigma_1^2(X)}{e(X)^2} + \frac{(1-A)\,\sigma_0^2(X)}{(1-e(X))^2}.$$
Thus the policy-relevant variance contribution depends on both the treatment propensity and the conditional outcome variability within each treatment arm.

\section{Additional algorithmic details}\label{app:algorithms}

We provide additional implementation details for the procedures described in Sections~\ref{sec:crossfitting} and \ref{sec:sequential}. 
The algorithms below are specific to the odds-ratio setting for concreteness, though the same algorithm templates apply more generally to other mean-type functions and the broader EIF-based framework of Section~\ref{sec:general}. The policy-learning workflow and sampling steps remain the same across estimands; what changes is the form of the final estimator and its variance estimate, which are determined by the corresponding efficient influence function, affecting the optimization objective. For non-mean estimands, the subgroup mean estimators appearing below would be replaced by the corresponding one-step or TMLE estimator, and inference would be based on the empirical variance of the estimated influence-function contributions.

\paragraph{Cross-fitting}
We provide detailed pseudocode for the crossfitting procedure described in Section~\ref{sec:crossfitting}. In the odds-ratio setting, this takes the form of constructing fold-specific subgroup mean estimators and combining at the end, with a global label policy initially learned. More generally, in the EIF-based framework introduced in Section~\ref{sec:general}, the same template applies: one fold is used to train or update the nuisance quantities needed to construct the estimator, while the other fold is used to evaluate the correspponding one-step or TMLE estimator. The final estimator is obtained by aggregating the fold-specific estimates, and the variance estimate is based onthe empirical varinace of the corresponding estimated influence-function contributions. For simplicity, we present the algorithm as two-fold cross-fitting, but the same construction extends immediately to fixed $K$-fold cross-fitting by training on all folds except the held-out one.

Two-fold cross-fitting proceeds by training the predictor on one fold and evaluating the debiased estimator on the held-out fold, then reversing the roles of the folds. Then, the two estimators are aggregated to construct the cross-fit estimator. Equivalently, each observation is assigned an out-of-fold debiased pseudo-outcome $T_i^{\cf}$, and the final estimator averages these pseudo-outcomes within each group. A theoretical justification for the validity of the cross-fitting procedure is provided in Appendix~\ref{app:crossfitting_OR}.

\begin{algorithm}[H]
\caption{Algorithm for two-fold cross-fitted batch inference}
\label{alg:crossfitting}
\begin{algorithmic}[1]
\State \textbf{Input:} Unlabeled covariates $(X_1,\ldots,X_n)$ with group indicators $Z_i\in\{0,1\}$, initial predictor $f$, labeling budget $n_{\budget}$.
\State \textbf{Global policy.} Learn a global labeling policy $\pi:\mathcal X\to[0,1]$ using only the unlabeled covariates and the initial predictor $f$; see Section~\ref{sec:optim}.
\State \textbf{Labeling.} Sample $\xi_i\sim\Bern(\pi(X_i))$ independently and observe $Y_i$ for units with $\xi_i=1$.
\State \textbf{Splitting.} Randomly split indices $\{1,\ldots,n\}$ into two folds $\mathcal I_1$ and $\mathcal I_2$.
\State \textbf{Out-of-fold prediction.} Fine-tune two predictors:
$$\hat f^{(-1)} \text{ using labeled observations in } \mathcal I_2, \qquad \hat f^{(-2)} \text{ using labeled observations in } \mathcal I_1 .$$
\State \textbf{Cross-fitted pseudo-outcomes.} For each observation, define
$$T_i^{\cf} = \begin{cases} \displaystyle \hat f^{(-1)}(X_i) + \frac{\xi_i}{\pi(X_i)} \{Y_i-\hat f^{(-1)}(X_i)\}, & i\in\mathcal I_1,\\[1.25em] \displaystyle \hat f^{(-2)}(X_i) + \frac{\xi_i}{\pi(X_i)} \{Y_i-\hat f^{(-2)}(X_i)\}, & i\in\mathcal I_2. \end{cases}$$
The residual term is evaluated only for labeled units.
\State \textbf{Group mean estimation.} For $k\in\{0,1\}$, compute
$$\hat\mu_k^{\cf} = \frac{1}{n_k}\sum_{i:Z_i=k}T_i^{\cf}, \qquad n_k:=\sum_{i=1}^n \mathbf 1\{Z_i=k\}.$$
\State \textbf{Odds-ratio estimation.} Compute
$$\hat\theta^{\cf} = \log\frac{\hat\mu_1^{\cf}}{1-\hat\mu_1^{\cf}} - \log\frac{\hat\mu_0^{\cf}}{1-\hat\mu_0^{\cf}}.$$
\State \textbf{Variance estimation.} Let $\hat p_k=n_k/n$. Define
$$\hat\phi_{\theta,i}^{\cf} = \frac{\mathbf 1\{Z_i=1\}}{\hat p_1} \frac{T_i^{\cf}-\hat\mu_1^{\cf}} {\hat\mu_1^{\cf}(1-\hat\mu_1^{\cf})} - \frac{\mathbf 1\{Z_i=0\}}{\hat p_0} \frac{T_i^{\cf}-\hat\mu_0^{\cf}} {\hat\mu_0^{\cf}(1-\hat\mu_0^{\cf})}.$$
Set
$$\hat V_\theta^{\cf} = \frac{1}{n}\sum_{i=1}^n \left(\hat\phi_{\theta,i}^{\cf}\right)^2.$$
\State \textbf{Return} the confidence interval
$$\hat\theta^{\cf} \pm z_{1-\alpha/2}\sqrt{\frac{\hat V_\theta^{\cf}}{n}}.$$
\end{algorithmic}
\end{algorithm}

\paragraph{Full sequential pseudocode (deferred from Section~\ref{sec:sequential}).}
For readability, the main text emphasizes the conceptual differences between the batch and sequential settings and refers to Figure~\ref{fig:method} for a high-level overview. We provide here the complete step-by-step pseudocode for the sequential procedure, including burn-in, periodic model/policy updates, and budget-feasible probability clipping. Following Section~\ref{sec:sequential}, we enforce budget feasibility by clipping the proposed labeling probability using the cumulative-budget rule. Define the assigned cumulative budget at step $t$ by $$\budget^{(t)} = \frac{t\, \budget}{n},$$ and let $n_{\text{lab}}^{(t-1)}$ denote the number of labels collected up to and including step $t-1$. Then, the remaining step-$t$ allowance is defined as $$n_\Delta^{(t)} := \budget^{(t)} - n_{\text{lab}}^{(t-1}).$$
Given a proposed labeling probability $\pi_t(X_t)$ based on the optimized policy at time step $t$, we use the budget feasible probability (clipped to the feasible interval $[0,1]$) $$\tilde \pi_t(X_t) := \left[\min\left\{\pi_t(X_t), \, n_{\Delta}^{(t)}\right\}\right]_{[0,1]} = \min\left\{1, \max\left\{0, \min\left\{\pi_t(X_t), \, n_{\Delta}^{(t)}\right\}\right\}\right\}.$$

\begin{algorithm}
\caption{Sequential active inference with optimized labeling}
\label{alg:seq}
\begin{algorithmic}[1]
\Require Unlabeled features $(X_1,\dots,X_n)$ with group indicators $Z_i \in \{0,1\}$; total labeling budget $\budget$; update batch size $B$; optional initial predictor $f_0$
\State $\mathcal{D}^{\mathrm{lab}} \leftarrow \emptyset$ \Comment{newly labeled samples for periodic fine-tuning}
\State $b_1 \leftarrow \budget$ \Comment{remaining labeling budget}
\State $\mathcal{H} \leftarrow \emptyset$ 
\State $n_{\mathrm{burn}} \leftarrow 0$
\If{no initial predictor $f_0$ is provided}
    \State Choose burn-in size $n_{\mathrm{burn}}$ (with $0 \le n_{\mathrm{burn}} \le \min\{n,\budget\}$)
    \For{$i=1,\dots,n_{\mathrm{burn}}$}
        \State Collect label $Y_i$ and add $(X_i, Y_i)$ to $\mathcal{H}$ \Comment{burn-in labels used to initialize predictor}
    \EndFor
    \State Fit initial predictor $f_{n_{\text{burn}}}$ on $\{(X_i,Y_i)\}_{i=1}^{n_{\mathrm{burn}}}$
    \State $b_{n_{\text{burn} + 1}} \leftarrow \budget - n_{\text{burn}}$.
    \Else
    \State Set $n_{\mathrm{burn}} = 0$ and $f_1 \leftarrow f_0$.
\EndIf

\If{$n_{\text{burn}} > 0$ then}
\State{Using $f_{n_{\text{burn} + 1}}$, compute uncertainty scores for the remaining features and learn policies $\pi_{n_{\text{burn} + 1}}^{(1)}(\cdot)$ and $\pi_{n_{\text{burn} + 1}}^{(0)}(\cdot)$ via Section~\ref{sec:optim}}
\Else
\State Using $f_1$, compute uncertainty scores for all features and learn policies $\pi_1^{(1)}(\cdot)$ and $\pi_1^{(0)}(\cdot)$ via Section~\ref{sec:optim}.
\EndIf

\For{$t=n_{\text{burnin}}+1,\dots,n$}
    \State $r_t \leftarrow n-t+1$ \Comment{remaining features including $X_t$}
    \State $p_t(X_t) \leftarrow \pi_t^{(Z_t)}(X_t)$ 
    \Comment{group-based policy}
    \State $\tilde{p}_t(X_t) \leftarrow \texttt{clip}(p_t(X_t); b_t)$    
    \Comment{Budget-feasible clipping}

    \State Sample labeling decision $\xi_t \sim \Bern(\tilde p_t(X_t))$
    \If{$\xi_t=1$}
        \State Collect label $Y_t$
        \State $\mathcal{D}^{\mathrm{lab}} \leftarrow \mathcal{D}^{\mathrm{lab}} \cup \{(X_t,Y_t)\}$
    \EndIf

    \State $b_{t+1} \leftarrow b_t-\xi_t$ \Comment{update remaining budget}
    \State $f_{t+1} \leftarrow f_t$
    \State $\pi_{t+1}^{(1)}(\cdot) \leftarrow \pi_t^{(1)}(\cdot)$, \quad $\pi_{t+1}^{(0)}(\cdot) \leftarrow \pi_t^{(0)}(\cdot)$

    \If{$|\mathcal{D}^{\mathrm{lab}}| = B$ \textbf{and} $t<n$}
        \State Fine-tune predictor on collected labels:
        $$f_{t+1} \leftarrow \texttt{finetune}(f_t,\mathcal{D}^{\mathrm{lab}})$$
        \State Recompute uncertainty scores for remaining features $\{X_{t+1},\dots,X_n\}$ using $f_{t+1}$.
        \State Recompute labeling policies $\pi_{t+1}^{(1)}(\cdot), \pi_{t+1}^{(0)}(\cdot)$ via Section~\ref{sec:optim} using $f_{t+1}$.
        \State Reset $\mathcal{D}^{\mathrm{lab}} \leftarrow \emptyset$.
    \EndIf
\EndFor

\State Compute subgroup mean estimates $\hat{\mu}_1^\pi,\hat{\mu}_0^\pi$ using the realized clipped probabilities $\tilde{\pi}_t$
\State Form the sequential log odds-ratio estimator $$\hat\theta^\pi = \log\left(\frac{\hat{\mu}_1^\pi}{1 - \hat{\mu}_1^\pi}\right) - \log\left(\frac{\hat{\mu}_0^\pi}{1 - \hat{\mu}_0^\pi}\right).$$ 
\State Estimate the subgroup variances using the corresponding sequential plug-in estimators: $V_1, V_0$.
\State \textbf{Return} confidence interval based on the sequential asymptotic normality result.
\end{algorithmic}
\end{algorithm}

\section{Asymptotic theory} \label{app:CLT}

This appendix contains the main asymptotic arguments underlying Section~\ref{sec:batch}. We begin by collecting the regularity conditions and assumptions used throughout this section. Then, we present the intermediate results used in the proof of the batch CLT: a multivariate CLT for the subgroup mean estimators under the oracle policy, and the consistency and local asymptotic behavior of the estimated spline coefficients. These ingredients are then combined in Appendix~\ref{app:equicontinuity} to prove the batch CLT under the learned policy. We subsequently record the corresponding odds-ratio corollary, the extension to the broader EIF-based framework of Section~2.2, and the sequential asymptotic normality result. 

\subsection{Regularity conditions}
\label{app:regularity}

The regularity conditions used throughout Appendix~\ref{app:CLT} are all presented here. These assumptions are stated at a high level to avoid repeating conditions across the oracle-policy CLT, the learned-policy CLT, the generalized EIF result, and the sequential results in subsequent sections.

\begin{assumption}[Basic sampling, positivity, and moments]
\label{ass:basic_sampling}
The full data $(X_i,Y_i)_{i=1}^n$ are i.i.d. from a distribution $P$. In the two-group setting, $G_0$ and $G_1$ are disjoint, $\P(X\in G_k)>0$, and $\mu_k:=\E[Y\mid X\in G_k]$ is finite for $k\in\{0,1\}$. For log odds-ratio inference, there exists $\eta>0$ such that
$$\eta \leq \mu_k \leq 1-\eta,\qquad k\in\{0,1\}.$$
The relevant policies are uniformly positive: there exists $\underline\pi>0$ such that
$$\inf_{x,\beta\in\N(\beta^*)}\pi_\beta(x)\geq \underline\pi .$$
The predictor $f$ used in the estimating equations is fixed, trained on independent pilot data, or cross-fitted so that the usual first-order CLT arguments apply. The relevant second moments are finite; in particular,
$$\E\!\left[\left\{f(X)+\frac{\xi}{\pi_{\beta^*}(X)}(Y-f(X))\right\}^2\right]<\infty .$$
\end{assumption}

\begin{assumption}[Policy-learning regularity]
\label{ass:policy_learning}

Let $U=u(X)$ denote the uncertainty score and let $h(U)\in\mathbb R^B$
denote the fixed-dimensional spline basis. For $\beta\in\mathcal B$, define
$$s_\beta(X):=h(X)^\top\beta, \qquad \pi_\beta(X):=\exp\{-s_\beta(X)\}.$$
Let
$$m_\beta(X) := c(X)\exp\{s_\beta(X)\}$$
denote the objective contribution, and define (in empirical process notation)
$$Q(\beta):=P m_\beta = \E[c(X)\exp\{h(U)^\top\beta\}], \qquad G(\beta):=P \pi_\beta = \E[\exp\{-h(U)^\top\beta\}].$$
The population feasible set is
$$\mathcal \B_0 := \left\{ \beta\in\mathcal B: G(\beta) - \rho\le 0,\; s_\beta(U)\ge 0\ \text{a.s.} \right\},$$
and let $\beta^* := \arg\min_{\beta\in\mathcal \B_0} Q(\beta).$ We assume the following conditions.

\begin{enumerate}
\item[(B1)] \textbf{Compactness and identification.}
The parameter space $\mathcal B\subset\mathbb R^B$ is compact, the population
feasible set $\mathcal F$ is nonempty, and $\beta^*$ is the unique minimizer of
$Q(\beta)$ over $\mathcal F$.

\item[(B2)] \textbf{Bounded finite-dimensional basis.}
The basis $h$ is fixed-dimensional and uniformly bounded:
$$\sup_u \|h(u)\|_2 < \infty.$$
In addition, the variance weight is integrable:
$$\mathbb E|c(X)|<\infty.$$

\item[(B3)] \textbf{Inner feasible approximation.}
For every $\varepsilon>0$, there exists $\beta_\varepsilon\in\mathcal B$ such that
$$s_{\beta_\varepsilon}(U)\ge 0\; \text{a.s.}, \quad G(\beta_\varepsilon)=P g_{\beta_\varepsilon}<\rho, \quad \text{and} \quad Q(\beta_\varepsilon)\le Q(\beta^*)+\varepsilon.$$

\item[(B4)] \textbf{Estimated variance weights and budget convergence.}
The empirical budget fraction satisfies
$$\rho_n\to \rho.$$
If the empirical objective uses an estimated weight $\hat c$, define
$$\hat m_{\beta,n}(W,X) := \hat c(W)\exp\{s_\beta(U)\}.$$
We assume the plug-in objective is uniformly stable:
$$\sup_{\beta\in\mathcal B} \left| P_n\hat m_{\beta,n}-P_n m_\beta \right| =o_p(1).$$
Equivalently,
$$\sup_{\beta\in\mathcal B} \left| \frac1n\sum_{i=1}^n \{\hat c(W_i)-c(W_i)\} \exp\{h(U_i)^\top\beta\} \right| =o_p(1).$$
When $c$ is known or treated as fixed, this condition is omitted.
\end{enumerate}

\begin{remark}[Uniform convergence]\label{rem:ULLN}
Under (B1; compactness) and (B2; bounded basis, integrability), the classes
$$\{m_\beta:\beta\in\mathcal B\} \qquad\text{and}\qquad \{\pi_\beta:\beta\in\mathcal B\}$$
are Glivenko--Cantelli. Hence, for fixed $c$,
$$\sup_{\beta\in\mathcal B} |P_n m_\beta-Pm_\beta|=o_p(1), \qquad \sup_{\beta\in\mathcal B} |P_n \pi_\beta-P\pi_\beta|=o_p(1).$$
If $c$ is replaced by $\hat c$, assumption (B4) implies the corresponding uniform convergence of the plug-in objective:
$$\sup_{\beta\in\mathcal B} |P_n\hat m_{\beta,n}-Pm_\beta|=o_p(1).$$
The empirical pointwise constraints $s_\beta(X_i)=h(X_i)^\top\beta\ge 0,
\quad i\in [n],$
are controlled by the VC property of halfspaces. Indeed,
$$\left\{ x:h(x)^\top\beta<0,\ \beta\in\mathcal B \right\}$$
is a VC class in the transformed feature space $h(u)$. Therefore,
$$\sup_{\beta\in\mathcal B} \left| P_n\mathbf 1\{s_\beta(X)<0\} - P\mathbf 1\{s_\beta(X)<0\} \right| =o_p(1).$$
\end{remark}
\end{assumption}

\begin{assumption}[Empirical-process regularity for learned policies]
\label{ass:emp_process}
For each estimating function indexed by the policy parameter, the corresponding policy-indexed class is $P$-Donsker and is $L_2(P)$-continuous at $\beta^*$. In particular, if $\hat\beta\toP\beta^*$, then
$$\mathbb G_n\{m_{\hat\beta}-m_{\beta^*}\}=o_p(1), \qquad \mathbb G_n:=\sqrt{n}(\mathbb P_n-P).$$
\end{assumption}
\begin{remark}
    This assumption is verified for the fixed-dimensional spline policy class used in OPAL by the Donsker and stochastic equicontinuity arguments in Lemma~\ref{lem:donsker_equicont}. 
\end{remark}

\begin{assumption}[General EIF estimator regularity]
\label{ass:general_eif}
For generalized estimands in Section~\ref{sec:general}, the target admits an observed-data EIF of the form
$$\phi_\pi(X,\xi,Y) = h(X)-\psi+\frac{\xi}{\pi(X)}\zeta(X,Y), \qquad \E[\zeta(X,Y)\mid X]=0,\qquad \E[h(X)]=\psi .$$
The relevant second moments are finite:
$$\E[h(X)^2]<\infty, \qquad \E\!\left[\frac{\E[\zeta(X,Y)^2\mid X]}{\pi_{\beta^*}(X)}\right]<\infty .$$
The corresponding one-step or TMLE estimator satisfies the standard asymptotic linearity remainder condition under the learned policy:
$$\hat\psi^{\hat\pi}-\psi = \frac1n\sum_{i=1}^n \phi_{\pi_{\hat\beta}}(X_i,\xi_i,Y_i) +o_p(n^{-1/2}).$$
\end{assumption}

\begin{assumption}[Sequential martingale conditions]
\label{ass:sequential}
For the sequential results, the realized clipped probabilities $\tilde\pi_t(X_t)$ are predictable, i.e. measurable with respect to $(\mathcal F_{t-1},X_t)$, and satisfy the required positivity and moment conditions. The centered increments used in the sequential CLTs satisfy the variance-process convergence and Lindeberg conditions stated in Appendix~\ref{app:seq_general}.
\end{assumption}

\begin{remark}[Interpretation of the assumptions]
The conditions above are standard high-level regularity assumptions for combining constrained finite-dimensional
M-estimation with semiparametric one-step/TMLE theory. Assumptions~\ref{ass:basic_sampling} and
\ref{ass:policy_learning} ensure that the learned spline policy is well behaved and that inverse-propensity weights are
stable. Assumption~\ref{ass:general_eif}
is the usual asymptotic linearity condition for one-step and TMLE estimators; it abstracts away nuisance-estimation
details that depend on the target parameter. Finally, Assumption~\ref{ass:sequential} is the martingale analogue needed
for the sequential procedure. These assumptions are primarily used to isolate the effect of learning the labeling policy;
they can be weakened in special cases, for example when the policy is fixed or when the target is a simple mean.
\end{remark}

\subsection{Intermediate results for the batch CLT (deferred from Section~\ref{sec:batch})}
\label{app:batch_clt_proofs}

In Section~\ref{sec:batch}, we stated only the main batch CLT and its odds-ratio corollary for the sake of readability. We include here the two intermediate results used in the proof of Theorem~\ref{thm:batch_mean}: (i) a multivariate CLT for the subgroup mean estimators under the oracle policy $\pi_{\beta^*}$, and (ii) consistency and local asymptotic behavior of the estimated spline coefficients $\hat{\beta}$. These are then combined in Appendix~\ref{app:equicontinuity} via a stochastic decomposition argument to prove Theorem~\ref{thm:batch_mean}.

\subsubsection{Asymptotic normality under the oracle policy} \label{app:CLT_oracle}

We first consider the oracle setting in which the spline coefficients are fixed at their population-optimal values, so that the labeling policy is non-random. In this case, the observed-data summands are i.i.d., and the asymptotic normality of the subgroup mean estimators follows from a standard multivariate CLT.

\begin{proposition}[CLT using true optimization parameters]\label{prop:oracleCLT}
Under Assumption~\ref{ass:basic_sampling}, let
$$\hat{\bmu}^{\bbeta^*} = \begin{pmatrix} \hat{\mu}_1^{\bbeta_1^*} \\ \hat{\mu}_0^{\bbeta_0^*} \end{pmatrix}$$
denote the vector of subgroup mean estimators constructed under the oracle policies
$\pi_{\bbeta_1^*}^{(1)}$ and $\pi_{\bbeta_0^*}^{(0)}$. Then
$$\sqrt{n} \left( \begin{pmatrix} \hat{\mu}_1^{\bbeta_1^*} \\ \hat{\mu}_0^{\bbeta_0^*} \end{pmatrix} - \begin{pmatrix} \mu_1 \\ \mu_0 \end{pmatrix} \right) \toD \N\!\left( \mathbf 0, \begin{pmatrix} V_1 & 0 \\ 0 & V_0 \end{pmatrix} \right),$$
where
$$V_k = \Var\!\left( \frac{1\{X\in G_k\}}{\P(X\in G_k)} \left[ f(X)+\frac{\xi}{\pi_{\bbeta_k^*}^{(k)}(X)}(Y-f(X))-\mu_k \right] \right), \qquad k\in\{0,1\}.$$
\end{proposition}

\begin{proof}
For each $k\in\{0,1\}$, write the estimator in the form
$$\hat{\mu}_k^{\bbeta_k^*} = \frac{1}{n_k}\sum_{i=1}^n 1\{X_i\in G_k\} \left[ f(X_i)+\frac{\xi_i}{\pi_{\bbeta_k^*}^{(k)}(X_i)}(Y_i-f(X_i)) \right].$$
Since $n_k/n \toP \P(X\in G_k)$ by the law of large numbers,
$$\sqrt{n}\big(\hat{\mu}_k^{\bbeta_k^*}-\mu_k\big) = \frac{1}{\P(X\in G_k)}\frac{1}{\sqrt{n}}\sum_{i=1}^n \varphi_{k,i} +o_p(1),$$
where
$$\varphi_{k,i} := 1\{X_i\in G_k\} \left[ f(X_i)+\frac{\xi_i}{\pi_{\bbeta_k^*}^{(k)}(X_i)}(Y_i-f(X_i))-\mu_k \right].$$

We first verify that $\E[\varphi_{k,i}]=0$. Conditioning on $X_i$ and using the missing-at-random assumption $\xi_i \indep Y_i \mid X_i$,
\begin{align*}
\E\!\left[
f(X_i)+\frac{\xi_i}{\pi_{\bbeta_k^*}^{(k)}(X_i)}(Y_i-f(X_i))
\,\middle|\, X_i
\right]
&=
f(X_i)
+
\frac{\E[\xi_i\mid X_i]}{\pi_{\bbeta_k^*}^{(k)}(X_i)}
\bigl(\E[Y_i\mid X_i]-f(X_i)\bigr) \\
&=
f(X_i)+\E[Y_i\mid X_i]-f(X_i) \\
&=
\E[Y_i\mid X_i].
\end{align*}
Hence
$$\E[\varphi_{k,i}] = \E\!\left[ 1\{X_i\in G_k\}\bigl(\E[Y_i\mid X_i]-\mu_k\bigr) \right] = 0,$$
since $\mu_k=\E[Y\mid X\in G_k]$.

By Assumption 1, the random vectors $(\varphi_{1,i},\varphi_{0,i})^\top$ are i.i.d.\ with finite second moments, so the multivariate CLT implies
$$\frac{1}{\sqrt{n}}\sum_{i=1}^n \begin{pmatrix} \varphi_{1,i}/\P(X\in G_1) \\ \varphi_{0,i}/\P(X\in G_0) \end{pmatrix} \toD \N\!\left( 0,\Sigma^* \right).$$
It remains to identify the covariance matrix. The diagonal entries are exactly $V_1$ and $V_0$. For the off-diagonal term, note that $G_1\cap G_0=\emptyset$, so
$$1\{X_i\in G_1\}1\{X_i\in G_0\}=0 \qquad \text{a.s.}$$
Therefore
$$\Cov(\varphi_{1,i},\varphi_{0,i})=0.$$
Thus $\Sigma^*$ is diagonal, completing the proof.
\end{proof}

\subsubsection{Consistency and local asymptotic behavior of the learned policy parameters}
\label{app:beta_asymptotics}

This subsection studies the behavior of the spline coefficients obtained from the sample-level convex program in Section~\ref{sec:optim}. The result is stated in a form that applies both to the odds-ratio example and to the broader EIF-based setting of Section~\ref{sec:general}: the estimand affects the optimization only through the weight function $c(\cdot)$ (or its sample proxy $\hat c(\cdot)$), while the spline policy class and constraints are unchanged.

\paragraph{General policy-learning program.}
For $\beta\in\R^B$, recall the defined policy class (cf.\ Eq.~\eqref{eq:splineprobs})
$$\pi_\beta(X_i) := \exp\{-h(U_i)^\top \beta\}.$$
Let $\hat c(X_i)$ denote the empirical variance weight used in the objective (e.g.\ the plug-in proxy for $c(X)=\E[\zeta^2\mid X]$) and recall the empirical program (as presented in \ref{app:regularity})
\begin{align}
\min_{\beta\in\mathcal B}\quad & Q_n(\beta):=\frac{1}{n}\sum_{i=1}^n \hat c(X_i)\exp\{h(U_i)^\top \beta\} \label{eq:Qn_def}\\
\text{s.t.}\quad &
h(U_i)^\top \beta \ge 0,\quad i=1,\dots,n, \label{eq:nonneg_constraint}\\
& \frac{1}{n}\sum_{i=1}^n \exp\{-h(U_i)^\top \beta\}\le \rho_n, 
\qquad \rho_n:=\frac{n_{\budget}}{n}. \label{eq:budget_constraint_avg}
\end{align}
$\hat\beta_n$ denotes a (measurable) minimizer of \eqref{eq:Qn_def}--\eqref{eq:budget_constraint_avg}.

The corresponding population problem is
\begin{align}
\min_{\beta\in\mathcal B}\quad & Q_0(\beta):=\E\!\left[c(X)\exp\{h(u(X))^\top \beta\}\right] \label{eq:Q0_def}\\
\text{s.t.}\quad &
\E[\exp\{-h(u(X))^\top \beta\}] \le \rho, \label{eq:budget_constraint_pop} \\
& h(u(X))^\top\beta \geq 0 \quad a.s.
\end{align}

\paragraph{Consistency.}
\begin{lemma}[Consistency of the learned policy parameters]\label{lem:beta_consistency_general}
Under Assumption~\ref{ass:policy_learning}(B1)--(B4), $\hat\beta\toP \beta^*$.
\end{lemma}

\begin{proof}
    Recall empirical problem is $$\hat{\beta}_n \in \arg\min_{\beta \in \B} P_n [m_\beta(X)]$$ subject to $$P_n [\pi_\beta(X)] < \rho_n, \quad s_\beta(X_i) \geq 0, \; i \in [n]$$ where we've defined $s_\beta(X) := h(X)^\top\beta$, $\pi_\beta(X) = \exp\{-s_\beta(U)\}$, and $m_\beta(X) = c(X)\exp\{s_\beta(X)\}$. The population problem is $$\beta^* \in \arg\min_{\beta \in \B_0} P[m_\beta(X)],$$ where the population feasible set is defined as $$\B_0 = \{\beta \in \B: P \pi_\beta \leq \rho, \; s_\beta(X) \geq 0 \; a.s.\} = \{\beta \in \B: P \pi_\beta \leq \rho, \; \P(s_\beta(X) < 0) = 0\}.$$ First, we show $\lim\sup_{n \to \infty} P[m_{\hat\beta_n}(X)] \leq P[m_{\beta^*}(X)]$. Fix $\varepsilon > 0$. By feasibility [Assumption~\ref{ass:policy_learning}(B3)], there exists $\beta_\varepsilon \in \B$ such that $P[\pi_{\beta_\varepsilon}] < \rho$, $\P(s_{\beta_{\varepsilon}} < 0) = 0$, and 
    \begin{equation} \label{eq:feasibility}
        P[m_{\beta_\varepsilon}] \leq P[m_{\beta^*}] + \varepsilon.
    \end{equation} Since $P[\pi_{\beta_\varepsilon}] < \rho$, then by LLN and $\rho_n \to \rho$, $\P(P_n[\pi_{\beta_\varepsilon}] < \rho_n) \to 1$ and $\P(s_{\beta_{\varepsilon}}(X) < 0) = 0 \implies \P(s_{\beta_{\varepsilon}}(X_i) \geq 0 \; \forall i \in [n]) = 1$. Thus, $\beta_\varepsilon$ is sample feasible with high probability. By definition of optimality of $\hat{\beta}_n$ as an empirical problem minimizer, \begin{equation}\label{eq:optimizer}
        P_n[m_{\hat{\beta}_n}(X)] \leq P_n[m_{\beta_\varepsilon}(X)].
    \end{equation} Then, 
    \begin{align*}
        P[m_{\hat{\beta}_n}] &= P_n[m_{\hat{\beta}_n}] + (P - P_n)[m_{\hat{\beta}_n}] \\
        &\leq P_n[m_{\hat{\beta}_n}] + \sup_{\beta \in \B}|P_n m_\beta - Pm_\beta| \\
        &\leq P_n[m_{\beta_\varepsilon}] + o_p(1) & \text{by \eqref{eq:optimizer} and Remark~\ref{rem:ULLN}} \\
        &\leq P[m_{\beta_\varepsilon}] + o_p(1) &  \text{by LLN } P_n [m_{\beta_\varepsilon}] \to P[m_{\beta_\varepsilon}] \\
        &\leq P[m_{\beta^*}] + \varepsilon + o_p(1) & \text{by \eqref{eq:feasibility}}
    \end{align*}
    Thus, $\lim\sup_n P[m_{\hat{\beta}_n}] \leq P[m_{\beta^*}] + \varepsilon$, and taking $\varepsilon \to 0$, $\lim\sup_{n \to \infty} P[m_{\beta}(X)] \leq P[m_{\beta^*}(X)]$ as desired.

    Now, we show $\lim\inf_{n \to \infty} P[m_{\hat\beta_n}(X)] \geq P[m_{\beta^*}(X)]$. Take the sequence $\hat{\beta}_n$. By compactness [Assumption~\ref{ass:policy_learning}(B1)], there exists a subsequence $\hat{\beta}_{n_k} \to \bar{\beta}$. Since $\hat{\beta}_n$ is sample feasible, $P_n[\pi_{\hat{\beta}_n}] \leq \rho_n$. Thus, by uniform convergence (Remark~\ref{rem:ULLN}) and $\rho_n \to \rho$, $P[\pi_{\hat{\beta}_n}] \leq \rho + o_p(1)$. By continuity of $\pi_\beta$, this implies $P[\pi_{\bar{\beta}}] \leq \rho$. Sample feasibility also implies $s_{\hat{\beta}_n}(X_i) \geq 0$ for all $i \in [n]$, or equivalently, $1\{s_{\hat{\beta}_n}(X_i) < 0\} = 0$ for all $i \in [n]$ and hence $P_n 1\{s_{\hat{\beta}_n}(X_i) < 0\} = 0$.  Since the class $\{1\{s_{\beta}(x) < 0\}: \beta \in \B\}$ is VC (Remark~\ref{rem:ULLN}),
    \begin{align*}
        P[1\{s_{\hat{\beta}_n}(X) < 0\}] &= |0 - P[1\{s_{\hat{\beta}_n}(X) < 0\}]| \\
        &= |P_n[1\{s_{\hat{\beta}_n}(X) < 0\}] - P[1\{s_{\hat{\beta}_n}(X) < 0\}]| \\
        &<\sup_{\beta}|P_n[1\{s_{{\beta}_n}(X) < 0\}] - P[1\{s_{{\beta}_n}(X) < 0\}]| \to 0
    \end{align*}
    Thus, along any subsequence $\hat{\beta}_{n_k} \to \bar{\beta}$, $\P(s_{\bar{\beta}}(X) < 0) = 0$ by continuity, and $\bar{\beta} \in \B_0$ is population feasible. By definition of optimality of $\beta^*$ as the unique minimizer of the population problem, $P[m_{\bar{\beta}}] \geq P[m_{\beta^*}]$. Applying continuity and compactness: $\lim\inf_{n \to \infty} P[m_{\beta}(X)] \geq P[m_{\beta^*}(X)]$.
\end{proof}

The consistency result above is sufficient for the learned-policy CLT, since the empirical-process argument only
requires $\hat\beta\toP\beta^*$. The following proposition is a stronger characterization of the learned policy
parameters themselves. It is not needed for validity of the final OPAL estimator, but records conditions under which the
empirical optimizer also admits a standard $\sqrt n$-asymptotic normal expansion.

\paragraph{Application to the odds-ratio case.}
In the odds-ratio setting, we use a block-structured parameter $\beta=(\beta^{(1)},\beta^{(0)})$ and group-specific policies $\pi^{(k)}_{\beta^{(k)}}(X)=\exp\{-h(u(X))^\top\beta^{(k)}\}$. The objective and constraint in \eqref{eq:Qn_def}--\eqref{eq:budget_constraint_avg} become the corresponding sums over $i\in G_1$ and $i\in G_0$ (cf.\ Eq.~\eqref{eq:optim_full}). The preceding result applies componentwise to $(\hat\beta^{(1)},\hat\beta^{(0)})$.

\paragraph{$\sqrt{n}$-rate and asymptotic normality via KKT.}
We now impose additional local regularity conditions for characterizing the
learned policy parameter itself. These assumptions are stronger than those
needed for consistency.

\begin{assumption}[Local KKT regularity]\label{ass:kkt-regularity}
Let $s_\beta(U)=h(U)^\top\beta$. There exists a neighborhood
$N(\beta^*)\subset \operatorname{int}(\mathcal B)$ and constants
$0<\delta<M<\infty$ such that, for all $\beta\in N(\beta^*)$,
$$\delta \le s_\beta(U)\le M \qquad \text{a.s.}$$
The population budget constraint binds,
$$S(\beta^*)=\rho,$$
and the associated multiplier satisfies $\lambda^*>0$. The population moment
map
$$\Psi(\theta) = \begin{pmatrix} \nabla_\beta Q(\beta)+\lambda\nabla_\beta S(\beta)\\ S(\beta)-\rho \end{pmatrix}, \qquad \theta=(\beta,\lambda),$$
is continuously differentiable in a neighborhood of
$\theta^*=(\beta^*,\lambda^*)$, and its Jacobian
$$A:=\nabla_\theta\Psi(\theta^*)$$
is nonsingular. The empirical Jacobian converges locally uniformly to $A$, and
$$\sqrt n\,\Psi_n(\theta^*) \rightsquigarrow N(0,\Sigma).$$
Finally, $\sqrt n(\rho_n-\rho)=o(1)$.
\end{assumption}

Under this local regularity condition, the pointwise constraints
$$s_\beta(U_i)=h(U_i)^\top\beta\ge 0,\qquad i=1,\dots,n,$$
are inactive in a neighborhood of $\beta^*$. Therefore, by consistency of
$\hat\beta$, with probability tending to one, the only active constraint at
$\hat\beta$ is the budget constraint. Define the empirical Lagrangian
$$\mathcal L_n(\beta,\lambda) := Q_n(\beta)+\lambda\{S_n(\beta)-\rho_n\}, \qquad \lambda\ge 0,$$
where
$$S_n(\beta):=\frac1n\sum_{i=1}^n \exp\{-h(U_i)^\top\beta\}.$$
Then, with probability tending to one, there exists a multiplier
$\hat\lambda\ge 0$ such that $(\hat\beta,\hat\lambda)$ satisfies the local
KKT system
\begin{align}
\nabla_\beta Q_n(\hat\beta)
+
\hat\lambda\,\nabla_\beta S_n(\hat\beta)
&=0,
\label{eq:kkt_stationarity_beta}\\
S_n(\hat\beta)-\rho_n
&=0.
\label{eq:kkt_budget_active}
\end{align}
Equivalently, writing $\theta:=(\beta,\lambda)\in\mathbb R^{B+1}$, define
$$\Psi_n(\theta) := \begin{pmatrix} \nabla_\beta Q_n(\beta)+\lambda\,\nabla_\beta S_n(\beta)\\[3pt] S_n(\beta)-\rho_n \end{pmatrix}.$$
Then \eqref{eq:kkt_stationarity_beta}--\eqref{eq:kkt_budget_active} is
equivalent to
$$\Psi_n(\hat\theta)=0, \qquad \hat\theta:=(\hat\beta,\hat\lambda).$$

\begin{proposition}[$\sqrt{n}$-behavior of the learned policy parameters]\label{prop:beta_asymptotics_general}
Under Assumption~\ref{ass:policy_learning}(B1)--(B4) and Assumption~\ref{ass:kkt-regularity},
$$\sqrt{n}\big((\hat\beta,\hat\lambda)-(\beta^*,\lambda^*)\big)\ \toD\ \N\!\big(0,\ A^{-1}\Sigma(A^{-1})^\top\big),$$
where $A:=\nabla_\theta \Psi(\theta^*)$ is the Jacobian of the population moment map $\Psi(\theta):=\E[\Psi_n(\theta)]$ at $\theta^*=(\beta^*,\lambda^*)$, and $\Sigma:=\Var(\psi_{\theta^*}(W))$ is the covariance of the per-observation influence vector associated with $\Psi_n(\theta^*)$.
In particular, $\sqrt{n}(\hat\beta-\beta^*)=O_p(1)$.
\end{proposition}

\begin{proof}
By Lemma~\ref{lem:beta_consistency_general}, $\hat\theta$ lies in a shrinking neighborhood of $\theta^*$ with probability tending to one. By~\ref{ass:kkt-regularity}, a first-order Taylor expansion yields
$$0=\Psi_n(\hat\theta)=\Psi_n(\theta^*) + A(\hat\theta-\theta^*) + r_n,$$
where $r_n=o_p(\|\hat\theta-\theta^*\|)+o_p(n^{-1/2})$. Rearranging,
$$\sqrt{n}(\hat\theta-\theta^*) = -A^{-1}\sqrt{n}\,\Psi_n(\theta^*) + o_p(1).$$
Finally, conditional on any auxiliary data used to construct $\hat c(\cdot)$ (e.g.\ an independent pilot/burn-in sample, or cross-fitting so that $\hat c$ is treated as fixed when evaluating $Q_n$), $\Psi_n(\theta^*)$ is an empirical average of i.i.d.\ per-observation contributions. Hence the multivariate CLT yields
$$\sqrt{n}\,\Psi_n(\theta^*)\toD \N(0,\Sigma).$$
The result follows by Slutsky's theorem.
\end{proof}

\subsection{Proof of the batch CLT (Theorem~\ref{thm:batch_mean})}
\label{app:equicontinuity}

We prove Theorem~\ref{thm:batch_mean} by coupling all policies on a single probability space and applying an empirical process decomposition (c.f. the stochastic decomposition following Example~1.5 in \cite{epnotes}, and standard Donsker/equicontinuity results in \cite{vdv2}). 

\paragraph{Step 1: Coupling and notation.}
For each $i$, augment the observation by an independent uniform random variable:
$$O_i=(X_i,Y_i,V_i),\qquad V_i\sim\Unif(0,1) \indep (X_i,Y_i).$$
Fix a subgroup $k\in\{0,1\}$ and write
$$g_k(O):=\mathbf 1\{X\in G_k\}, \qquad P_k:=P(g_k)=\P(X\in G_k)>0.$$
For any policy parameter $\beta\in\mathcal B$, define the coupled labeling indicator
$$\xi_i(\beta):=\mathbf 1\{V_i\le \pi^{(k)}_\beta(X_i)\}, \qquad \pi^{(k)}_\beta(x)=\exp\{-h(u(x))^\top\beta\},$$
so that $\xi_i(\beta)\mid X_i\sim\Bern(\pi^{(k)}_\beta(X_i))$ for every $\beta$, and all policies are coupled through the same $\{V_i\}_{i=1}^n$.

Define the single-observation contribution
$$m_{k,\beta}(O) := g_k(O)\left[ f(X)+\frac{\xi(\beta)}{\pi^{(k)}_\beta(X)}\{Y-f(X)\} \right], \qquad \xi(\beta):=\mathbf 1\{V\le \pi^{(k)}_\beta(X)\}.$$
Let $\P_n$ and $P$ denote empirical and population measures on $O=(X,Y,V)$, and define the empirical process
$$\G_n:=\sqrt{n}(\P_n-P).$$
The subgroup estimator under policy parameter $\beta$ can be written as the ratio
$$\hat\mu_k(\beta)=\frac{\P_n m_{k,\beta}}{\P_n g_k}, \qquad \hat\mu_k^\pi=\hat\mu_k(\hat\beta_k), \qquad \hat\mu_k^{\pi^*}=\hat\mu_k(\beta_k^*).$$

\paragraph{Connection to Eq.~\eqref{eq:expansion}.}
Eq.~\eqref{eq:ratio_lin_batch} reduces the subgroup CLT to controlling the empirical process
$\G_n(m_{k,\hat\beta_k}-\mu_k g_k)$. Decomposing,
$$\G_n(m_{k,\hat\beta_k}-\mu_k g_k) = \underbrace{\G_n(m_{k,\beta_k^*}-\mu_k g_k)}_{\text{oracle fluctuation}} + \underbrace{\G_n(m_{k,\hat\beta_k}-m_{k,\beta_k^*})}_{\text{remainder}} + \underbrace{\sqrt{n}\{P m_{k,\hat\beta_k}-P m_{k,\beta_k^*}\}}_{\text{parameter-estimation cost}}.$$
The last term corresponds to $\sqrt{n}\{\psi(\hat\beta_k)-\psi(\beta_k^*)\}$ in Eq.~\eqref{eq:expansion}
with the identification $\psi(\beta):=P m_{k,\beta}$. By Lemma~\ref{lem:policy_invariance} (stated and proved below),
$P m_{k,\beta}$ is constant in $\beta$, hence this parameter-estimation term is identically zero.
Thus, the only additional work beyond the oracle CLT is to show the empirical-process remainder
$\G_n(m_{k,\hat\beta_k}-m_{k,\beta_k^*})$ is $o_p(1)$ via asymptotic equicontinuity and $\hat\beta_k\toP\beta_k^*$.

\paragraph{Step 2: Population invariance (parameter-estimation cost vanishes).}
\begin{lemma}[Policy-invariance of the population moment]\label{lem:policy_invariance}
Assume $\pi^{(k)}_\beta(X)>0$ almost surely. Then for every $\beta\in\mathcal B$,
$$P m_{k,\beta} = \E[m_{k,\beta}(O)] = \E[g_k(X)Y] = \mu_k\,P_k,$$
where $\mu_k=\E[Y\mid X\in G_k]$. In particular, $P m_{k,\hat\beta_k}-P m_{k,\beta_k^*}=0$ almost surely.
\end{lemma}

\begin{proof}
Since $V\perp (X,Y)$ and $\xi(\beta)=\mathbf 1\{V\le \pi^{(k)}_\beta(X)\}$,
$$\E[\xi(\beta)\mid X,Y]=\P(V\le \pi^{(k)}_\beta(X)\mid X)=\pi^{(k)}_\beta(X).$$
Hence
$$\E\!\left[\frac{\xi(\beta)}{\pi^{(k)}_\beta(X)}\{Y-f(X)\}\,\middle|\,X,Y\right]=Y-f(X).$$
Substituting into the definition of $m_{k,\beta}$ and applying iterated expectation yields
$$P m_{k,\beta} = \E[g_k(X)f(X)]+\E[g_k(X)\{Y-f(X)\}] =\E[g_k(X)Y] =\mu_k P_k.$$
\end{proof}

\begin{remark}[Relation to grid-based tuning arguments]\label{rem:grid_relation}
Our coupling via $\xi_i(\beta)=\mathbf 1\{V_i\le \pi_\beta(W_i)\}$ is a continuous-parameter analogue of coupling arguments used for grid-based tuning in \cite{active}. In the grid-based analysis the assumption is made that there exists $\eta^*\in\mathcal H$ (a finite set) such that $\P(\hat\eta\neq \eta^*)\to 0$, and concludes learned- and oracle-rule estimators coincide asymptotically because their difference is supported on the vanishing-probability event $\{\hat\eta\neq \eta^*\}$. In our setting, $\beta$ is continuous and learned by convex optimization, so exact selection consistency is not available. Instead, the coupling places all policies on one probability space and the learned-policy effect is controlled by asymptotic equicontinuity together with $\hat\beta_k\toP\beta_k^*$.
\end{remark}

\paragraph{Step 3: Ratio linearization.}
Since $P g_k=P_k>0$, a standard ratio expansion yields
\begin{equation}\label{eq:ratio_lin_batch}
\sqrt{n}\{\hat\mu_k(\hat\beta_k)-\mu_k\}
=
\frac{1}{P_k}\,\G_n\!\left(m_{k,\hat\beta_k}-\mu_k g_k\right)+o_p(1).
\end{equation}
Thus it suffices to study $\G_n(m_{k,\hat\beta_k}-\mu_k g_k)$.

\paragraph{Step 4: Oracle term.}
Decompose
$$\G_n\!\left(m_{k,\hat\beta_k}-\mu_k g_k\right) = \underbrace{\G_n\!\left(m_{k,\beta_k^*}-\mu_k g_k\right)}_{A_{n,k}} + \underbrace{\G_n\!\left(m_{k,\hat\beta_k}-m_{k,\beta_k^*}\right)}_{B_{n,k}}.$$
Because $\beta_k^*$ is non-random, the oracle term satisfies an ordinary empirical-process CLT:
$$A_{n,k} = \G_n\!\left(m_{k,\beta_k^*}-\mu_k g_k\right) \ \toD\ \N\!\left(0,\ \Var\!\left(m_{k,\beta_k^*}(O)-\mu_k g_k(O)\right)\right),$$
under $\E[(Y-f(X))^2]<\infty$ and uniform positivity of $\pi^{(k)}_{\beta_k^*}$.

\paragraph{Step 5: Remainder via asymptotic equicontinuity.}
We show $B_{n,k}=\G_n(m_{k,\hat\beta_k}-m_{k,\beta_k^*})=o_p(1)$.
Let $\mathcal M_k := \{m_{k,\beta}: \beta \in \mathcal B\}$. By Lemma~\ref{lem:beta_consistency_general}, 
$\hat\beta_k \toP \beta_k^*$. Therefore, to apply the stochastic equicontinuity result in 
Lemma~\ref{lem:donsker_equicont}, it remains to verify that $\beta\mapsto m_{k,\beta}$ is $L_2(P)$-continuous at $\beta_k^*$.

A convenient bound uses the coupling through $V$. Write $\pi_\beta=\pi^{(k)}_\beta$ and note that
$$\xi(\beta)-\xi(\beta^*)\neq 0 \ \Rightarrow\ V\in(\min\{\pi_\beta(X),\pi_{\beta^*}(X)\},\max\{\pi_\beta(X),\pi_{\beta^*}(X)\}],$$
so conditional on $(X,Y)$ the probability that the indicators differ is $|\pi_\beta(X)-\pi_{\beta^*}(X)|$.
Under bounded basis functions and compact $\mathcal B$, $\beta\mapsto \pi_\beta(x)$ is Lipschitz uniformly in $x$ in a neighborhood of $\beta^*$, hence
$\sup_x|\pi_\beta(x)-\pi_{\beta^*}(x)|\lesssim \|\beta-\beta^*\|$.
Moreover, uniform positivity implies $\inf_x \pi_\beta(x)\ge \underline\pi>0$ near $\beta^*$.
Using these facts and $\E[(Y-f(X))^2]<\infty$, one obtains
$$\|m_{k,\beta}-m_{k,\beta^*}\|_{L_2(P)}^2 \ \le\ C\,\E\!\left[g_k(X)(Y-f(X))^2\right]\cdot \sup_x|\pi_\beta(x)-\pi_{\beta^*}(x)| \ \to\ 0 \quad \text{as }\beta\to\beta^*,$$
for a constant $C<\infty$ depending only on $\underline\pi$.
Combining this $L_2(P)$-continuity with Lemma~\ref{lem:donsker_equicont} and the consistency
$\hat\beta_k\toP\beta_k^*$ from Lemma~\ref{lem:beta_consistency_general}, we obtain
$$B_{n,k} = \mathbb G_n(m_{k,\hat\beta_k}-m_{k,\beta_k^*}) = o_p(1).$$

\paragraph{Step 6: Conclude the subgroup CLT and joint normality.}
Combining Steps~3--5,
$$\sqrt{n}\{\hat\mu_k(\hat\beta_k)-\mu_k\} = \frac{1}{P_k}\,A_{n,k}+o_p(1) \toD \N\!\left(0,\ \frac{1}{P_k^2}\Var\!\left(m_{k,\beta_k^*}(O)-\mu_k g_k(O)\right)\right).$$
Applying the same argument to $k=0$ and $k=1$ yields joint asymptotic normality of
$(\hat\mu_1^\pi,\hat\mu_0^\pi)$. The asymptotic covariance is diagonal because the linear influence terms are supported on disjoint groups:
$g_1(O)g_0(O)=0$ almost surely. 

\begin{lemma}[Donsker and equicontinuity of the policy-indexed class]\label{lem:donsker_equicont}
Assume:
(i) $\mathcal B\subset\R^B$ is compact and $\sup_u\|h(u)\|_2<\infty$;
(ii) there exists $\underline\pi>0$ such that $\inf_{x,\beta\in\mathcal B}\pi^{(k)}_\beta(x)\ge \underline\pi$;
and (iii) $\E[(Y-f(X))^2]<\infty$.  
Then the class $\mathcal M_k=\{m_{k,\beta}:\beta\in\mathcal B\}$ is $P$-Donsker, and the empirical process $\G_n=\sqrt{n}(\P_n-P)$ is asymptotically equicontinuous on $\mathcal M_k$ with respect to the $L_2(P)$ metric. In particular, if $\hat\beta_k\toP\beta_k^*$, then
$$\G_n\!\left(m_{k,\hat\beta_k}-m_{k,\beta_k^*}\right)=o_p(1).$$
\end{lemma}

\begin{proof}[Proof sketch]
Write $O=(X,Y,V)$ with $V\sim\Unif(0,1)$ independent of $(X,Y)$. Note that
$$m_{k,\beta}(O) = g_k(X)f(X) + g_k(X)\,(Y-f(X))\cdot \frac{\mathbf 1\{V\le \pi^{(k)}_\beta(X)\}}{\pi^{(k)}_\beta(X)}.$$
The first term does not depend on $\beta$. For the second term, consider the indicator class
$$\mathcal I = \left\{ (x,v)\mapsto \mathbf 1\{v\le \pi^{(k)}_\beta(x)\}:\beta\in\mathcal B \right\}.$$
Since $\pi^{(k)}_\beta(x)=\exp\{-h(u(x))^\top\beta\}$, we can rewrite the threshold event, for $v > 0$, as
$$\mathbf 1\{v\le \pi^{(k)}_\beta(x)\} = \mathbf 1\{\log v + h(u(x))^\top\beta \le 0\}.$$
Because $h(u(x))^\top\beta$ ranges over a finite-dimensional vector space of functions indexed by $\beta\in\mathcal B$, the class of subgraphs
$\{(x,v)\mapsto \log v + h(u(x))^\top\beta:\beta\in\mathcal B\}$
is a VC-subgraph class, hence $\mathcal I$ is VC and therefore $P$-Donsker; see Chapter~2.5 of van der Vaart and Wellner \cite{vdv2}.

By Assumption (ii) in the statement, $1/\pi_\beta^{(k)}(X)$ is uniformly bounded over $\beta\in\mathcal B$. Hence the policy-indexed factor
$$\frac{\mathbf 1\{V\le \pi_\beta^{(k)}(X)\}}{\pi_\beta^{(k)}(X)}$$
is uniformly bounded by $1/\underline\pi$. Assumption (iii) then gives an $L_2(P)$ envelope for the product class, since
$$\left| g_k(X)(Y-f(X)) \frac{\mathbf 1\{V\le \pi_\beta^{(k)}(X)\}}{\pi_\beta^{(k)}(X)} \right| \le \frac{|Y-f(X)|}{\underline\pi}, \qquad \E[(Y-f(X))^2]<\infty.$$
Together with the VC-subgraph/Donsker property of the threshold class and standard preservation results for bounded transformations and square-integrable multipliers, this implies that the second term in $\mathcal M_k$ is $P$-Donsker.

It remains to note that $\beta\mapsto m_{k,\beta}$ is $L_2(P)$-continuous
at $\beta_k^*$. Indeed, the denominator is uniformly bounded away from zero,
$\pi_\beta^{(k)}$ is uniformly continuous in $\beta$, and
$$\P\!\left( \mathbf 1\{V\le \pi_\beta^{(k)}(X)\} \neq \mathbf 1\{V\le \pi_{\beta_k^*}^{(k)}(X)\} \,\middle|\, X \right) = |\pi_\beta^{(k)}(X)-\pi_{\beta_k^*}^{(k)}(X)|.$$
Together with Assumption (iii) which states $\E[(Y-f(X))^2]<\infty$, dominated convergence gives
$$\|m_{k,\beta}-m_{k,\beta_k^*}\|_{L_2(P)}\to 0 \qquad\text{as }\beta\to\beta_k^*.$$
Since $P$-Donsker classes are asymptotically equicontinuous in $L_2(P)$,
$\hat\beta_k\toP\beta_k^*$ implies
$$\mathbb G_n(m_{k,\hat\beta_k}-m_{k,\beta_k^*})=o_p(1).$$
\end{proof}

\subsection{Batch CLT for generalized asymptotically linear estimands}
\label{app:generalCLT}

This subsection extends the batch CLT argument in Appendix~\ref{app:equicontinuity} from subgroup means
to the broader EIF-based framework of Section~\ref{sec:general}. The key point is that, after coupling,
the learned-policy effect again enters only through an empirical-process remainder indexed by $\beta$,
and is controlled by the same Donsker/equicontinuity logic as in the mean case.

\paragraph{Setup.}
Let $\psi=\psi(P)$ admit an observed-data EIF of the form
$$\phi_{\pi}(X,\xi,\xi Y) = h(X)-\psi+\frac{\xi}{\pi(X)}\zeta(X,Y), \qquad \E[\zeta(X,Y)\mid X]=0,$$
under missing-at-random labeling $\xi\perp Y\mid X$ with $\xi\mid X\sim\Bern(\pi(X))$.
Write $c(X):=\E[\zeta(X,Y)^2\mid X]$.

Let $\{\pi_\beta:\beta\in\mathcal B\}$ be the spline policy class from Section~\ref{sec:optim}, and let
$\beta^*$ denote the unique population minimizer of $\E[c(X)/\pi_\beta(X)]$ over the feasible set.
Let $\hat\beta$ be the sample-level optimizer (Appendix~\ref{app:beta_asymptotics}) and define
$\pi^*:=\pi_{\beta^*}$ and $\hat\pi:=\pi_{\hat\beta}$.

Consider a one-step/TMLE-style estimator satisfying the usual asymptotic linearity remainder condition:
\begin{equation}\label{eq:general_AL_expansion}
\hat\psi^{\hat\pi}-\psi
=
\frac{1}{n}\sum_{i=1}^n
\Bigl\{
h(X_i)-\psi+\frac{\xi_i}{\hat\pi(X_i)}\zeta(X_i,Y_i)
\Bigr\}
+o_p(n^{-1/2}).
\end{equation}
(For instance, this holds under standard TMLE/one-step conditions with cross-fitting; see, e.g., \cite{vdv}.)

\paragraph{Coupling.}
Augment each observation with $V_i\sim\Unif(0,1)$ independent of $(X_i,Y_i)$ and define the coupled
label indicator $\xi_i(\beta):=\mathbf 1\{V_i\le \pi_\beta(X_i)\}$.
Define the policy-indexed single-observation map
$$m_{\beta}(O) := h(X)-\psi+\frac{\xi(\beta)}{\pi_\beta(X)}\zeta(X,Y), \qquad O=(X,Y,V).$$
Let $\P_n$ and $P$ denote empirical and population measures on $O$, and $\G_n:=\sqrt{n}(\P_n-P)$.

\paragraph{Oracle CLT and learned-policy remainder.}
Using \eqref{eq:general_AL_expansion} and the identity
$\frac{1}{\sqrt{n}}\sum_{i=1}^n m_{\hat\beta}(O_i)=\G_n m_{\hat\beta}$, we have
$$\sqrt{n}(\hat\psi^{\hat\pi}-\psi) = \G_n m_{\hat\beta}+o_p(1) = \underbrace{\G_n m_{\beta^*}}_{\text{oracle fluctuation}} + \underbrace{\G_n(m_{\hat\beta}-m_{\beta^*})}_{\text{learned-policy remainder}} +o_p(1).$$
Since $\beta^*$ is non-random, $\G_n m_{\beta^*}\toD \N(0,\Var(m_{\beta^*}(O)))$ by the classical CLT
under $\E[m_{\beta^*}(O)^2]<\infty$. It remains to show $\G_n(m_{\hat\beta}-m_{\beta^*})=o_p(1)$.

\paragraph{Equicontinuity.}
Assume (as in Appendix~\ref{app:equicontinuity}) that the policy-indexed class
$\mathcal M:=\{m_\beta:\beta\in\mathcal B\}$ is $P$-Donsker and $\G_n$ is asymptotically equicontinuous
on $\mathcal M$ in the $L_2(P)$ metric (this is implied by the same VC-subgraph argument as
Lemma~\ref{lem:donsker_equicont}, with the square-integrable multiplier $(Y-f(X))$ replaced by
$\zeta(X,Y)$ and the moment condition $\E[\zeta(X,Y)^2]<\infty$).
Then $\hat\beta\toP\beta^*$ (Lemma~\ref{lem:beta_consistency_general}) implies
$\G_n(m_{\hat\beta}-m_{\beta^*})=o_p(1)$.

\paragraph{Limiting variance.}
Finally, we compute $\Var(m_{\beta^*}(O))$.
By $\E[\zeta\mid X]=0$ and $\E[\xi(\beta^*)\mid X]=\pi^*(X)$,
$$\E\!\left[\frac{\xi(\beta^*)}{\pi^*(X)}\zeta(X,Y)\,\middle|\,X\right]=0,$$
so $\Cov(h(X)-\psi,\frac{\xi(\beta^*)}{\pi^*(X)}\zeta(X,Y))=0$. Moreover,
$$\Var\!\left(\frac{\xi(\beta^*)}{\pi^*(X)}\zeta(X,Y)\,\middle|\,X\right) = \E\!\left[\frac{\xi(\beta^*)}{\pi^*(X)^2}\zeta(X,Y)^2\,\middle|\,X\right] = \frac{c(X)}{\pi^*(X)}.$$
Therefore,
$$\Var(m_{\beta^*}(O)) = \Var(h(X))+\E\!\left[\frac{c(X)}{\pi^*(X)}\right].$$

\begin{theorem}[Batch CLT under the learned policy for generalized estimands]\label{thm:generalCLT}
Under the policy-learning conditions (B1)--(B5), the asymptotic linearity remainder
\eqref{eq:general_AL_expansion}, and the Donsker/equicontinuity conditions described above,
$$\sqrt{n}(\hat\psi^{\hat\pi}-\psi)\ \toD\ \N\!\left(0,\ \sigma_*^2\right), \qquad \sigma_*^2:=\Var(h(X))+\E\!\left[\frac{c(X)}{\pi^*(X)}\right].$$
\end{theorem}

\begin{remark}[The same convex program applies]
Because $\Var(h(X))$ does not depend on $\pi$, minimizing the asymptotic variance within the policy class
reduces to minimizing $\E[c(X)/\pi_\beta(X)]$, which is exactly the objective targeted by the convex program
in Section~\ref{sec:optim}.
\end{remark}

\subsection{Sequential CLT}\label{app:seq_general}
In the online setting of Section~\ref{sec:sequential} and Algorithm~\ref{alg:seq}, both the predictor $f_t$ and the
(unclipped) policy proposal $\pi_t$ may be updated based on the observed history, and the realized labeling
probabilities are clipped to enforce the budget constraint. As a result, the estimator is no longer an i.i.d.\ average. We can analyze the sequential estimator via a martingale CLT: we express the
estimation error as a sum of (centered) martingale increments and assume a convergent variance process and a Lindeberg condition. This is precisely the argument in  Zrnic and Cand\`es~\cite{active}. The stochastic equicontinuity decomposition is no longer needed here, since we are conditioning on the past history, which is used to construct the labeling policy.

\section{Power tuning and cross-fitting theory}\label{app:add-on}
Throughout this section, we work under the assumptions from
Section~\ref{sec:setup} / Appendix~\ref{app:CLT}: the label indicators are independent of labels conditional on $X_i$, positivity of the labeling policy and the finite-moment conditions needed for the relevant covariance and variance
quantities to be well-defined (Assumption~\ref{ass:basic_sampling}). When nuisance functions are estimated, we also assume the
same pilot-sample conditions used in the main asymptotic results
(Assumption~\ref{ass:general_eif}; with a modified version for cross-fitting introduced in one of the subsections below), so that plug-in estimates of the covariance quantities below are consistent.

Power tuning follows the control-variate idea from PPI++~\cite{ppi++}. For a fixed feasible labeling policy $\pi$, the
parameter $\lambda$ interpolates between a label-only Horvitz--Thompson estimator (see \cite{horvitz1952generalization, cochran1977sampling}) and a prediction-assisted estimator.
Because the estimator remains unbiased for any fixed $\lambda$, we choose $\lambda$ to minimize variance.

\subsection{Power tuning for mean estimation}\label{app:powertuning_mean}

We start with the simplest case of mean estimation. Let $\psi=\E[Y]$ and fix a policy $\pi(X)=\P(\xi=1\mid X)$. Given a predictor $f$, define
\begin{equation}\label{eq:power_mean_app}
\hat\mu^{\pi,\lambda} := \frac{1}{n}\sum_{i=1}^n \left[ \lambda f(X_i) + \frac{\xi_i}{\pi(X_i)}\{Y_i-\lambda f(X_i)\}\right].
\end{equation}
Since $\xi_i \indep Y_i \mid X_i$, this estimator is unbiased for $\mu=\E[Y]$ for every fixed $\lambda$:
$$\E\!\left[\lambda f(X)+\frac{\xi}{\pi(X)}\{Y-\lambda f(X)\} \,\middle|\, X,Y \right] = Y.$$
Equivalently, \eqref{eq:power_mean_app} can be written as
$$\hat\mu^{\pi,\lambda} = \hat\mu^{\mathrm{HT}} + \lambda\{\bar f-\hat f^{\mathrm{HT}}\},$$ where
$$\hat\mu^{\mathrm{HT}}:=\P_n\!\left[\frac{\xi}{\pi}Y\right], \quad \hat f^{\mathrm{HT}}:=\P_n\!\left[\frac{\xi}{\pi}f\right], \quad \text{and} \quad \bar f:=\P_n[f].$$
The term $\bar f-\hat f^{\mathrm{HT}}$ is mean-zero and therefore acts as a control variate.

Let $A:=\frac{\xi}{\pi(X)}Y$ and $C:=f(X)-\frac{\xi}{\pi(X)}f(X)$.
Then $\hat\mu^{\pi,\lambda}=\P_n(A+\lambda C)$, and the variance-minimizing value is
\begin{equation}\label{eq:lambda_control_app}
\lambda^* = -\frac{\Cov(A,C)}{\Var(C)}.
\end{equation}
Under Assumption~\ref{ass:basic_sampling}, the denominator is finite and nonzero unless the control
variate is degenerate. In practice, we estimate $\lambda^*$ by replacing the population covariance and variance in
\eqref{eq:lambda_control_app} with their sample analogues, using pilot data, burn-in data, or cross-fitted estimates.

\subsection{Groupwise means and log odds ratios}\label{app:powertuning_odds}

For the log odds ratio, we apply the same construction separately within each group. Let $G_k$ denote the subgroup indices, $n_k=|G_k|$, and $\pi^{(k)}$ the group-specific labeling policy. For $k\in\{0,1\}$, define
\begin{equation}\label{eq:power_group_app}
\hat\mu_k^{\pi,\lambda_k} = \frac{1}{n_k}\sum_{i\in G_k}
\left[\lambda_k f(X_i) + \frac{\xi_i}{\pi^{(k)}(X_i)}\{Y_i-\lambda_k f(X_i)\}\right].
\end{equation}
Again, for each fixed $\lambda_k$, the estimator is unbiased for the empirical subgroup mean by independence of $\xi_i$ with $Y_i$ conditional on $X_i$: 
$$\E\!\left[\hat\mu_k^{\pi,\lambda_k}\mid X_{1:n},Y_{1:n}\right] = \frac{1}{n_k}\sum_{i:Z_i=k}Y_i.$$ This converges to $\mu_k = \E[Y \mid Z = k]$ by a standard law of large numbers argument.

The optimal groupwise tuning parameter is the corresponding control-variate coefficient. Writing
$$A_k:=\mathbf 1\{X\in G_k\}\frac{\xi}{\pi^{(k)}(X)}Y, \qquad C_k:=\mathbf 1\{X\in G_k\} \left[f(X)-\frac{\xi}{\pi^{(k)}(X)}f(X)\right],$$
the variance-minimizing coefficient is
\begin{equation}\label{eq:lambda_group_control_app}
\lambda_k^* = -\frac{\Cov(A_k,C_k)}{\Var(C_k)}.
\end{equation}
Equivalently, this may be written in the weighted covariance form used in PPI++\cite{ppi++}:
\begin{equation}\label{eq:lambda_group_ppi_app}
\lambda_k^* = \frac{\Cov(\tilde f_k,\tilde Y_k)}{(1+\rho)\Var(\tilde f_k)}, \qquad \rho:=\frac{n_{\budget}}{n},
\end{equation}
where $\tilde f_k$ and $\tilde Y_k$ denote the corresponding importance-weighted prediction and label variables. The precise form of the finite-sample factor depends on the normalization used for the full prediction average and labeled correction; in our implementation we use the PPI++ plug-in estimator matching the estimator in \eqref{eq:power_group_app}.

The power-tuned log odds ratio estimator is then
$$\hat\theta^{\pi,\lambda} = \log\frac{\hat\mu_1^{\pi,\lambda_1}}{1-\hat\mu_1^{\pi,\lambda_1}} - \log\frac{\hat\mu_0^{\pi,\lambda_0}}{1-\hat\mu_0^{\pi,\lambda_0}}.$$
Under the same assumptions used for Corollary~\ref{corr:batch_odds}, its asymptotic distribution follows from the joint CLT for $(\hat\mu_1^{\pi,\lambda_1},\hat\mu_0^{\pi,\lambda_0})$ and the Delta method.

\subsection{General EIF-based estimands}\label{app:powertuning_eif}

The same control-variate construction extends to the general EIF setting of Section~\ref{sec:general}. Suppose the
full-data EIF admits the decomposition
$$\phi_{\mathrm{full}}(X,Y)=h(X)-\psi+\zeta(X,Y), \qquad \E[\zeta(X,Y)\mid X]=0,$$
so that the pseudo-outcome $W:=h(X)+\zeta(X,Y)$ satisfies $\E[W] = \phi$. Under Assumption~\ref{ass:basic_sampling}, we define the estimator
\begin{equation}\label{eq:power_eif_app}
\hat\psi^{\pi,\lambda} = \frac{1}{n}\sum_{i=1}^n\left[\lambda h(X_i)
+ \frac{\xi_i}{\pi(X_i)}\{W_i-\lambda h(X_i)\}\right], \qquad W_i:=h(X_i)+\zeta(X_i,Y_i).
\end{equation}
For any fixed $\lambda$, this estimator is unbiased for $\psi$, since $\E[W] = \psi$ and
$$\E\!\left[\lambda h(X)+\frac{\xi}{\pi(X)}\{W-\lambda h(X)\} \,\middle|\,X,Y\right] = W.$$ 
Setting $\lambda = 1$ recovers the usual one-step / EIF estimator,
$$\hat\psi^{\pi,1} = \frac{1}{n}\sum_{i=1}^n \left[h(X_i) + \frac{\xi_i}{\pi(X_i)}\zeta(X_i,Y_i)\right],$$
while $\lambda = 0$ corresponds to the label-only Horvitz--Thompson estimator of the pseudo-outcome $W$.

Then, with $A = \frac{\xi}{\pi(X)}W$ and $C = h(X) - \frac{\xi}{\pi(X)}h(X)$, we have $\hat\psi^{\pi, \lambda} = \P_n(A + \lambda C)$, and the variance minimizing coefficient is
\begin{equation}\label{eq:lambda_eif_app}
\lambda^* = -\frac{\Cov(A,C)}{\Var(C)}.
\end{equation}
Under Assumption~\ref{ass:basic_sampling}, this coefficient is well-defined whenever the control variate
is nondegenerate. In practice, $h$, $\zeta$, and hence $W$ are replaced by their nuisance estimates, and $\lambda^*$ is estimated by plug-in covariance and variance estimates under the same nuisance-estimation conditions used for the one-step/TMLE results (Assumption~\ref{ass:general_eif}).

\subsection{Cross-fitting theory for odds ratio estimation} \label{app:crossfitting_OR}
We now justify the cross-fitting implementation introduced in Algorithm~\ref{alg:crossfitting}. As is standard in cross-fitting procedures, the key idea is the out-of-fold construction: the predictor used to evaluate observation $i$ is trained without using the label $Y_i$, with potential finetuning done on independent, out-of-fold data points. Thus, the correction term remains conditionally mean-zero.

For simplicity, we state the result for the two-fold algorithm. The same argument applies to any fixed number of folds. Let $\mathcal I_1,\mathcal I_2$ be a random split independent of the data, with
$|\mathcal I_\ell|/n\to 1/2$. Let $\hat f^{(-1)}$ be trained using labeled observations in $\mathcal I_2$ and evaluated on $\mathcal I_1$, and let $\hat f^{(-2)}$ be trained using labeled observations in $\mathcal I_1$ and evaluated on $\mathcal I_2$. Then, define the out-of-fold predictor $$\hat f_i^{\cf} := \hat f^{(-1)}(X_i)\mathbf 1\{i\in\mathcal I_1\} + \hat f^{(-2)}(X_i)\mathbf 1\{i\in\mathcal I_2\}.$$
The cross-fitted pseudo-outcome is 
$$T_i^{\cf} := \hat f_i^{\cf} + \frac{\xi_i}{\pi(X_i)}\{Y_i-\hat f_i^{\cf}\},$$ where $\pi$ is the global policy learned before the labels are used for cross-fitting. For $k \in \{0,1\}$, let $g_k(X)=\mathbf 1\{X\in G_k\}$, $p_k=\P(X\in G_k)$, and
$\hat\mu_k^{\cf} = \frac{1}{n_k} \sum_{i:Z_i=k}T_i^{\cf}$, $n_k=\sum_{i=1}^n g_k(X_i).$
\begin{assumption}
    [Cross-fitted nuisance stability]\label{ass:crossfit_stability}
In addition to the standing batch assumptions for the learned policy $\hat\pi$, assume that the out-of-fold
predictors are stable in $L_2(P_X)$; that is, there exists a square-integrable function $f_\infty$ such that
$$\max_{\ell=1,2}\|\hat f^{(-\ell)} - f_\infty\|_{L_2(P_X)}\toP 0.$$
Assume also the same positivity and moment conditions as in the batch CLT, so that $\hat\pi(X)$ is bounded away from zero with probability tending to one and the displayed variances are finite.
\end{assumption}

\begin{remark}
Assumption~\ref{ass:crossfit_stability} is a convenient sufficient condition, but not a necessary one. The proof only requires the out-of-fold nuisance remainder to be $o_p(1)$ after $\sqrt{n}$-scaling (conditional on the training folds used to construct the out-of-fold predictors) and the cross-fitted influence-function array satisfies a Lindeberg--Feller CLT. This is the standard
sample-splitting/cross-fitting argument in which the nuisance estimates are treated as fixed after
conditioning on the auxiliary sample; see, e.g., Chernozhukov et al.~\cite[Thm.~3.1, proof in App.]{dml}. For the underlying triangular-array CLT, see van der Vaart~\cite[Prop.~2.27]{vdv} or Billingsley~\cite[Thm.~27.2]{billingsley}.
\end{remark}

\noindent Let $\pi_k^*(X)$ denote the population-limit policy for subgroup $k$, $\pi_k^*(x) := \pi^{(k)}_{\beta_k^*}(x)$ denote the population-optimal oracle policy in the proof of Theorem~\ref{thm:batch_mean}. Define
$$T_k^*(X,\xi,Y) := f_\infty(X)+\frac{\xi}{\pi_k^*(X)}\{Y-f_\infty(X)\},$$
where $\xi\mid X,Z=k\sim\Bern(\pi_k^*(X))$, and set
$V_{k,\cf} := \Var\!\left(g_k(X)\{A_k^*(X,\xi,Y)-\mu_k\} \right).$

\begin{theorem}[Batch CLT with two-fold cross-fitting]\label{thm:batch_crossfit_mean}
Under the standard batch assumptions, the learned-policy conditions of Theorem~\ref{thm:batch_mean}, and
Assumption~\ref{ass:crossfit_stability},
$$\sqrt{n} \begin{pmatrix} \hat\mu_1^{\cf}-\mu_1\\ \hat\mu_0^{\cf}-\mu_0 \end{pmatrix} \toD \N\left(0, \begin{pmatrix} V_{1,\cf}/p_1^2 & 0\\ 0 & V_{0,\cf}/p_0^2 \end{pmatrix}\right).$$
Consequently, for $\mu_0,\mu_1\in(0,1)$, the cross-fitted log odds ratio estimator
$$\hat\theta^{\cf} = \log\frac{\hat\mu_1^{\cf}}{1-\hat\mu_1^{\cf}} - \log\frac{\hat\mu_0^{\cf}}{1-\hat\mu_0^{\cf}}$$
satisfies $$\sqrt{n}(\hat\theta^{\cf}-\theta) \toD \N(0,V_{\theta,\cf}), \quad \text{where } \quad V_{\theta,\cf} = \frac{V_{1,\cf}}{p_1^2\mu_1^2(1-\mu_1)^2} + \frac{V_{0,\cf}}{p_0^2\mu_0^2(1-\mu_0)^2}.$$
Moreover, the empirical influence-function variance estimator in Algorithm~\ref{alg:crossfitting} is consistent for $V_{\theta,\cf}$.
\end{theorem}

\begin{proof}
We prove the result for a fixed subgroup $k$. Let
$$\hat T_i^{\cf} := \hat f_i^{\cf} + \frac{\xi_i}{\hat\pi(X_i)}\{Y_i-\hat f_i^{\cf}\}.$$
Since $n_k/n\toP p_k$,
$$\sqrt{n}(\hat\mu_k^{\cf}-\mu_k) = \frac{1}{p_k}\sqrt{n}\,P_n\left[g_k(X)\{\hat A^{\cf}-\mu_k\}\right]+o_p(1).$$
We compare $\hat T_i^{\cf}$ to the oracle contribution based on $f_\infty$ and the population-limit policy $\pi_k^*$:
$$T_i^* = f_\infty(X_i) + \frac{\xi_i}{\pi_k^*(X_i)} \{Y_i-f_\infty(X_i)\}.$$
It is enough to show
$$\sqrt{n}\,P_n\left[g_k(X)\{\hat T^{\cf}-T^*\}\right] = o_p(1).$$
First, we consider the error from estimating $f$, temporarily keeping the policy fixed. On fold $\ell$, write $\delta_\ell(x)=\hat f^{(-\ell)}(x)-f_\infty(x)$. The corresponding term is
$$R_{n,k,\ell} = \frac{1}{\sqrt{n}}\sum_{i\in\mathcal I_\ell}g_k(X_i)\left(1-\frac{\xi_i}{\pi_k^*(X_i)}\right)\delta_\ell(X_i).$$
Conditional on the training data used to construct $\hat f^{(-\ell)}$ and on the covariates in the evaluation fold, the summands have mean zero because $\E[\xi_i\mid X_i]=\pi_k^*(X_i)$. Their conditional variance is bounded by
$$C\cdot\frac{1}{n}\sum_{i\in\mathcal I_\ell} \delta_\ell(X_i)^2$$
for a constant $C<\infty$ depending only on the positivity bound. Since $\hat f^{(-\ell)}$ is trained without labels from $\mathcal I_\ell$, the evaluation-fold covariates are independent of the fitted function, and the last display is $o_p(1)$ by Assumption~\ref{ass:crossfit_stability}. Hence $R_{n,k,\ell}=o_p(1)$ for each fold, and since the number of folds is fixed, the full cross-fitting error from replacing $f_\infty$ by $\hat f_i^{\cf}$ is $o_p(1)$.

Second, replacing the learned policy $\hat\pi$ by its population limit $\pi^*$ contributes $o_p(1)$ by the same learned-policy stochastic equicontinuity argument used in the proof of Theorem~\ref{thm:batch_mean} (see Appendix~\ref{app:equicontinuity}), now applied with the fixed limiting predictor $f_\infty$. Therefore
$$\sqrt{n}(\hat\mu_k^{\cf}-\mu_k) = \frac{1}{p_k}\frac{1}{\sqrt{n}} \sum_{i=1}^n g_k(X_i)\{T_i^*-\mu_k\}+o_p(1).$$
The summands in the leading term are i.i.d. mean-zero variables with variance $V_{k,\cf}$, so the univariate CLT gives the desired marginal limit. Joint normality for $k=0,1$ follows by the multivariate CLT, and the covariance is zero because $g_1(X)g_0(X)=0$ almost surely. The log odds ratio result follows by the Delta method.

Finally, the variance estimator in Algorithm~\ref{alg:crossfitting} is the empirical second moment of the plug-in influence-function contributions. The preceding $L_2$ approximation, consistency of $\hat\mu_k^{\cf}$, and $\hat p_k\toP p_k$ imply that these plug-in contributions converge in $L_2$ to the limiting influence function for $\theta$. A law of large numbers argument then gives consistency of $\hat V_{\theta}^{\cf}$.
\end{proof}

\subsection{Cross-fitting theory for general EIF estimands}\label{app:crossfitting_EIF}

The same argument extends to the general EIF setting of Section~\ref{sec:general}. Let $\mathcal I_1,\mathcal I_2$ be the same two-fold split, and let
$(\hat h^{(-\ell)},\hat\zeta^{(-\ell)})$ denote nuisance estimates trained without using labels from fold $\mathcal I_\ell$. Define the out-of-fold nuisance estimates $$\hat h_i^{\cf} = \hat h^{(-1)}(X_i)\mathbf 1\{i\in\mathcal I_1\} + \hat h^{(-2)}(X_i)\mathbf 1\{i\in\mathcal I_2\},$$
and similarly $\hat\zeta_i^{\cf}=\hat\zeta^{(-\ell(i))}(X_i,Y_i)$. The cross-fitted one-step estimator is $$\hat\psi^{\cf} = \frac{1}{n}\sum_{i=1}^n\left[ \hat h_i^{\cf} + \frac{\xi_i}{\hat\pi(X_i)}\hat\zeta_i^{\cf} \right].$$

\begin{assumption}[Cross-fitted EIF expansion]\label{ass:crossfit_eif}
The cross-fitted nuisance estimates and learned policy satisfy $$\hat\psi^{\cf}-\psi = \frac{1}{n}\sum_{i=1}^n \left[h(X_i)-\psi + \frac{\xi_i}{\pi^*(X_i)}\zeta(X_i,Y_i) \right] + o_p(n^{-1/2}),$$
where $\pi^*=\pi_{\beta^*}$ is the population-limit policy.
\end{assumption}

\begin{theorem}[Cross-fitted CLT for general EIF estimands]\label{thm:crossfit_general}
Under the standing EIF assumptions and Assumption~\ref{ass:crossfit_eif}, $$\sqrt{n}(\hat\psi^{\cf}-\psi) \toD \N\left(0,\sigma_{\cf}^2\right),$$
where $$\sigma_{\cf}^2 = \Var(h(X)) + \E\left[\frac{c(X)}{\pi^*(X)}\right], \qquad c(X)=\E[\zeta(X,Y)^2\mid X].$$
\end{theorem}

\begin{proof}
By Assumption~\ref{ass:crossfit_eif}, $$\sqrt{n}(\hat\psi^{\cf}-\psi) = \frac{1}{\sqrt{n}}\sum_{i=1}^n \left[h(X_i)-\psi + \frac{\xi_i}{\pi^*(X_i)} \zeta(X_i,Y_i)\right]+o_p(1).$$
The leading summands are i.i.d. and mean zero, so asymptotic normality follows from the classical CLT. The variance
identity follows from $\E[\zeta(X,Y)\mid X]=0$ and $\E[\xi\mid X]=\pi^*(X)$:
$$\Var\left(h(X)-\psi+\frac{\xi}{\pi^*(X)}\zeta(X,Y)\right) = \Var(h(X)) + \E\left[\frac{c(X)}{\pi^*(X)}\right].$$
\end{proof}

\section{Experimental performance evaluation details}

\subsection{Effective sample size for log odds ratio}\label{app:ESS}

Suppose there are $n$ total data points, all unlabeled at the start. Let $\hat{V}^{\mathrm{human}} / n_{\text{eff}}$ be an estimate of estimator variance when using only $n_{\text{eff}}$ labeled data points. Similarly, let $\hat{V}^{\text{method}}/n$ be an estimate of estimator variance for a different estimator which uses predicted labels as well (and thus is divided by $n$ since all $n$ data points are used in some fashion). We define $n_{\text{eff}}$ to be the number of human labeled data points needed for the confidence interval constructed with only labels to match the width of an interval constructed for an estimator $\hat{\theta}^{\text{method}}$ using the labeled data as well as predicted labels for other data points, satisfying $\Var(\hat{\theta}^{\text{method}}) = \Var(\hat{\theta}^{\text{human}}_{n_{\text{effective}}})$. Thus, we evaluate performance of any estimator $\hat{\theta}^{\text{method}}$ by calculating the associated effective sample size: 
$$\frac{\hat{V}^{\text{method}}}{n} = \frac{\hat{V}^{\mathrm{human}}}{n_{\text{eff}}} \implies n_{\text{eff}} = n \cdot \frac{\hat{V}^{\text{human}}}{ \hat{V}^{\text{method}}}$$

For estimating the log odds ratio with binary labels $Y_i \in \{0,1\}$, in a population of $n$ data points, $$\hat{V}^{\text{human}} = n \cdot \bigg(\frac{1}{\sum_{i\in G_1} (1 - Y_i)} + \frac{1}{\sum_{i \in G_1} Y_i} + \frac{1}{\sum_{i \in G_0} (1 - Y_i)} + \frac{1}{\sum_{i \in G_0} Y_i}\bigg),$$ where $G_1$ and $G_0$ perfectly separate the population of $n$ points, so $G_1 \cup G_0 = [n]$ and $|G_1| + |G_0| = n$. This is an oracle quantity which is used for the purpose of method evaluation. In a real-world setting where we do not have access to the true labels, then we can approximate $\hat{V}^{\text{human}}$ with the sample version of the component inside the parentheses using existing labeled data from a training split, or in the case of sequential inference where there is no training portion, based on the same burn-in labeled sample used to fit predictive model $f$.

\subsection{Finite-population coverage calibration for in-sample parameter evaluation}
\label{app:overcoverage}

Section~\ref{sec:overcoverage} explains why coverage against a fixed evaluation-set parameter $\theta_N$ can be conservative when intervals are calibrated to the superpopulation parameter $\theta(P)$. It also derives the basic mean-estimation calculation under a uniform labeling policy and introduces the finite-population variance estimator $\hat V_{\mathrm{HT}}$. Here we provide the corresponding details for nonuniform policies and for the log odds-ratio estimand used in the real-data experiments.

Let $F_N=\{(X_i,Y_i,Z_i)\}_{i=1}^N$
denote the fixed evaluation population, where $Z_i \in \{0,1\}$ is the group indicator. We condition throughout on $F_N$. Let $\pi_i = \pi(X_i)$ be the realized labeling probability for unit $i$, and suppose
$\xi_i\mid F_N \sim \Bern(\pi_i)$ independently across units. Let $a_i$ denote the augmentation term used by the estimator. For example, $a_i=f(X_i)$ for the base estimator, $a_i=\lambda f(X_i)$ for a power-tuned estimator, and $a_i=\hat f_i^{\mathrm{cf}}(X_i)$ for a cross-fitted estimator. Define the difference $R_i:=Y_i-a_i.$ For a finite-population mean,
$\theta_N=\frac{1}{N}\sum_{i=1}^N Y_i,$ the augmented inverse-probability weighted estimator is
$$\hat\theta = \frac{1}{N}\sum_{i=1}^N \left[a_i+\frac{\xi_i}{\pi_i}(Y_i-a_i)\right].$$
As shown in Section~\ref{sec:overcoverage},
$$\hat\theta-\theta_N = \frac{1}{N}\sum_{i=1}^N \left(\frac{\xi_i}{\pi_i}-1\right)R_i,$$
so the conditional finite-population variance is
$$V_{\mathrm{act}} = \Var(\hat\theta-\theta_N\mid F_N) = \frac{1}{N^2} \sum_{i=1}^N \frac{1-\pi_i}{\pi_i}R_i^2.$$
Because $R_i$ is observed only when $\xi_i = 1$, we estimate this quantity by
$$\hat V_{\mathrm{HT}} = \frac{1}{N^2} \sum_{i=1}^N \xi_i\frac{1-\pi_i}{\pi_i^2}R_i^2,$$
which satisfies
$\E[\hat V_{\mathrm{HT}}\mid F_N]=V_{\mathrm{act}}.$

\paragraph{Log odds-ratio calibration.}
For the log odds ratio, define the finite-population subgroup means
$$\mu_{k,N} = \frac{1}{N_k}\sum_{i:Z_i=k}Y_i, \qquad N_k=\sum_{i=1}^N \mathbf 1\{Z_i=k\}, \qquad k \in \{0,1\}.$$
The finite-population log odds ratio is
$$\theta_N = g(\mu_{1,N},\mu_{0,N}) = \log\frac{\mu_{1,N}}{1-\mu_{1,N}} - \log\frac{\mu_{0,N}}{1-\mu_{0,N}}.$$
The corresponding augmented estimator of the $k$th subgroup mean is
$$\hat\mu_k = \frac{1}{N_k} \sum_{i:Z_i=k}\left[ a_i+\frac{\xi_i}{\pi_i}(Y_i-a_i)\right].$$
Its conditional error is
$$\hat\mu_k-\mu_{k,N} = \frac{1}{N_k} \sum_{i:Z_i=k} \left(\frac{\xi_i}{\pi_i}-1\right)R_i,$$
with conditional variance
$$V_{k,\mathrm{act}} = \Var(\hat\mu_k-\mu_{k,N}\mid F_N) = \frac{1}{N_k^2} \sum_{i:Z_i=k} \frac{1-\pi_i}{\pi_i}R_i^2.$$
We estimate this variance using the groupwise Horvitz--Thompson estimator
$$\hat V_{k,\mathrm{HT}} = \frac{1}{N_k^2} \sum_{i:Z_i=k} \xi_i\frac{1-\pi_i}{\pi_i^2}R_i^2, \qquad \E[\hat V_{k,\mathrm{HT}}\mid F_N] = V_{k,\mathrm{act}}.$$
A first-order Delta-method approximation around the finite-population subgroup means gives
$$\hat\theta-\theta_N \approx d_{1,N}(\hat\mu_1-\mu_{1,N}) + d_{0,N}(\hat\mu_0-\mu_{0,N}),$$
where
$$d_{1,N} = \frac{1}{\mu_{1,N}(1-\mu_{1,N})}, \qquad d_{0,N} = -\frac{1}{\mu_{0,N}(1-\mu_{0,N})}.$$
Since the two subgroups are disjoint and labels are queried independently across units,
$$V_{\theta,\mathrm{act}} \approx d_{1,N}^2V_{1,\mathrm{act}} + d_{0,N}^2V_{0,\mathrm{act}}.$$
In practice, we plug in the estimated subgroup means:
$$\hat d_1 = \frac{1}{\hat\mu_1(1-\hat\mu_1)}, \qquad \hat d_0 = -\frac{1}{\hat\mu_0(1-\hat\mu_0)}.$$
Thus the finite-population variance estimate for the log odds ratio is
$$\hat V_{\theta,\mathrm{HT}} = \hat d_1^2\hat V_{1,\mathrm{HT}} + \hat d_0^2\hat V_{0,\mathrm{HT}}.$$
The finite-population-calibrated interval reported in the real-data coverage evaluation is therefore
$$\hat\theta \pm z_{1-\alpha/2}\sqrt{\hat V_{\theta,\mathrm{HT}}}.$$

\paragraph{Implementation in the real-data experiments.}
For each replicate in the real-data experiments, we compute both the usual Wald interval and the finite-population
calibrated interval. The implementation is:
\begin{enumerate}
    \item Compute the point estimate $\hat \theta$ using the queried labels and the chosen augmentation $a_i$.
    \item Compute the usual variance estimate $V_{\theta,\mathrm{int}}$ used in the superpopulation-style Wald interval.
    \item For each queried unit, compute $R_i=Y_i-a_i$. Unqueried units do not contribute to
    $\hat V_{k,\mathrm{HT}}$, since their contribution is multiplied by $\xi_i = 0$.
    \item Compute $\hat V_{1,\mathrm{HT}}$ and $\hat V_{0,\mathrm{HT}}$, then combine them through the Delta method to obtain $\hat V_{\theta,\mathrm{HT}}$.
    \item Report the finite-population-calibrated interval $\hat\theta
    \pm
    z_{1-\alpha/2}\sqrt{\hat V_{\theta,\mathrm{HT}}}.$
    \item Report $\hat\gamma_\theta=V_{\theta,\mathrm{int}}/\hat V_{\theta,\mathrm{HT}}$ as a diagnostic of the
    variance inflation of the usual Wald interval relative to the finite-population conditional variance.
\end{enumerate}
This calibration is used only for coverage evaluation against the fixed finite-population benchmark $\theta_N$. It does not replace the usual superpopulation interval when the inferential target is $\theta(P)$.

\section{Additional experimental results and details}
\subsection{Additional details: CheXpert real data example} \label{app:CHX_details}

We evaluate our procedure on the CheXpert chest radiograph dataset \cite{CheXpert, CheXpert-paper}, using a pretrained chest X-ray classifier to provide black-box predictions. This experiment is intended to mimic a deployment setting where a large collection of covariates and images is available, a pretrained model can produce predictions for all observations, but true labels are treated as costly and can only be queried for a subset of individuals.

\subsubsection{Dataset and preprocessing} \label{app:chex-data}

We use publicly available CheXpert training data, restricting the data to one study per patient to avoid repeated observations from the same individual, which violates the independence assumption. Demographic variables used include age, sex, frontal/lateral view, AP/PA view, and the original CheXpert pathology labels. CheXpert outcome labels take values in $\{0, 1, -1, \texttt{NA}\}$, where $-1$ denotes uncertainty and $\texttt{NA}$ denotes an unmentioned or undocumented finding. For our analyses, we restrict to observations with definitive ground-truth labels $Y \in \{0,1\}$ for the target pathology. If the \texttt{No finding} field is labeled as 1.0 in the original dataset, then all pathology labels are set to 0, even if they are currently \texttt{NA} or $\texttt{NaN}$. We set patient cardiomegaly status as the $Y$ label. The remaining demographic variables as well as X-ray images are used for generating predictions $f(X)$, while age binarized to above and below 40 years old defines the two disjoint subgroups of patients.
 
\subsubsection{Pretrained predictions} \label{app:chex-predictions}

We obtain predictions from a pretrained DenseNet-based chest radiograph classifier, using the \texttt{densenet121-res224-chex} weights from \texttt{torchxrayvision}, an open-source Python library accessible on GitHub. Images are center-cropped and resized to $224 \times 224$ using the preprocessing pipeline from \texttt{torchxrayvision}. The model outputs a predicted probability for each pathology, which we treat as a fixed black-box prediction function $$\hat{p}_i = f(X_i) = \hat{\P}(Y_i = 1 \mid x_i).$$ 

Summary statistics for the predictive model, including age-stratified performance, are given in Figure~\ref{fig:CHX_stats}. Although the age-stratified AUROCs are similar, indicating comparable ranking performance across age groups, the fixed threshold of 0.5 yields different operating statistics. In particular, the younger group has higher recall but lower accuracy, suggesting that the model is more sensitive but less specific in this stratum. 

\begin{figure}
    \centering
    \includegraphics[width=\linewidth]{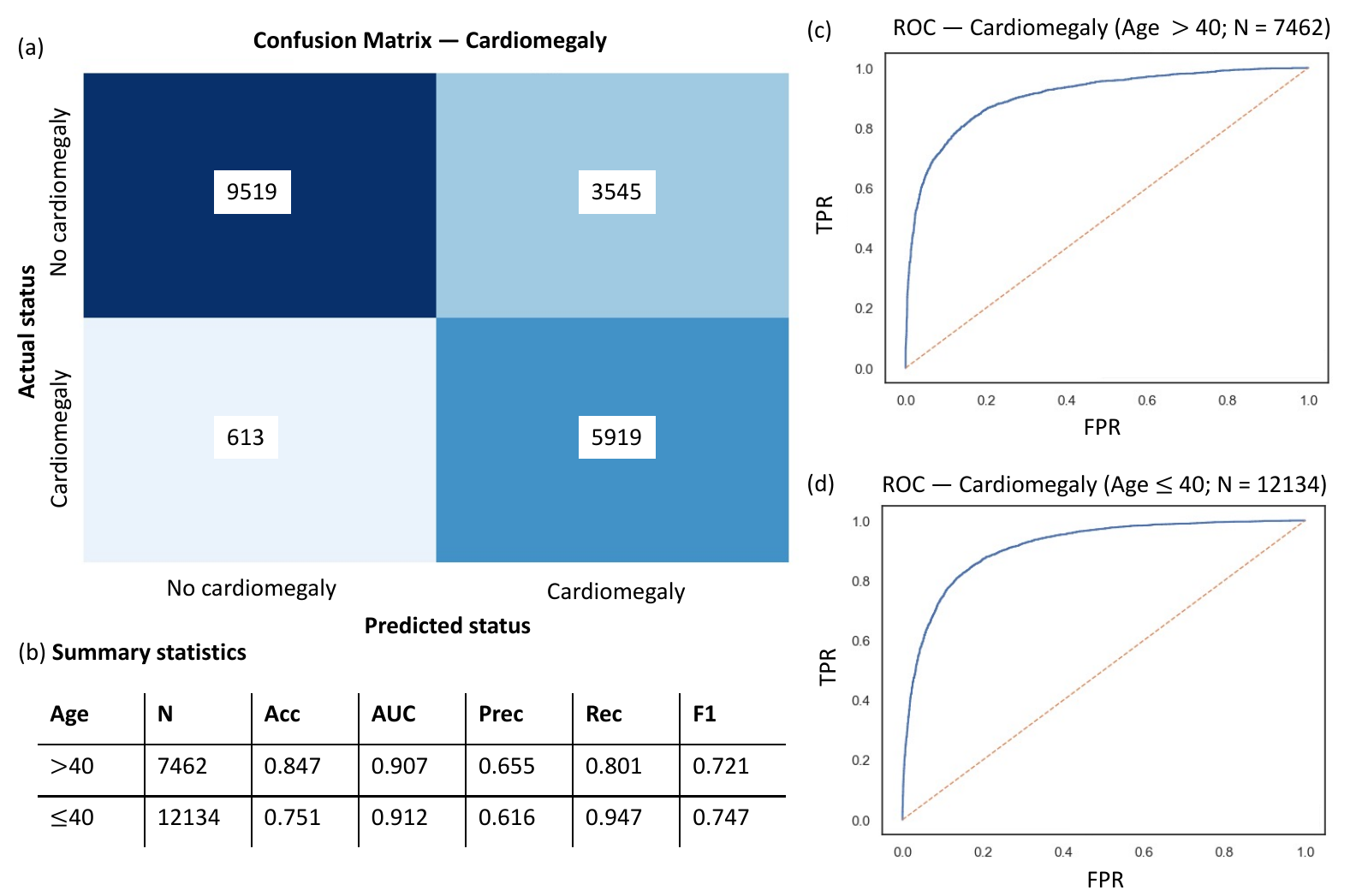}
    \caption{Predictive performance of the CheXpert-pretrained model for cardiomegaly: (a) shows the overall confusion matrix using a probability threshold of 0.5, evaluated on patients with definitive cardiomegaly labels; (b) reports age-stratified summary statistics, including accuracy, AUROC, precision, recall, and F1 score; (c) and (d) show ROC curves for patients with age $> 40$ and age $\leq 40$, respectively. The model exhibits strong discriminative performance in both age strata, with AUROC $0.907$ for patients with age $> 40$ and AUROC $0.912$ for patients with age $\leq 40$. However, the thresholded characteristics differ across age groups: recall is higher among patients with age $\leq 40$, while accuracy and precision are higher among patients with age $> 40$.}
    \label{fig:CHX_stats}
\end{figure}

\subsection{Additional details: breast cancer subtype real data example} \label{app:BRCA_details}

This appendix describes the construction of the The Cancer Genome Atlas Breast Cancer (TCGA-BRCA) whole-slide imaging (WSI) dataset used in the empirical study. The pathology slides serve as high-dimensional covariates $X_i$ from which we construct an auxiliary prediction $\hat{p}_i = f(X_i)$ of triple-negative breast cancer (TNBC) status. These predictions are used to guide label allocation. The true TNBC labels used for estimation and evaluation are clinically derived receptor-status labels.

\subsubsection{Acquisition and organization of whole-slide images}
\label{app:tcga-brca-wsi-acquisition}

Whole-slide pathology images for the TCGA-BRCA cohort were obtained from the Genomic Data Commons (GDC) using a reproducible, programmatic download pipeline. Because the raw imaging files are large and the full cohort requires substantial disk storage, all download and validation steps were carried out on a computing cluster. 

We first constructed an imaging manifest by querying the GDC \texttt{/files} endpoint for files associated with the TCGA-BRCA project. The query was restricted to slide images by filtering for
\texttt{cases.project.project\_id = TCGA-BRCA}, 
\texttt{files.data\_category = Biospecimen}, and \texttt{files.data\_type = Slide Image}.
The response was requested in tab-separated format and included the fields \texttt{file\_id}, \texttt{file\_name}, and \texttt{cases.case\_id}. This produced a file-level manifest in which each row corresponds to one candidate whole-slide image.

The file manifest uses GDC internal case universally unique identifiers (UUIDs), whereas the clinical data are indexed by TCGA submitter barcodes. To reconcile these identifiers, we separately queried the GDC \texttt{/cases} endpoint to obtain the mapping between \texttt{case\_id} and \texttt{submitter\_id}. The latter corresponds to the TCGA patient barcode. We merged the imaging manifest with this case-level
mapping to obtain a patient-linked imaging table. Since some patients have multiple associated slide images, we constructed a patient-level imaging cohort by retaining a single representative slide per patient, chosen as the first available slide in the manifest by default.

The imaging files are whole-slide pathology images, stored in SVS format. These files are high-resolution, multi-scale digital pathology images and are substantially larger than the accompanying clinical metadata. The downloaded data were stored on Sherlock in a directory structure indexed by GDC file identifiers, with a separate metadata table linking each file to its TCGA patient barcode and clinical annotations.

Due to intermittent connection failures in bulk transfer, the final acquisition procedure used a manifest-driven manual loop over GDC file identifiers. After each download pass, the set of successfully downloaded files was compared against the original manifest. Files that were missing, incomplete, or failed integrity checks were identified and re-queued for subsequent download attempts. This sweep-and-retry procedure was repeated until no additional files could be recovered. A small number of files remained unavailable after repeated attempts, due either to persistent connection failures or file-level unavailability from the remote source. These files were excluded from the final imaging cohort. The final dataset was obtained by merging the validated imaging table with the cleaned clinical cohort using TCGA patient barcodes.

\subsubsection{Generation of whole-slide subtype predictions}
\label{app:wsi-prediction-generation}

Each patient-level observation consists of a whole-slide image (WSI), which is a high-resolution digital pathology image. Since a WSI cannot be processed directly at full resolution by a standard convolutional neural network, we first converted each slide into a small collection of informative image patches. For each slide, we identified tissue-containing regions using a low-resolution thumbnail and excluded background regions with little or no tissue content. We then sampled a fixed number of RGB patches from the tissue mask. In our experiments, we used $K_i = 16$ patches per patient, each of size $224 \times 224$. This produced a set of image patches per patient $i$: $X_i = \{X_{i1}, \dots X_{iK_i}\}$, where $X_{ij}$ denotes patch $j$ from the patient.

We use a weakly supervised multiple-instance learning architecture to map the set of patches to a patient-level subtype prediction. Each patch was passed through a convolutional feature extractor based on a pretrained ResNet18 model, with the final fully connected layer removed. Let $z_{ij} = f(X_{ij}) \in \R^{512}$ denote the feature vector extracted from patch $j$ of patient $i$. The feature extractor was used to obtain patch-level representations, while the patient-level label was used only at the slide level.

To obtain a single slide-level representation, we aggregated patch-level feaure vectors by mean pooling: $$\bar{z}_i = \frac{1}{K_i}\sum_{j=1}^{K_i} z_{ij}.$$ This yields a 512-dimensional feature vector per patient. The aggregated feature vector was then passed through a trainable linear classifier, $\ell_i = w^\top \bar{z_i} + b,$ and transformed into a predicted probability of TNBC by a logistic link: $\hat{p}_i = \sigma(\ell_i) = (1 + \exp(-\ell_i))^{-1}.$ Recall $Y_i = 1$ denotes TNBC positive while $Y_i = 0$ denotes no TNBC.

The subtype label is observed at the patient level, while the model works on sampled patches from the patient's corresponding WSI. Thus, the patches are used to construct a slide-level representation unsupervised, while the label is only used for patient-level prediction. 

As noted in Section~\ref{sec:BRCA}, the classifier achieves accuracy around 80\% on the test set, but performs poorly in the minority group, resulting in a large fraction of false negatives. This is summarized in more detail in Figure~\ref{fig:BRCA_stats}.

The TNBC prediction task is substantially more challenging than the CheXpert cardiomegaly task because TNBC is rare (both in the broader population and in the data set), comprising approximately $17\%$ of observations. As a result, overall accuracy is a misleading performance measure: a majority classifier that always predicts non-TNBC already achieves approximately 83\% accuracy. The CNN's thresholded predictions behave similarly. At the default threshold, the model predicts very few TNBC cases, yielding high specificity but very low sensitivity. In the test set, the model correctly identifies only 8 of 101 TNBC cases, corresponding to a recall of approximately 0.08.

The subgroup results reveal an additional concern: the model's errors are not uniform across racial groups. The true TNBC rate is substantially higher among Black patients than among White patients, but the CNN predicts very low TNBC rates in both groups. Thus, the thresholded model substantially attenuates the
observed subgroup difference in TNBC prevalence. This is especially problematic for our estimand, an odds ratio comparing TNBC rates between Black and White patients, because using thresholded CNN predictions as surrogate labels would bias the estimated subgroup contrast toward the null.

The ROC curves suggest that the continuous CNN scores retain some ranking information, but this signal also differs by subgroup. The overall AUROC is approximately $0.70$, while the subgroup AUROCs are approximately 0.62 for Black
patients and 0.72 for White patients. Thus, the model discriminates TNBC cases less well in the Black subgroup, precisely the subgroup with higher TNBC prevalence. This disparity highlights that the model is not merely poorly calibrated at the default threshold; its predictive performance also varies across the groups relevant to the odds-ratio analysis. However, despite the poor performance of the predictive model on all standard metrics, we still see an increase in ESS by using \method~over using labels only, uniform labeling, and proportional-to-uncertainty sampling. This is promising, as it separates good results from a performative model/label predicting mechanism. We can view this as a pedagogical example to see how \method~performs with poor $f(X) = \hat{Y}$.

\begin{figure}
    \centering
    \includegraphics[width=\linewidth]{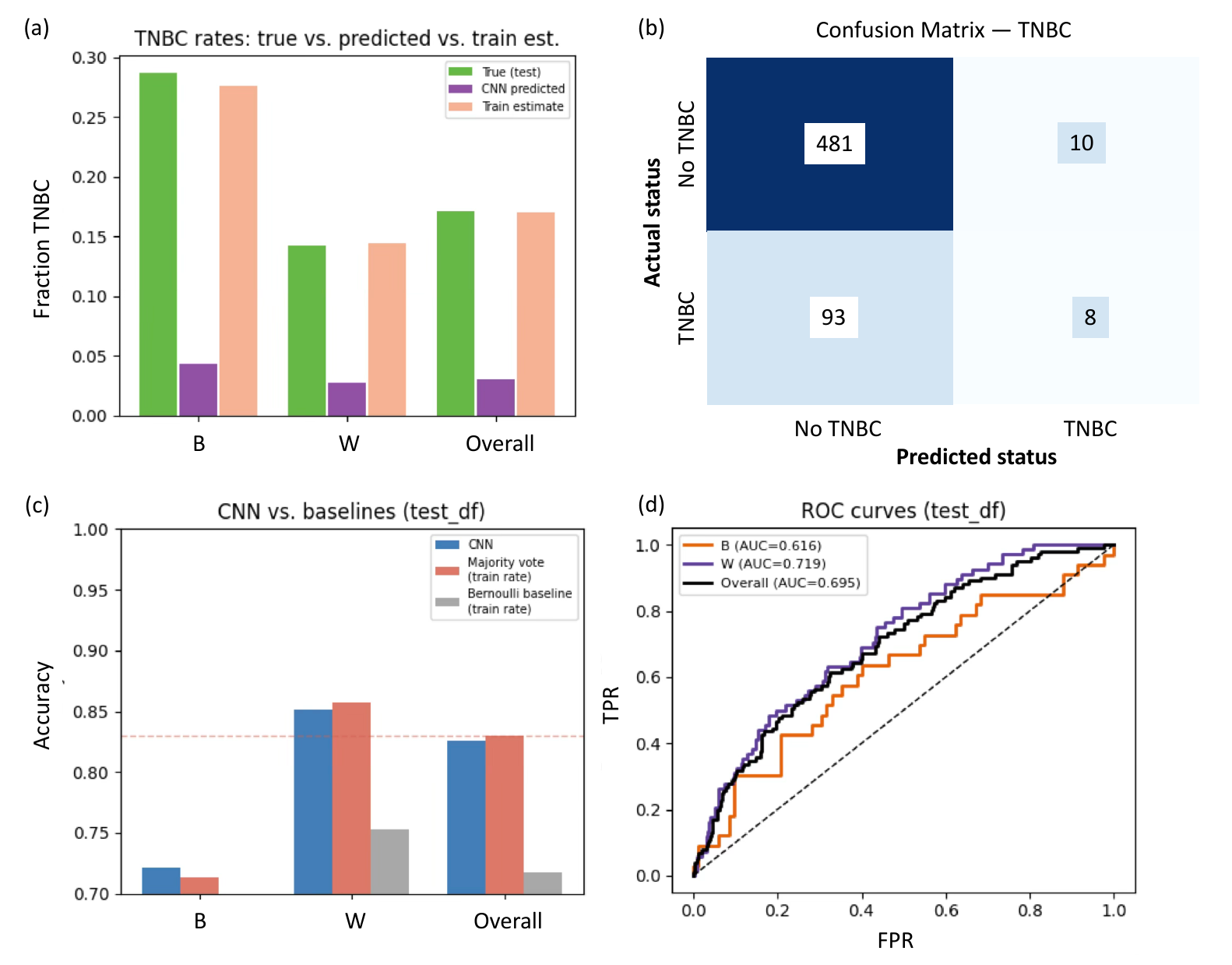}
    \caption{Predictive performance of the CNN for TNBC classification: Panel (a) compares the true TNBC rate, the CNN-predicted TNBC rate, and the training-set prevalence estimate overall and by racial group. The CNN substantially underpredicts TNBC, reflecting the difficulty of the imbalanced classification task. Panel (b) shows the confusion matrix at the default threshold, illustrating that most true TNBC cases are classified as non-TNBC. Panel (c) compares CNN accuracy to simple baselines based on the training-set prevalence. Because TNBC comprises only approximately $17\%$ of the data, a majority-class classifier already achieves high accuracy. Panel (d) shows ROC curves overall and by racial group.}
    \label{fig:BRCA_stats}
\end{figure}

\subsection{Additional details: media stance on global warming example}\label{app:stance}

We use the media-stance data analyzed in \cite{llms}, which builds on the stance dataset of \cite{stance}. The dataset consists of 2300 news headlines annotated according to whether they agree, are neutral toward, or disagree with the statement that global warming is a serious concern. Following \cite{llms}, we focus on the odds ratio measuring whether headlines containing affirming devices are more likely to agree with this statement.

Let $Y_i \in \{0,1\}$ indicate whether headline $i$ is annotated as agreeing with the global-warming statement, and let $Z_i \in \{0,1\}$ indicate the presence of an affirming device, such as words or phrases including ``expert,'' ``proven,'' or ``renowned.'' The target parameter is the log odds ratio
$$\theta = \log\frac{\mu_1/(1-\mu_1)}{\mu_0/(1-\mu_0)}, \qquad \mu_z:=\P(Y=1\mid Z=z),$$ which measures whether headlines containing affirming devices are more likely to agree with the global-warming statement. The human annotations are treated as the ground-truth labels. The predictive information is constructed following the two-stage LLM annotation procedure in \cite{llms}. In the first stage, an LLM is prompted zero-shot to annotate the stance of each headline, $\hat{Y}_i$. In the second stage, the LLM is asked to provide a verbalized confidence score $C_i$ for its annotation. The resulting LLM annotation is used as the surrogate prediction for the human stance label, while the confidence score is used to estimate the probability that the LLM annotation is incorrect. In \cite{llms}, GPT-4o (gpt-4o-2024-05-13 version) and GPT-3.5-turbo (gpt-3.5-turbo-0125 version) were used for data annotation via zero-shot prompting (see Table 2 of the paper for prompt texts in both stages, and Appendix A.3 for more specific details such as parameter values). 

Then, following \cite{llms}, we fit an XGBoost model to estimate the error probability $\widehat{\mathrm{err}}(C_i)$ as a function
of the verbalized confidence score $C_i$. This estimated error probability serves as the uncertainty score for label acquisition: headlines with larger predicted annotation error are treated as more informative to query. Thus, in this experiment, \method~uses the LLM-generated stance annotation $\hat{Y}_i$ as the black-box prediction for $Y_i$, and uses the XGBoost-estimated error probability $\widehat{\mathrm{err}}(C_i)$ as the input uncertainty score for constructing the labeling policy.

Coverage of intervals constructed using the varying methods, both uncorrected and adjusted are given in Figure~\ref{fig:stance-ESS-combined}.

\begin{figure}
    \centering
    \includegraphics[width = \linewidth]{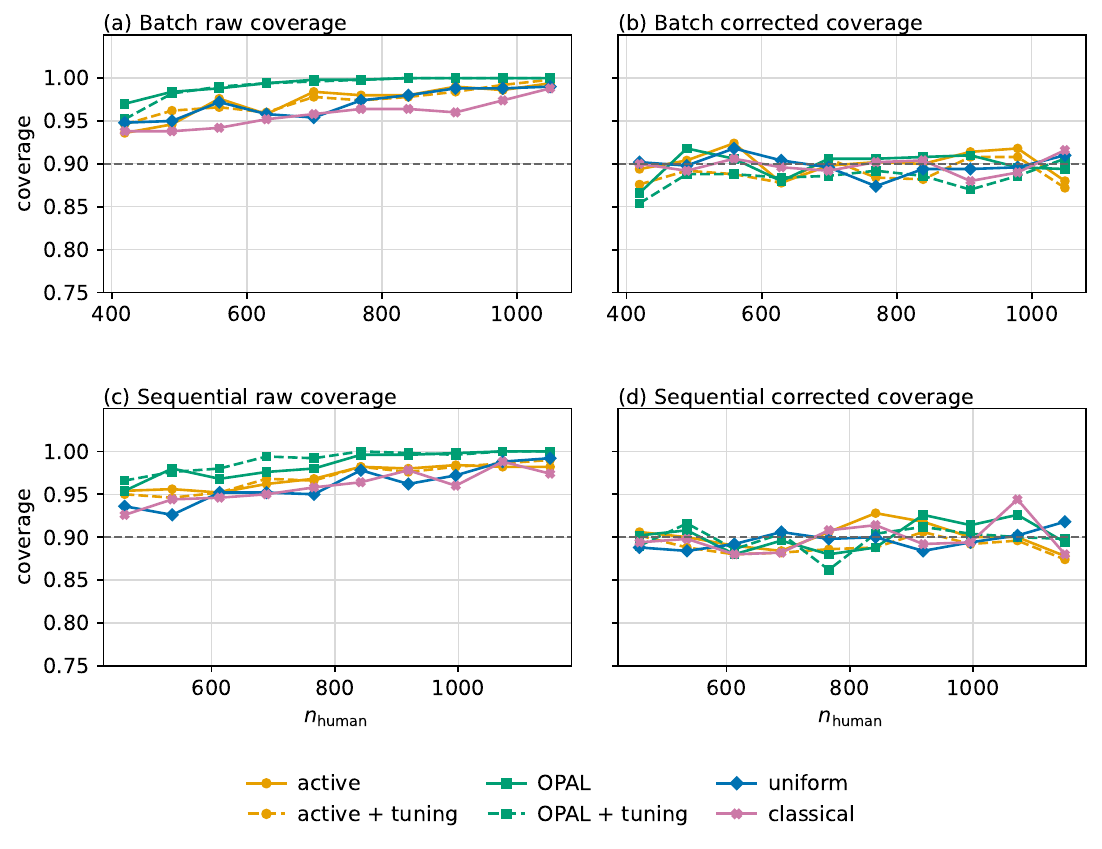}
    \caption{\textbf{Odds ratio estimation of global warming stance with affirming devices} Coverage of each method under (a) batch sampling, uncorrected and (b) batch, adjusted for finite population, (c) sequential sampling, uncorrected, and (d) sequential, adjusted for finite population. We perform 500 trials per method at each budget level (20-50\%), and average over these trials in the reported results.}
    \label{fig:stance-coverage-combined}
\end{figure}

Figure~\ref{fig:stance-coverage-combined} shows variance across endpoints of intervals generated from the 500 different Monte Carlo trials are shown below on a log scale.

\begin{figure}[h]
    \centering
    \includegraphics[width = \linewidth]{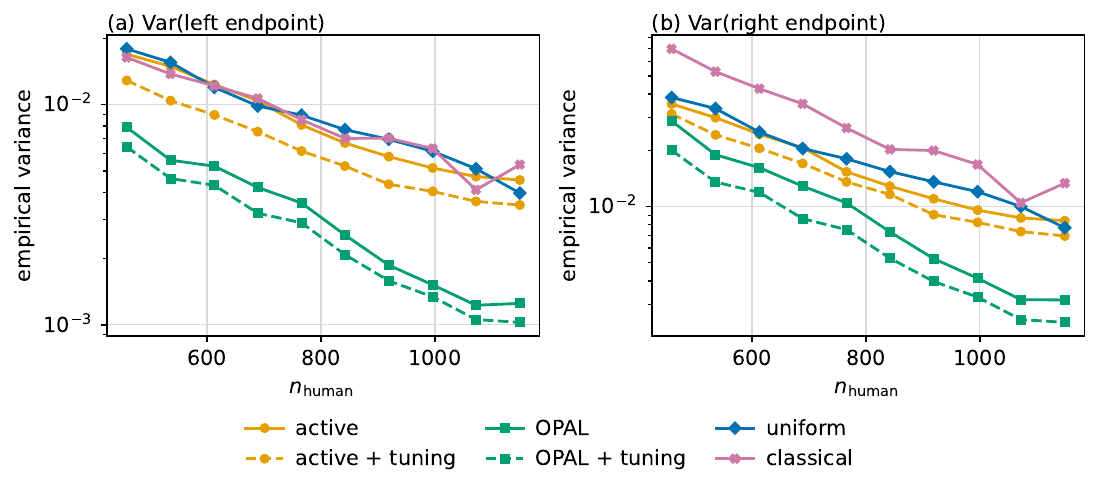}
    \caption{\textbf{Stability of odds ratio estimation of global warming stance in the media in the presence of affirming devices vs. no affirming devices.} Variability of estimates, interval widths, left and right endpoints over 500 Monte Carlo trials. The budget given on the x-axis (denoted by the number of labels acquired, $n_{\text{human}}$) ranges from 10\% to 20\% of the total unlabeled observations.}
    \label{fig:stance_var}
\end{figure}

\begin{figure}
    \centering
    \includegraphics[width=0.8\linewidth]{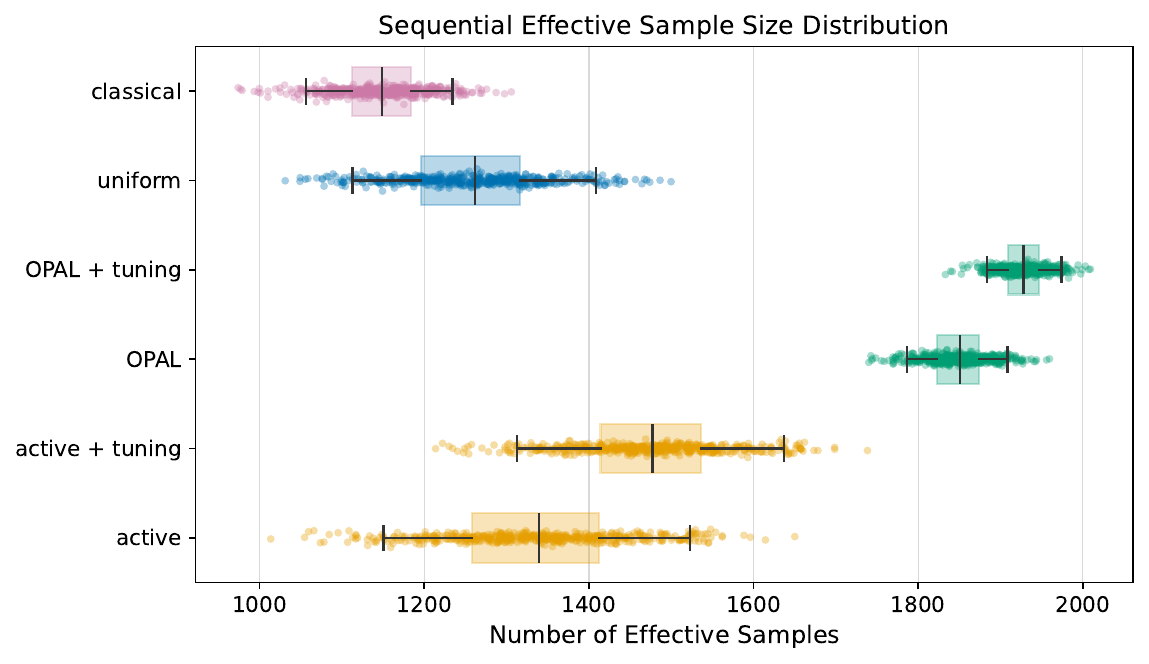}
    \caption{\textbf{Odds ratio estimation of global warming stance with affirming devices: sequential sampling.} Distribution of effective sample size (x-axis) of each method across 500 iterations}
    \label{fig:stance-sequential-dist}
\end{figure}

\subsection{Additional details: Alphafold-derived predictors for intrinsic disorder prediction example}\label{app:alphafold}

For the proteomics experiment, we use the AlphaFold-based example from \cite{bludau, ppi, active}. The scientific question is whether phosphorylation, a functional protein property, is associated with being part of an intrinsically disordered region
(IDR), a structural property which can only be obtained from knowledge about the protein structure measured accurately only via expensive experimental techniques. In our notation, $Y_i \in \{0,1\}$ indicates whether protein $i$ is part of an IDR, and $Z_i \in \{0, 1\}$ indicates whether the protein is phosphorylated. The target is again the log odds ratio $$\theta = \log \frac{\mu_1/(1 - \mu_1)}{\mu_0/(1 - \mu_0)}, \qquad \mu_z := \P(Y = 1 \mid Z = z),$$ where $z = 1$ corresponds to phosphorylated proteins and $z = 0$ to non-phosphorylated proteins.

Following \cite{active}, the gold-standard IDR measurements are treated as expensive labels, while the post-processed
AlphaFold outputs are treated as predictive information available prior to label acquisition. The predictive model $f$ is a logistic regression model trained to predict the IDR indicator $Y_i$ from the AlphaFold-derived features. In our implementation, we also include the phosphorylation indicator $Z_i$ as a covariate in the predictive model, since phosphorylation may itself be informative about intrinsic disorder. The fitted probabilities $\hat{p}_i = f(X_i)$ are for uncertainty scores $u(X_i) = \hat{p}_i (1 - \hat{p}_i)$.

\section{Additional simulation results and details}

\subsection{Mixing with uniform}
\label{app:mixing}

We also considered a simple data-driven way to mix an adaptive sampling rule with uniform sampling.  Let
$\pi_m(x)$ denote the sampling probability produced by method $m$, where $m$ is either active uncertainty sampling or OPAL, and let $\pi_{\mathrm{unif}}(x)$ denote the uniform policy with the same expected labeling budget.  For a grid of values $\lambda \in [0,1]$, we formed
$$\pi_{m,\lambda}(x) = \lambda \pi_m(x) + (1-\lambda)\pi_{\mathrm{unif}}(x),$$
with probabilities renormalized to preserve the target expected number of labels.  Under this convention, $\lambda=1$ gives the original adaptive method and $\lambda=0$ gives pure uniform sampling. For each budget, we selected $\lambda$ using a labeled pilot sample.  Specifically, for each candidate value of $\lambda$, we repeatedly simulated Bernoulli label acquisition from $\pi_{m,\lambda}$ on the pilot sample and computed the prediction-assisted log odds-ratio estimate.  The tuning criterion was the empirical Monte Carlo variance of this log-scale estimate; under the usual Wald approximation, this is the same variance term that determines confidence interval width up to a fixed multiplicative constant.  We then chose
$$\lambda_m^* = \arg\min_{\lambda} \widehat{\mathrm{Var}}_{\mathrm{MC}}\{\hat\theta_{\lambda,m}\},$$
where the variance is taken over the simulated label-acquisition step.  Equivalently, if $\tau=1-\lambda$ denotes the uniform-mixing weight, this procedure selects the amount of uniform exploration that minimizes the estimated Monte Carlo variability of the downstream estimator.

Figure~\ref{fig:sim-lambdatune} shows the selected $\lambda_m^*$ values across budgets.  OPAL typically selects a value close to one, meaning that the estimated optimum is the unmixed or nearly unmixed OPAL policy.  In the 500-repetition run, the median selected OPAL weight was $0.95$, and $17$ of the $20$ budgets selected $\lambda^* \geq 0.9$.  In contrast, active uncertainty sampling selected substantially smaller and more budget-dependent values: the median selected active weight was $0.25$, so the corresponding median uniform weight was $0.75$.  Thus the active rule often benefited from considerable uniform exploration, whereas OPAL did not show a systematic need for additional mixing. This pattern is consistent with the intuition behind OPAL.  Pure uncertainty sampling can concentrate labels in a narrow region of the score distribution and assign very small probabilities elsewhere.  Even if this is useful for prediction, it can be unstable for inverse-probability-weighted inference because observations with small sampling probabilities can create large weights.  Mixing with uniform sampling raises these lower-tail probabilities and reduces the chance that a small part of the population is effectively unsampled.  OPAL, by contrast, directly optimizes an inference-oriented variance criterion, so its policy already balances concentration against weight stability.

The tuning rule should be viewed as a diagnostic.  It requires a labeled pilot sample, and with a small or unrepresentative pilot sample the estimated Monte Carlo variance can be noisy and may not transfer to the final inference population.  Uniform mixing can also be inefficient when the variance-optimal policy is genuinely highly nonuniform, because part of the labeling budget is spent uniformly on observations with low influence.  These caveats are most relevant for active uncertainty sampling, for which explicit exploration can be important; in this experiment, OPAL was already stable without substantial uniform mixing.

\begin{figure}[t]
    \centering
    \includegraphics[width=0.55\linewidth]{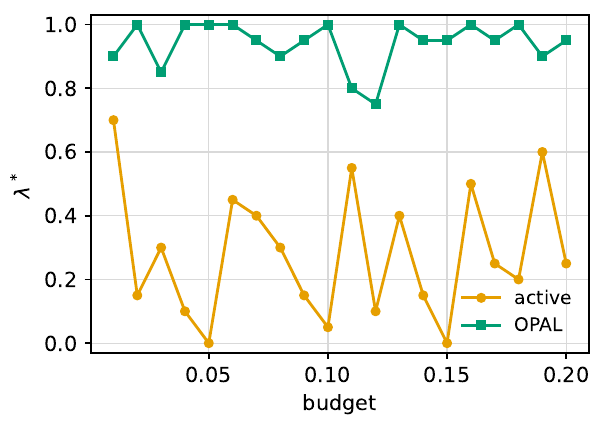}
    \caption{Optimal mixing weight with uniform sampling in the odds-ratio simulation.  For each budget (displayed as a fraction of all unlabeled data points) and each base method, $\lambda^*$ is selected by minimizing the empirical Monte Carlo variance of the prediction-assisted log odds-ratio estimate on the labeled pilot sample.  The convention is $\lambda=1$ for the unmixed method and $\lambda=0$ for pure uniform sampling.}
    \label{fig:sim-lambdatune}
\end{figure}

\subsection{Odds ratio simulation} \label{}

As discussed in Section~\ref{sec:sim_odds}, ESS is a ratio metric and therefore depends on the variance scale of the classical labeled-only baseline. Figure~\ref{fig:sim_0.95_ESS} illustrates the same phenomenon under more severe subgroup imbalance, with $95\%$ of observations in group 1 and only $5\%$ in group 0. Because classical labeling allocates labels roughly in proportion to group size, it spends relatively few labels on the smaller subgroup, which can dominate the uncertainty of the odds-ratio estimator. \method~adapts to these subgroup-specific variance contributions, so the ratio of classical variance to \method variance can become especially large. Thus, the larger ESS values in this figure should be read as reflecting both \method's adaptive allocation and the greater inefficiency of the classical baseline under severe imbalance.

\begin{figure}[th]
    \centering
    \includegraphics[width=0.95\linewidth]{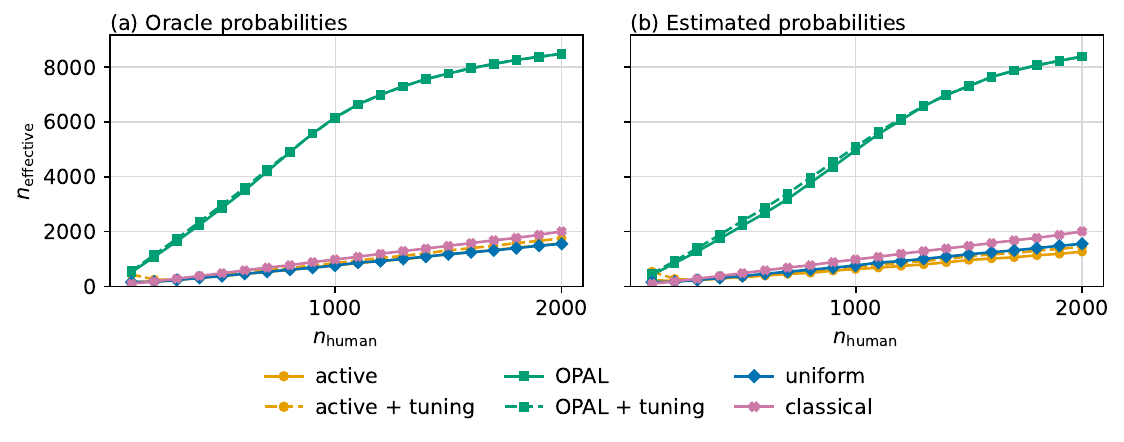}
    \caption{Unbalanced group sizes: group 1 comprises 95\% of population (a) oracle uncertainties used for finding optimal labeling probabilities; (b) estimated uncertainties used. Both plots show effective sample size for the methods we consider in our comparison.} 
    \label{fig:sim_0.95_ESS}
\end{figure}

\subsection{Kendall's tau simulation implementation}
In Section~\ref{app:EIF}, we derive the binary-outcome projection and resulting conditional variance weight $$c_i = c(X_i) = p(X_i)\{1-p(X_i)\}\left(1 - 2[p_0F_0(X_i) + p_1F_1(X_i)]\right)^2.$$
We use this to construct the sampling policies and corresponding one-step estimator.

We generate covariates $X$, a latent score $S^*$, an observed score $S = S^* + \text{noise}$, and labels $Y \sim \Bern(\mu(X))$. A small subset is used to fit the nuisance regression $\hat{\mu}(X) = \P(Y = 1 \mid X)$ by logistic regression with isotonic calibration to adjust the predicted probabilities. The one-step estimator, given labeling indicator $\xi_i$ and sampling probability $\pi_i$ is 
\begin{equation}\label{eq:kendall_tau_1step}
    \hat{\tau} = \frac{1}{n} \sum_i \left[\frac{\xi_i}{\pi_i}(Y_i - \hat{\mu}_i) \hat{A}_i + \hat{\mu}_i\hat{A}_i - \hat{B}_i\right],
\end{equation}
where $\hat{A}_i$ and $\hat{B}_i$ are empirical quantities based on :
\begin{itemize}
    \item $A_i = \tilde\psi(S_i, 1) - \tilde\psi(S_i, 0) = 2F_S(S_i) - 1$, the projection difference, where $F_S(s) = \P(S \leq s) := p_0F_0(s) + p_1F_1(s)$, so $$\hat{A}_i = 2\widehat F_S(S_i) - 1 = 2\frac{\operatorname{rank}(S_i)}{n}-1.$$
    \item $B_i = -\tilde\pi(S_i, 0) = p_1\{2F_1(S_i) - 1\}$, so $$\hat{B}_i = \frac{1}{n}\sum_{j: S_j < S_i} \hat{\mu}_j - \frac{1}{n}\sum_{j: S_j > S_i} \hat{\mu}_j.$$
\end{itemize}
Since $Y_i = 1 \implies Y_iA_i - B_i = A_i - B_i = \tilde{\psi}(S_i, 1)$ while $Y_i = 0 \implies Y_iA_i - B_i = -B_i = \tilde{\psi}(S_i, 0)$, then $$\tilde{\psi}(S_i, Y_i) = Y_iA_i - B_i \implies \E[\tilde\psi(S_i, Y_i) \mid S_i] = \E[Y_i \mid S_i]A_i - B_i.$$
We then get the form of the one-step estimator in Eq.~\eqref{eq:kendall_tau_1step} due to the plug-in conditional projection $$\E[\tilde \psi(S_i, Y_i) \mid S_i] = \mu_iA_i - B_i \implies \widehat{\E}[\tilde \psi(S_i, Y_i) \mid S_i] = \hat{\mu}_i\hat{A}_i -\hat{B}_i,$$ and residual direction $$\zeta_i = \tilde{\psi}(S_i, Y_i) - \E[\tilde \psi(S_i, Y_i) \mid S_i] = (Y_i - \mu_i)A_i,$$ which can be estimated by $\hat{\zeta}_i = (Y_i - \hat{\mu}_i)\hat{A}_i$. 

The variance is estimated by the empirical variance of the estimated one-step contributions. The empirical plug-in conditional variance terms are $$\hat{c}_i = \hat{A}_i^2\hat{\mu}_i(1 - \hat{\mu}_i).$$ 
\end{document}